\documentclass[lettersize,journal]{IEEEtran}
\usepackage{amsmath,amssymb,amsfonts}
\usepackage{algorithmic}
\usepackage{algorithm}
\usepackage{array}
\usepackage{textcomp}
\usepackage{stfloats}
\usepackage{url}
\usepackage{verbatim}
\usepackage{graphicx}
\usepackage{cite}
\usepackage{textcomp}
\usepackage{amsthm}
\usepackage{xcolor}
\usepackage{units}
\usepackage{balance}
\usepackage{acro}
\DeclareAcronym{5g}{
short=5G,
long= fifth generation,
}
\DeclareAcronym{ofdm}{
short=OFDM,
long= orthogonal frequency-division multiplexing,
}
\DeclareAcronym{bs}{
short=BS,
long= base station,
}
\DeclareAcronym{ue}{
short=UE,
long= user equipment,
}
\DeclareAcronym{ris}{
short=RIS,
long= reconfigurable intelligent surface,
}
\DeclareAcronym{siso}{
short=SISO,
long= single-input-single-output,
}
\DeclareAcronym{simo}{
short=SIMO,
long= single-input-multiple-output,
}
\DeclareAcronym{gcs}{
short=GCS,
long= global coordinate system,
}
\DeclareAcronym{lcs}{
short=LCS,
long= local coordinate system,
}
\DeclareAcronym{awgn}{
short=AWGN,
long= additive white Gaussian noise,
}
\DeclareAcronym{iid}{
short=i.i.d.,
long= independent and identically distributed,
}
\DeclareAcronym{snr}{
short=SNR,
long= signal-to-noise ratio,
}
\DeclareAcronym{aoa}{
short=AOA,
long= angle-of-arrival,
}
\DeclareAcronym{aod}{
short=AOD,
long= angle-of-departure,
}
\DeclareAcronym{los}{
short=LOS,
long= line-of-sight,
}
\DeclareAcronym{nlos}{
short=NLOS,
long= non-line-of-sight,
}
\DeclareAcronym{3d}{
short=3D,
long= three-dimensional,
}
\DeclareAcronym{2d}{
short=2D,
long= two-dimensional,
}
\DeclareAcronym{rfc}{
short=RFC,
long= radio-frequency chain,
}
\DeclareAcronym{jrcup}{
short=JrCUP,
long= joint RIS calibration and user positioning,
}
\DeclareAcronym{esprit}{
short=ESPRIT,
long= Estimation of Signal Parameters via Rotational Invariance Techniques,
}
\DeclareAcronym{cp}{
short=CP,
long= canonical polyadic,
}
\DeclareAcronym{ls}{
short=LS,
long= least-squares,
}
\DeclareAcronym{ula}{
short=ULA,
long= uniform linear array,
}
\DeclareAcronym{uav}{
short=UAV,
long= unmanned aerial vehicles,
}
\DeclareAcronym{ml}{
short=ML,
long= maximum likelihood,
}
\DeclareAcronym{crlb}{
short=CRLB,
long= Cram\'er-Rao lower bound,
}
\DeclareAcronym{fim}{
short=FIM,
long= Fisher information matrix,
}
\DeclareAcronym{rmse}{
short=RMSE,
long= root mean square error,
}
\DeclareAcronym{mmwave}{
short=mmWave,
long= millimeter wave,
}
\DeclareAcronym{thz}{
short=THz,
long= terahertz,
}
\DeclareAcronym{upa}{
short=UPA,
long= uniform planar array,
}
\DeclareAcronym{sp}{
short=SP,
long= scattering point,
}

\usepackage{tabu,longtable}
\usepackage{color}
\newtheorem{proposition}{Proposition}

\newtheorem{remark}{Remark}
\usepackage{makecell}
\usepackage[font=footnotesize]{caption}
\usepackage[font=footnotesize]{subcaption}
\usepackage{tikz} % for tikz
\usepackage[utf8]{inputenc}
\usepackage{pgfplots} 
\usepackage{pgfgantt}
\usepackage{pdflscape}
\usepackage{bm}
\pgfplotsset{compat=newest}
\usetikzlibrary{plotmarks}
\usetikzlibrary{arrows.meta}
\usepgfplotslibrary{patchplots}
\usepackage{grffile}

\begin{document}

% \bstctlcite{IEEEexample:BSTcontrol}

\title{
  JrCUP: Joint RIS Calibration and User Positioning for 6G Wireless Systems
}

\author{Pinjun~Zheng,~\IEEEmembership{Student Member,~IEEE},
Hui~Chen,~\IEEEmembership{Member,~IEEE}, 
Tarig~Ballal,~\IEEEmembership{Member,~IEEE}, 
Mikko~Valkama,~\IEEEmembership{Fellow,~IEEE},
Henk~Wymeersch,~\IEEEmembership{Fellow,~IEEE}, 
and~Tareq~Y.~Al-Naffouri,~\IEEEmembership{Senior Member,~IEEE}

\thanks{Pinjun~Zheng, Tarig~Ballal, and  Tareq~Y.~Al-Naffouri are with the Division of Computer, Electrical and Mathematical Science \& Engineering, King Abdullah University of Science and Technology (KAUST), Thuwal, 23955-6900, KSA. (Email: \{pinjun.zheng; tarig.ahmed; tareq.alnaffouri\}@kaust.edu.sa). Hui Chen and Henk Wymeersch are with the Department of Electrical Engineering, Chalmers University of Technology, 41296 Gothenburg, Sweden (Email: \{hui.chen; henkw\}@chalmers.se). Mikko Valkama is with the Department of Electrical Engineering, Tampere University, Finland (Email: mikko.valkama@tuni.fi).}
\thanks{This publication is based upon the work supported in part by the King Abdullah University of Science and Technology (KAUST) Office of Sponsored Research (OSR) under Award No. ORA-CRG2021-4695, in part by the EU H2020 RISE-6G project under grant 101017011, and in part by the Finnish Funding Agency for Innovation under the 6G-ISAC project.}
}
\maketitle

\begin{abstract}
  Reconfigurable intelligent surface (RIS)-assisted localization has attracted extensive attention as it can enable and enhance localization services in extreme scenarios.
  However, most existing works treat RISs as anchors with known positions and orientations, which is not realistic in applications with mobile or uncalibrated RISs.
  This work considers the \textit{joint RIS calibration and user positioning} (JrCUP) problem with an active RIS.
  We propose a novel two-stage method to solve the considered JrCUP problem. The first stage comprises a tensor-estimation of signal parameters via rotational invariance techniques (tensor-ESPRIT), followed by a channel parameters refinement using least-squares.
  In the second stage, a two-dimensional search algorithm is proposed to estimate the three-dimensional user and RIS positions, one-dimensional RIS orientation, and clock bias from the estimated channel parameters.
  The Cram\'er-Rao lower bounds of the channel parameters and localization parameters are derived to verify the effectiveness of the proposed tensor-ESPRIT-based algorithms. 
  In addition, simulation results reveal that the active RIS can significantly improve the localization performance compared to the passive case under the same system power supply in practical regions.
  Moreover, we observe the presence of blind areas with limited JrCUP localization performance, which can be mitigated by either leveraging more prior information or deploying extra base stations.
\end{abstract}

  
\section{Introduction}

% \blue{[Introduction to RIS]}
\Ac{ris} is an emerging technology for 5G/6G and beyond, which consists of an array of reflecting elements and offers distinctive characteristics that make the propagation environment controllable~\cite{Liu2021Reconfigurable,Pan2022Overview}.
With the flexibility to reshape wireless channels and the low cost of deployment, \ac{ris} has become one of the key enablers for future \ac{mmwave} and \ac{thz} band communication systems~\cite{He2022Beyond,Sarieddeen2021Overview,Chen2022Tutorial}.
Over the years, different types of \acp{ris} have been proposed and widely studied, including passive \acp{ris}, active \acp{ris}, hybrid \acp{ris}, and simultaneous transmitting and receiving (STAR) \acp{ris}~\cite{Zhang2022Active,Schroeder2022Two,Mu2022Simultaneously}.

% \blue{[RIS-aided localization works overview]}
Apart from the benefits to communication, the inclusion of \acp{ris} also opens new opportunities for radio localization.
Radio location via wireless networks has been regarded as an indispensable function in advanced 5G/6G systems, which plays an increasingly important role in various applications. For example, radio and network localization systems in scenarios where the global positioning system (GPS) is insufficient or not available are well-studied~\cite{Peral2018Survey,Zheng20235G,fang20203}.
A significant advantage that \acp{ris} offer in radio localization is reducing the number of \acp{bs} required to perform localization. A \ac{ris} can not only act as a new synchronized location reference but also provide additional geometric measurements thanks to its high angular resolution.
With \acp{ris} being introduced properly, it is possible to perform localization using a single \ac{bs}~\cite{Bjornson2022Reconfigurable} or even without any \acp{bs}~\cite{chen2023riss,Kim2022RIS,Chen2023Multi} at all.
Recent studies have shown the potential of \ac{ris}-assisted localization systems in various scenarios, e.g., localization under user mobility~\cite{Keykhosravi2022RIS}, simultaneous indoor and outdoor localization~\cite{He2022Simultaneous}, received-signal-strength based localization~\cite{Zhang2021MetaLocalization}, etc.
Besides, the position and orientation estimation error bounds for the \ac{ris}-assisted localization are derived in~\cite{Elzanaty2021Reconfigurable}. A few \ac{ris} beamforming design optimization works can be found in~\cite{Fascista2022RIS,Gao2022Wireless,Keskin2022Optimal}. 

% \blue{[RIS calibration/positioning is necessary]}
Although promising results on \ac{ris}-assisted localization are shown in the literature, most of the existing works regard \ac{ris} as an anchor with known position and orientation, which is not realistic in some application scenarios involving  \acp{ris} with calibration errors or mobile \acp{ris}.
As a matter of fact, calibration errors in the \ac{ris} placement and geometric layout are unavoidable in practice, making \ac{ris} calibration a necessity for performing a high-precision localization.
The results in~\cite{Zheng2022Misspecified} reveal that a minor calibration error on \ac{ris} geometry can cause a non-negligible model mismatch and result in performance degradation, especially in a high \ac{snr} scenarios.
In this context, Bayesian analysis for a \ac{ris}-aided localization problem under \ac{ris} position and orientation offsets is carried out in~\cite{Emenonye2022RIS}, which discussed the possibility of correcting the \ac{ris} position and orientation under the near-field and far-field models.
Recently, some efforts have also been directed towards the integration of \acp{uav} and \ac{ris}s~\cite{Samir2021Optimizing,Ge2022Reconfigurable}, which extends the application scenarios where \ac{ris}s locations vary with time. In such scenarios, the localization of the \ac{ris} becomes a newly introduced issue that needs to be tackled. As a consequence, localizing the \ac{ris} itself while localizing the user has become an increasingly important problem to solve today.

% \blue{[Review on existing JrCUP works and their shortcomings]}
The \ac{jrcup} problem refers to localizing the \ac{ris} and user simultaneously (with or without a priori information about the \ac{ris} position and orientation).
The \ac{3d} \ac{jrcup} localization problem was first formulated in~\cite{Lu2022Joint}, which explored the relationship between the channel parameters and localization unknowns, with the corresponding Fisher information matrix derived and analyzed.
Nonetheless, the adopted passive RIS in~\cite{Lu2022Joint} limits the localization performance, and the design of an efficient channel estimator for \ac{jrcup} is still missing. 
In~\cite{Ghazalian2022Joint}, a multi-stage solution for the \ac{2d} \ac{jrcup} problem in a hybrid \ac{ris}-assisted system is reported.
However, the hybrid \ac{ris} setup requires an extra central processing unit (CPU) for the receiver and \ac{ris} to share observations, which increases the system complexity.
Recently, active \ac{ris}, which can simultaneously reflect and amplify the incident signals without requiring extra CPU, has shown the potential to provide better localization performance compared to the passive \ac{ris}, as it enhances the reflected signals thus avoiding the overwhelming dominance of the TX-RX \ac{los} channel~\cite{Mylonopoulos2022Active,Mylonopoulos2023Maximum,chen2023riss}. 

In this work, we extend the \ac{3d} \ac{jrcup} problem in~\cite{Lu2022Joint} by developing an efficient channel estimation algorithm and utilizing active \acp{ris} to improve localization performance. The contributions of this work can be summarized as follows:
\begin{itemize}
  \item We formulate the \ac{jrcup} problem in an uplink \ac{simo} scenario with an active \ac{ris}. Motivated by the mobile \ac{ris} with vertical attitude adjustment capability, we define the unknown localization parameters as the 3D position of \ac{ue}, the 3D position and 1D orientation of the \ac{ris}, and the clock bias between the \ac{ue} and \ac{bs}. By performing a localizability analysis, we show that this problem is solvable and that more dimensions of \ac{ris} orientation can be estimated by introducing more \acp{ue}.
  \item We propose a two-stage solution to the \ac{jrcup} problem. In the first stage, we perform a coarse channel estimation via tensor-\ac{esprit} followed by a channel parameter refinement through a  \ac{ls}-based algorithm. In the second stage, a \ac{2d}-search-based algorithm is proposed to estimate the localization parameters.  
  \item The fundamental \acp{crlb} for the \ac{jrcup} problem are derived, which consist of the \acp{crlb} for the estimations of the channel parameters and localization parameters. We show that the received noise in the active \ac{ris}-involved system is colored and unknown which causes the proposed algorithm to fail to reach the theoretical bounds. However, the very minor gaps between the tested \acp{rmse} and the derived \acp{crlb} still show the effectiveness of our algorithms.
  \item Based on the derived \acp{crlb}, we compare the active \ac{ris} and the passive \ac{ris} setups. It is shown that the active \ac{ris} can outperform the passive \ac{ris} within practical power supply regions. The localization performance of the active \ac{ris} setup is improved with the power supply increasing up to a certain level where the performance saturates. Furthermore, we show that blind areas exist in the \ac{jrcup} problem, which can be restrained by leveraging prior geometric information and/or deploying more \acp{bs}.
\end{itemize}

The remainder of this paper is organized as follows.
Section \ref{sec_model} introduces the system model and formulates the \ac{jrcup} problem.
Section \ref{sec_mathpre} reviews the necessary mathematical preliminaries for the tensor decomposition and \ac{esprit} algorithms, based on which a two-stage solution for the \ac{jrcup} problem is proposed in Section~\ref{sec_algos}. 
The \acp{crlb} of the underlying estimation problems are derived in Section~\ref{sec_crlb}.
Simulation results are presented in Section~\ref{sec_sims} followed by the conclusion of this work in Section~\ref{sec_conclusion}. 

Notations: Italic letters denote scalars (e.g., $a$), bold lower-case letters denote vectors (e.g., $\mathbf{a}$), bold upper-case letters denote matrices (e.g., $\mathbf{A}$), and bold calligraphic letters denote tensors (e.g., $\bm{\mathcal{A}}$).
The notations $(\cdot)^\mathsf{T} $, $(\cdot)^\mathsf{*} $, $(\cdot)^\mathsf{H} $, $(\cdot)^{-1}$, $(\cdot)^\dagger$, and $\text{tr}(\cdot)$ are reserved for the transpose, conjugate, conjugate transpose, inverse, Moore-Penrose pseudo-inverse, and the matrix trace operations.
The notation $\otimes$ denotes the Kronecker product, $\circ$  denotes the outer product, and $\odot$ denotes the Hadamard product.
We use $[\mathbf{x}]_i$ to represent the $i$th entry of a vector $\mathbf{x}$, and $[\mathbf{X}]_{i,j}$ to represent the entry in the $i$th row and $j$th column of a matrix $\mathbf{X}$. 
The notations $\mathfrak{R}(\cdot)$ and $\mathfrak{I}(\cdot)$ denote the operations of taking
the real and imaginary parts of a complex quantity, respectively.

%%%%%%%%%%%%%%%%%%%%%%%%%%%%%%%%%%%%%%%%%%%%%%%%%%%%%%%%%%%%%%%%
%%%%%%%%%%%%%%%%%%%%%%%%%%%%%%%%%%%%%%%%%%%%%%%%%%%%%%%%%%%%%%%%
%%%%%%%%%%%%%%%%%%%%%%%%%%%%%%%%%%%%%%%%%%%%%%%%%%%%%%%%%%%%%%%%

\section{System Model}
\label{sec_model}

\subsection{Active \ac{ris}}\label{eq_IIA}

\begin{figure}[t]
  \centering
  \includegraphics[width=3.6in]{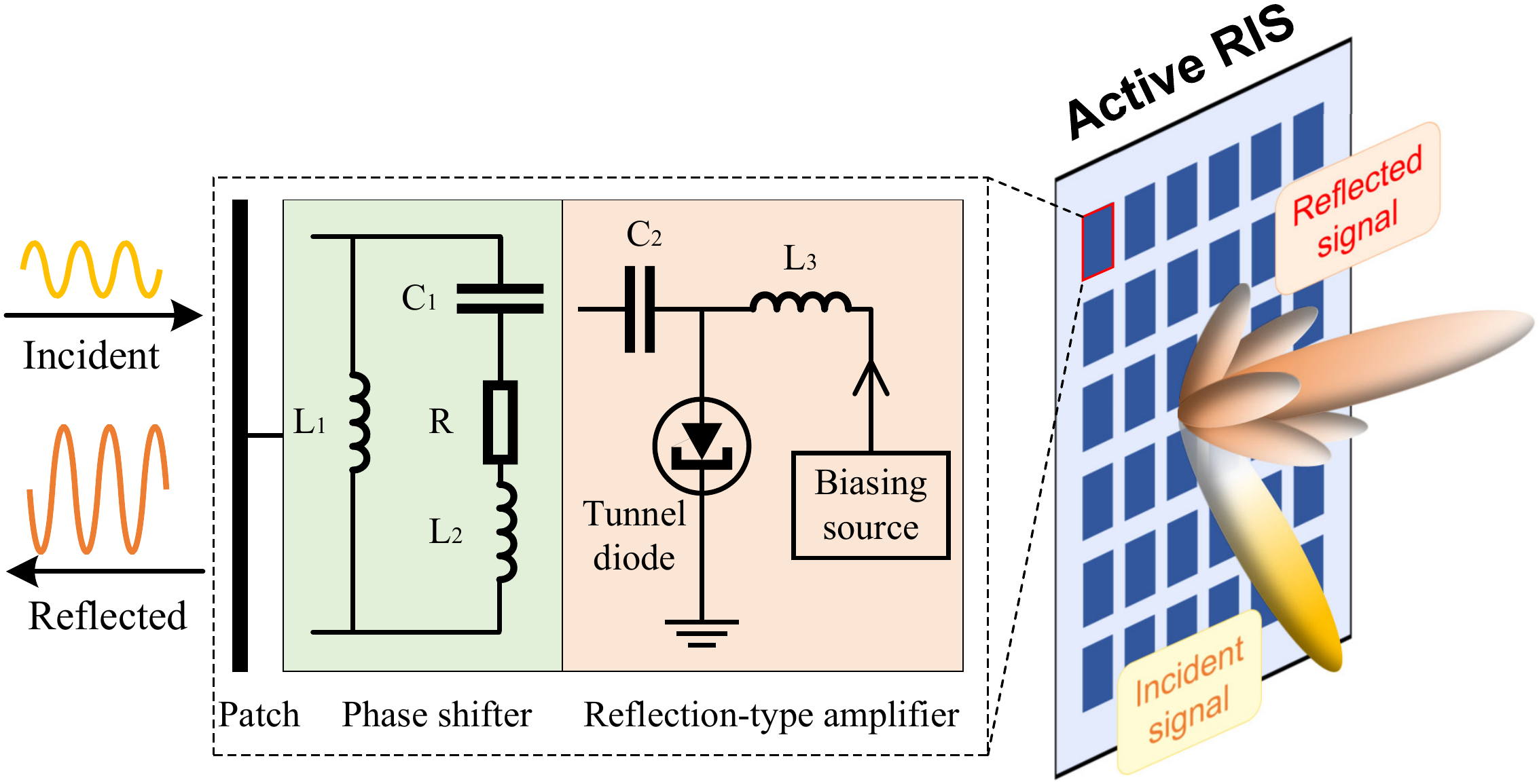}
  \vspace{-1.5em}
  \caption{ 
    Schematic diagram of the hardware architecture of a typical active \ac{ris}. Each \ac{ris} element integrates a phase shifter~\cite{Abeywickrama2020Intelligent} and a reflection-type amplifier~\cite{Long2021Active}.
    }
  \label{fig_activeRIS}
\end{figure}

Active \ac{ris} is an array of active elements that can scatter the incident signals with both amplification and tunable phases. 
Fig.~\ref{fig_activeRIS} presents a schematic diagram of the hardware architecture of a typical active \ac{ris}, where the \ac{ris} element consists of a phase shifter and a reflection-type amplifier~\cite{Zhang2022Active}.
Many circuit implements have been proposed to achieve the function of phase shift and amplification.
For example, the phase shift can be realized through a parallel resonant circuit~\cite{Abeywickrama2020Intelligent}, and the amplification can be realized by a tunnel diode circuit~\cite{Long2021Active}, as depicted in Fig.~\ref{fig_activeRIS}.
Usually, the circuit network in an active \ac{ris} element can be modeled as a two-port component whose characteristics can be described by an S-parameter matrix~\cite{Rao2023Active}.\footnote{The investigation into detailed circuit networks is beyond the scope of this work, so we simply adopt two cascading mathematical operations (i.e., phase shift and amplification) to model the function of the active RIS.} 
By modulating the load impedance of each \ac{ris} element, the desired amplitude and phase change can be obtained (examples can be found in, e.g.,~\cite{Long2021Active,Abeywickrama2020Intelligent}).

Consider an active \ac{ris} of size $N_{\text{R},1}\times N_{\text{R},2}$.
Assuming all the \ac{ris} elements keep the same amplification coefficient, the \ac{ris} profile for a transmission can be denoted as $\bm{\gamma}=p\tilde{\bm{\gamma}}$, where $p>1$ is the amplification coefficient and $\tilde{\bm{\gamma}}\in\mathbb{C}^{N_{\text{R},1}N_{\text{R},2}}$ is the phase-shift vector with unit-modulus entries.
We define the incident signal power \textit{per \ac{ris} element} (assumed identical across \ac{ris} elements) as $P_\text{in}$. Then the required power supply of the active \ac{ris} can be calculated as~\cite{Peng2022Multi}
\begin{equation}\label{eq_PR1}
	P_\text{R} = (p^2-1)N_{\text{R},1}N_{\text{R},2}(P_\text{in} + \sigma_r^2),
\end{equation}
where $\sigma_r^2$ denotes the power of the thermal noise introduced in each active \ac{ris} element.
Note that~\eqref{eq_PR1} is an idealized model, and the actual energy consumption of an active RIS needs to further account for the RIS mutual coupling (MC) effect,\footnote{{The MC effect on RIS refers to coupling that arises between adjacent \ac{ris} elements, which causes a higher energy consumption and lower achievable data rates in wireless communication systems~\cite{Saab2022Optimizing,Di2022Communication}. The usage of active RIS can further accentuate this impact. The method development in this paper is based on a RIS MC-free scenario, while the impact of RIS MC will be examined in Subsection~\ref{sec_MC}.}} electronic circuit energy consumption, energy efficiency, etc.

\subsection{Geometry Model}\label{sec_geomodel}

\begin{figure}[t]
  \centering
  \includegraphics[width=3.3in]{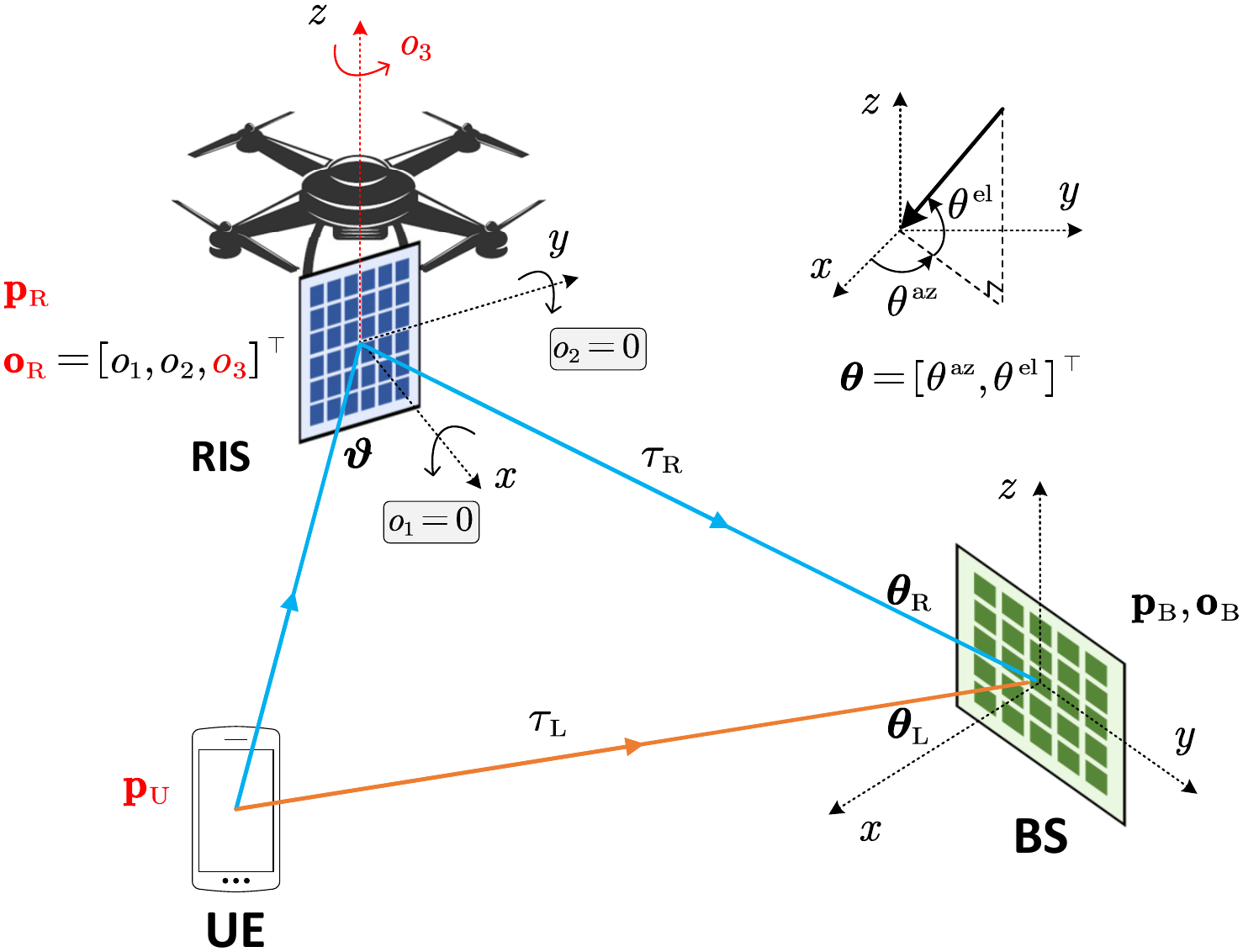}
  \caption{ 
    {Illustration of an uplink \ac{simo} \ac{jrcup} scenario, where the states of the \ac{ue} and \ac{ris} are unknown.
    The \ac{ris} is deployed on a drone with gravity sensors, where the \ac{ris} is always perpendicular to the ground; hence, only one orientation angle ($[\mathbf{o}_\text{R}]_3$) needs to be estimated.} 
    }
  \label{fig_system}
\end{figure}

This subsection describes the geometric relationship among the system devices. Here, we only focus on the \ac{los} components of the \ac{ue}-\ac{bs} and the \ac{ue}-\ac{ris}-\ac{bs} paths which are used for localization, while the \ac{nlos} multipath will be discussed in Subsection~\ref{sec_channelModel}.

We consider an uplink \ac{simo} wireless system consisting of a single-antenna \ac{ue}, an $N_{\text{R},1}\times N_{\text{R},2}$-element active \ac{ris}, and an $N_{\text{B},1}\times N_{\text{B},2}$-element \ac{bs} located at $\mathbf{p}_\text{U}\in\mathbb{R}^3$, $\mathbf{p}_\text{R}\in\mathbb{R}^3$, and $\mathbf{p}_\text{B}\in\mathbb{R}^3$, respectively, as shown in Fig.~\ref{fig_system}.
The orientations of the \ac{ris} and the \ac{bs} are denoted as Euler angles $\mathbf{o}_\text{R}\in\mathbb{R}^3$ and $\mathbf{o}_\text{B}\in\mathbb{R}^3$.
Note that the orientations can also be represented by the rotation matrices
$\mathbf{R}_\text{R}\in\text{SO}(3)$ and $\mathbf{R}_\text{B}\in\text{SO}(3)$, which is constrained in the group of \ac{3d} rotations defined as
$
  \text{SO}(3) = \{\mathbf{R}|\mathbf{R}^\mathsf{T}\mathbf{R}=\mathbf{I}_3,\det(\mathbf{R}) = 1\}.
$
The mapping between a rotation matrix and its Euler angles can be found in~\cite{Chen2022Tutorial}.
In this work, the states of the \ac{ue} and the \ac{ris} are the unknowns to be estimated, while the \ac{bs} is used as the reference point in the considered coordinate system with known states. 
Specifically, we consider a scenario where the \ac{ris} is mobile and has vertical attitude adjustment capability, making a single degree of freedom in the \ac{ris}'s orientation.
For example, the \ac{ris} can be deployed on a drone with gravity sensors/accelerometers~\cite{Mitikiri2019Acceleration,Ge2022Reconfigurable}. Then we can assume the \ac{ris} plane to always be perpendicular to the ground (fixed pitch and yaw angles) and only a 1D orientation (roll angle) in the horizontal plane needs to be estimated, as depicted in Fig.~\ref{fig_system}.
Hence, the unknown geometric parameters consist of the 3D position of \ac{ue}, the 3D position, and the 1D orientation of \ac{ris}.

In the \ac{bs}'s \ac{lcs}, the \ac{aoa} for the \ac{los} path and the \ac{ris} reflection path are denoted as $\bm{\theta}_\text{L}$ and $\bm{\theta}_\text{R}$.
Note that each \ac{aoa} pair consists of an azimuth angle and an elevation angle, i.e., $\bm{\theta}_\text{L}=[\theta_\text{L}^\text{az},\theta_\text{L}^\text{el}]^\mathsf{T} $ and $\bm{\theta}_\text{R}=[\theta_\text{R}^\text{az},\theta_\text{R}^\text{el}]^\mathsf{T} $. 
Those angles are related to the geometric parameters as follows:
\begin{align}\label{eq_geo1}
  \theta_\text{L}^\text{az} &= \arctan 2\left([\mathbf{R}_\text{B}^\mathsf{T} (\mathbf{p}_\text{U}-\mathbf{p}_\text{B})]_2, [\mathbf{R}_\text{B}^\mathsf{T} (\mathbf{p}_\text{U}-\mathbf{p}_\text{B})]_1\right), \\
  \theta_\text{L}^\text{el} &= \arcsin\left({[\mathbf{R}_\text{B}^\mathsf{T} (\mathbf{p}_\text{U}-\mathbf{p}_\text{B})]_3}/{\|\mathbf{p}_\text{U}-\mathbf{p}_\text{B}\|_2}\right).
\end{align}
A similar relationship holds for $\theta_\text{R}^\text{az}$ and $\theta_\text{R}^\text{el}$.
Analogously, in the \ac{ris}'s \ac{lcs}, the \ac{aoa} from the \ac{ue} can be denoted as $\bm{\phi}_\text{A}=[\phi_\text{A}^\text{az},\phi_\text{A}^\text{el}]^\mathsf{T} $ and the \ac{aod} towards the \ac{bs} is $\bm{\phi}_\text{D}=[\phi_\text{D}^\text{az},\phi_\text{D}^\text{el}]^\mathsf{T}$.
Furthermore, we assume an unknown clock bias $\Delta\in\mathbb{R}$ exists between the \ac{ue} and \ac{bs}~\cite{Fascista2022RIS,Chen2022Tutorial}.
Therefore, delays over the \ac{ue}-\ac{bs} and the \ac{ue}-\ac{ris}-\ac{bs} paths are given by
\begin{align}
  \tau_\text{L} &= {\|\mathbf{p}_\text{B}-\mathbf{p}_\text{U}\|_2}/{c} + \Delta, \\
  \tau_\text{R} &= {\|\mathbf{p}_\text{R}-\mathbf{p}_\text{U}\|_2/{c} + \|\mathbf{p}_\text{B}-\mathbf{p}_\text{R}\|_2}/{c} + \Delta,
  \label{eq_tau}
\end{align}
where $c$ is the speed of light.

\subsection{Signal Model}
The \ac{bs} observes and processes the uplink signals received through the \ac{ue}-\ac{bs} channel and the \ac{ue}-\ac{ris}-\ac{bs} channel simultaneously. 
We consider the transmission of $G$ \ac{ofdm} pilot symbols with $K$ subcarriers.
The frequency of the $k$th subcarrier is denoted as $f_k= f_c + \frac{(2k-1-K)}{2}\Delta_f, k=1,\dots,K$, where $f_c$ is the carrier frequency, $\Delta_f=B/K$ is the subcarrier spacing, and $B$ is the bandwidth.
With reflection-type amplifiers supported by a power supply, the active \ac{ris} profile at the $g$th transmission is denoted by $\bm{\gamma}_g=[pe^{j\beta_{1,g}},\dots,pe^{j\beta_{N_{\text{R},1}N_{\text{R},2},g}}]^\mathsf{T} \in\mathbb{C}^{N_{\text{R},1}N_{\text{R},2}}$, $g=1, \dots, G$, where $p>1$ denotes the amplification factor and $\beta_{n,g}$ denotes the phase shift of the $n$th \ac{ris} element at the $g$th transmission.
Again, here we assume that each active \ac{ris} element keeps the same amplification factor. The reflection matrix of the active \ac{ris} can be defined as $\bm{\Gamma}_g=\text{diag}(\bm{\gamma}_g)$.
Suppose that all the \ac{bs} antennas are connected to an $N_1\times N_2$ \ac{rfc} array. We assume $N_1>1$ and $N_2>1$.
The received baseband signal for the $g$th transmission and the $k$th subcarrier, $\mathbf{y}_{g,k}\in\mathbb{C}^{N_1N_2\times 1}$, can be expressed as
\begin{equation}\label{eq_y}
  \mathbf{y}_{g,k} = \mathbf{W}^\mathsf{H} \left( \mathbf{h}_\text{L}^kx_{g,k} + \mathbf{H}_{\text{R},2}^k\bm{\Gamma}_g\mathbf{h}_{\text{R},1}^kx_{g,k} + \mathbf{H}_{\text{R},2}^k\bm{\Gamma}_g\mathbf{n}_r + \mathbf{n}_0 \right),
\end{equation}
where $x_{g,k}\in\mathbb{C}$ is the transmitted signal with average transmission power $|x_{g,k}|^2={P_\text{T}}$, $\mathbf{h}_\text{L}^k$ is the \ac{ue}-\ac{bs} channel vector, $\mathbf{h}_{\text{R},1}^k$ is the \ac{ue}-\ac{ris} channel vector, $\mathbf{H}_{\text{R},2}^k$ is the \ac{ris}-\ac{bs} channel matrix, $\mathbf{n}_r\sim\mathcal{CN}(\mathbf{0},\sigma_r^2\mathbf{I})$ denotes the thermal noise introduced in the active \ac{ris}, $\mathbf{n}_0\sim\mathcal{CN}(\mathbf{0},\sigma_0^2\mathbf{I})$ denotes the thermal noise at the receiver, and $\mathbf{W}\in\mathbb{C}^{N_{\text{B},1}N_{\text{B},2}\times N_1N_2}$ is the combiner matrix which will be specified in Subsection~\ref{sec_corase}.
Note that model~\eqref{eq_y} is reduced to the passive \ac{ris} case when $p=1$ and $\sigma_r=0$.

According to \eqref{eq_y}, we have the total received noise for the $g$th transmission and the $k$th subcarrier as
\begin{equation}\label{eq_noisenew}
  \mathbf{n}_{g,k} = \mathbf{W}^\mathsf{H}(\mathbf{H}_{\text{R},2}^k\bm{\Gamma}_g\mathbf{n}_r + \mathbf{n}_0).
\end{equation}
Here we notice the received noise $\mathbf{n}_{g,k}$ is colored and related to the unknown \ac{ris}-\ac{bs} channel $\mathbf{H}_{\text{R},2}^k$.
The statistics of $\mathbf{n}_{g,k}$ will be derived in Subsection~\ref{sec_CRLB1}.
The channel model~\eqref{eq_y} can be rewritten as 
\begin{equation}\label{eq_y_new}
  \mathbf{y}_{g,k} = \mathbf{W}^\mathsf{H} \mathbf{h}_{g,k}x_{g,k} + \mathbf{n}_{g,k},
\end{equation}
where $\mathbf{h}_{g,k} = \mathbf{h}_\text{L}^k + \mathbf{H}_{\text{R},2}^k\bm{\Gamma}_g\mathbf{h}_{\text{R},1}^k$.

\subsection{Channel Model}\label{sec_channelModel}

By introducing the \ac{nlos} multipath, the \ac{ue}-\ac{bs} channel $\mathbf{h}_\text{L}^k$, the \ac{ue}-\ac{ris} channel $\mathbf{h}_{\text{R},1}^k$, and the \ac{ris}-\ac{bs} channel $\mathbf{H}_{\text{R},2}^k$ are given by~\cite{Chen2022Tutorial}
\begin{align}
\label{eq_mulpa1}
  \mathbf{h}_\text{L}^k &= \sum_{i=0}^{I_\text{L}}\alpha_\text{L}^ie^{-j2\pi(k-1)\Delta_f\tau_\text{L}^i}\mathbf{a}_\text{B}(\bm{\theta}_\text{L}^i),\\
  \mathbf{h}_{\text{R},1}^k &= \sum_{i=0}^{I_{\text{R},1}} \alpha_{\text{R},1}^i\mathbf{a}_\text{R}(\bm{\phi}_\text{A}^i),
  \label{eq_mulpa2}\\
  \mathbf{H}_{\text{R},2}^k &= \sum_{i=0}^{I_{\text{R},2}} \alpha_{\text{R},2}^ie^{-j2\pi(k-1)\Delta_f\tau_\text{R}^i}\mathbf{a}_\text{B}(\bm{\theta}_\text{R}^i)\mathbf{a}_\text{R}^\mathsf{T} (\bm{\phi}_\text{D}^i),
\label{eq_mulpa3}
\end{align}
where $\mathbf{a}_\text{B}(\bm{\theta})\in\mathbb{C}^{N_{\text{B},1}N_{\text{B},2}}$ and $\mathbf{a}_\text{R}(\bm{\phi})\in\mathbb{C}^{N_{\text{R},1}N_{\text{R},2}}$ denote the array response vectors of the \ac{bs} and the \ac{ris}, and $\alpha_\text{L}^i$, $\alpha_{\text{R},1}^i$ and $\alpha_{\text{R},2}^i$ are the complex channel gains for the \ac{ue}-\ac{bs}, the \ac{ue}-\ac{ris}, and the \ac{ris}-\ac{bs} channels, respectively. Here, $i=0$ represents the \ac{los} channel and the rest are the \ac{nlos} multipath channels reflected by \acp{sp}. The numbers of the \ac{nlos} paths of the \ac{ue}-\ac{bs}, the \ac{ue}-\ac{ris} and the \ac{ris}-\ac{bs} channels are denoted as $I_\text{L}$, $I_{\text{R},1}$, and $I_{\text{R},2}$.

In this work, the multipath components are not used to perform localization, thus, it acts as a negative effect that generates additional noise. Given that, the algorithm development in this paper is based on the multipath-free model; however, the impact of the multipath effect will be evaluated in Section~\ref{sec_sims}.
Ignoring the multipath, the channel models~\eqref{eq_mulpa1}--\eqref{eq_mulpa3} are reduced to
\begin{align}
  \mathbf{h}_\text{L}^k &= \alpha_\text{L}e^{-j2\pi(k-1)\Delta_f\tau_\text{L}}\mathbf{a}_\text{B}(\bm{\theta}_\text{L}),\label{eq_chan1}\\
  \mathbf{h}_{\text{R},1}^k &= \alpha_{\text{R},1}\mathbf{a}_\text{R}(\bm{\phi}_\text{A}),\\
  \mathbf{H}_{\text{R},2}^k &= \alpha_{\text{R},2}e^{-j2\pi(k-1)\Delta_f\tau_\text{R}}\mathbf{a}_\text{B}(\bm{\theta}_\text{R})\mathbf{a}_\text{R}^\mathsf{T} (\bm{\phi}_\text{D}),\label{eq_chan3}
\end{align}
where $\alpha_\text{L}$, $\alpha_{\text{R},1}$ and $\alpha_{\text{R},2}$ are the complex channel gains for the corresponding \ac{los} channels. 
The array response vectors of the \ac{bs} and the \ac{ris} are defined as 
\begin{align}\label{eq_ars}
  \left[\mathbf{a}_\text{B}(\bm{\theta})\right]_i &= e^{j\frac{2\pi f_c}{c}\mathbf{t}(\bm{\theta})^\mathsf{T}\mathbf{p}_{\text{B},i}},\\
  \left[\mathbf{a}_\text{R}(\bm{\phi})\right]_i &= e^{j\frac{2\pi f_c}{c}\mathbf{t}(\bm{\phi})^\mathsf{T}\mathbf{p}_{\text{R},i}},
\end{align}
where $\mathbf{p}_{\text{B},i}$ and $\mathbf{p}_{\text{R},i}$ are respectively the positions of the $i$th element of the \ac{bs} and the \ac{ris} given in their \ac{lcs}, and $\mathbf{t}(\bm{\theta})$ is the direction vector defined as 
$
  \mathbf{t}(\bm{\theta}) \triangleq 
\begin{bmatrix}
  \cos(\theta^{\text{az}})\cos(\theta^{\text{el}}),
  \sin(\theta^{\text{az}})\cos(\theta^{\text{el}}), 
  \sin(\theta^{\text{el}})
\end{bmatrix}^\mathsf{T} .
$
For later derivation, we assume both the antenna arrays in the \ac{bs} and the \ac{ris} to be \acp{upa} and their element spacings are denoted as $d_\text{B}$ and $d_\text{R}$.

We assume that all the \ac{bs} and \ac{ris} elements are deployed on the YOZ plane of their \acp{lcs}, i.e., the $x$-coordinates of these elements' positions are zeros.
Since the first entry of the $\mathbf{p}_{\text{R},i}$ is zero,  we introduce the intermediate \ac{ris}-related angles $\bm{\vartheta}=[\vartheta_2,\vartheta_3]^\mathsf{T} $ as~\cite{Lu2022Joint}
  \begin{align}\label{eq_vartheta2}
    \vartheta_2 &= \sin(\phi_\text{A}^\text{az})\cos(\phi_\text{A}^\text{el}) + \sin(\phi_\text{D}^\text{az})\cos(\phi_\text{D}^\text{el}),\\
    \vartheta_3 &= \sin(\phi_\text{A}^\text{el}) + \sin(\phi_\text{D}^\text{el}).\label{eq_vartheta3}
  \end{align}
  Then, the total \ac{ris} array response for both signal arrival and departure can be represented through $\vartheta_2$ and $\vartheta_3$ as 
$
    \left[\mathbf{a}_\text{R}(\bm{\phi}_\text{A})\odot\mathbf{a}_\text{R}(\bm{\phi}_\text{D})\right]_i = e^{j\frac{2\pi f_c}{c}\left(\mathbf{t}(\bm{\phi}_\text{A})+\mathbf{t}(\bm{\phi}_\text{D})\right)^\mathsf{T}\mathbf{p}_{\text{R},i}}
    = e^{j\frac{2\pi f_c}{c}\left(\vartheta_2[\mathbf{p}_{\text{R},i}]_2 + \vartheta_3[\mathbf{p}_{\text{R},i}]_3\right)  }.
$
Thus, a more compact formulation for the \ac{ue}-\ac{ris}-\ac{bs} channel can be represented as
  \begin{multline}
    \mathbf{H}_{\text{R},2}^k\bm{\Gamma}_g\mathbf{h}_{\text{R},1}^k =\alpha_\text{R}\left[(\mathbf{a}_\text{R}(\bm{\phi}_\text{A})\odot\mathbf{a}_\text{R}(\bm{\phi}_\text{D}))^\mathsf{T}\bm{\gamma}_g\right]\\ \times e^{-j2\pi(k-1)\Delta_f\tau_\text{R}}\mathbf{a}_\text{B}(\bm{\theta}_\text{R}),
  \end{multline}
where $\alpha_\text{R} = \alpha_{\text{R},1}\alpha_{\text{R},2}$.

\subsection{Localizability Analysis and Problem Formulation}\label{sec_localizability}

This work aims to jointly estimate the \ac{ris} position and orientation, the \ac{ue} position, and the clock bias based on the received signals.  
We adopt a two-stage estimation framework consisting of a channel parameters estimation followed by a \ac{ris} and \ac{ue} states estimation from the obtained channel parameters.
Based on~\eqref{eq_chan1}--\eqref{eq_chan3}, we define the vector $\bm{\eta}_\text{ch}$ of \textit{all} unknown channel parameters  and the vector $\bm{\eta}$ that contains only localization-related channel parameters as
\begin{align}
  \bm{\eta}_\text{ch} &\triangleq [\bm{\eta}^\mathsf{T}, \mathfrak{R}(\alpha_\text{L}),\mathfrak{I}(\alpha_\text{L}),\mathfrak{R}(\alpha_\text{R}),\mathfrak{I}(\alpha_\text{R})]^\mathsf{T}\in\mathbb{R}^{12}, \\
  \bm{\eta} &\triangleq [\theta_\text{L}^{\text{az}},\theta_\text{L}^{\text{el}},\theta_\text{R}^{\text{az}},\theta_\text{R}^{\text{el}},\tau_\text{L},\tau_\text{R},\vartheta_2,\vartheta_3]^\mathsf{T}\in\mathbb{R}^8.\label{eq_eta}
\end{align}
Note that the parameters $\alpha_\text{L}$ and $\alpha_\text{R}$ that contribute to $\bm{\eta}_\text{ch}$ are nuisance parameters that will not be used for solving the \ac{jrcup} problem.
As an objective of the general \ac{jrcup} problem, the localization parameter vector that contains the \ac{ris} and \ac{ue} states is defined as
\begin{equation}
    \bm{\xi}_{\text{Loc}} \triangleq [\mathbf{p}_\text{U}^\mathsf{T} ,\mathbf{p}_\text{R}^\mathsf{T} ,\mathbf{o}_\text{R}^\mathsf{T},\Delta]^\mathsf{T}\in\mathbb{R}^{10}.
\end{equation}
The \ac{jrcup} problem refers to using estimates of the channel parameters $\bm{\eta}$ to determine the state vector $\bm{\xi}_{\text{Loc}}$.
However, in an estimation problem, the number of unknowns cannot exceed the number of observations $\bm{\eta}$.
In the general case, it is unlikely to find a unique value of $\bm{\xi}_{\text{Loc}}$ based solely on the observations $\bm{\eta}$.
The restriction of the \ac{ris} orientation as highlighted in Subsection~\ref{sec_geomodel} reduces the unknowns pertaining to the \ac{ris} orientation to a single parameter, making the number of observations equal to the number of unknowns.\footnote{In cases where the full 3D orientation of the \ac{ris} needs to be estimated, multiple \acp{ue} can be utilized to obtain more channel parameters. For example, two \acp{ue} at different locations can provide 16 localization-related channel parameters, and the dimension of the unknowns $\bm{\xi}_{\text{Loc}}$ becomes 14 (with one more \ac{ue} position and clock bias), making the problem solvable.}
Consequently, the localization parameters to be estimated in this work can be redefined as
\begin{equation}
  \bm{\xi} \triangleq [\mathbf{p}_\text{U}^\mathsf{T} ,\mathbf{p}_\text{R}^\mathsf{T} ,o_3,\Delta]^\mathsf{T}\in\mathbb{R}^8,
\end{equation}
where $o_3=[\mathbf{o}_\text{R}]_3$ is the Euler angles of the \ac{ris} orientation around the $Z$-axis; the rest of Euler angles (i.e., $o_1$ and $o_2$) are assumed to be known.

In this work, we focus on developing a solution to the \ac{jrcup} problem with a single \ac{ue} and 1D \ac{ris} orientation.
When the \ac{bs} received the \ac{ofdm} symbols from the \ac{ue} through both \ac{los} and \ac{ris} reflected channels, we first estimate $\bm{\eta}$ based on the received signals $\mathbf{y}_{g,k}, k=1,\dots,K, g=1,\dots,G$. 
Afterwards, $\bm{\xi}$ is estimated based on $\bm{\eta}$.

%%%%%%%%%%%%%%%%%%%%%%%%%%%%%%%%%%%%%%%%%%%%%%%%%%%%%%%%%%%%%%%%
%%%%%%%%%%%%%%%%%%%%%%%%%%%%%%%%%%%%%%%%%%%%%%%%%%%%%%%%%%%%%%%%
%%%%%%%%%%%%%%%%%%%%%%%%%%%%%%%%%%%%%%%%%%%%%%%%%%%%%%%%%%%%%%%%

\section{Mathematical Preliminaries}\label{sec_mathpre}

To make the paper self-contained, we provide a brief review of the \ac{cp} decomposition of tensors and the fundamentals of the \ac{esprit} method.
More details on these topics can be found in~\cite{Cichocki2015Tensor,Sidiropoulos2017Tensor,Haardt2008Higher,Roy1989Esprit}.

\subsection{Tensors \& \ac{cp} Decomposition}
\label{sec_tensor}

A tensor, also known as \emph{multi-way array}, is a generalization of data arrays to three or higher dimensions.
Let $\bm{\mathcal{T} }\in\mathbb{C}^{I_1\times I_2\times\dots\times I_N}$ denote an $N$th-order tensor.
The \emph{order} $N$ indicates the number of dimensions, and each dimension is called a \emph{mode}.

The polyadic decomposition approximates a tensor with a sum of $R$ rank-one tensors.
If the number of rank-one terms $R$ is minimum, the corresponding decomposition is called a \ac{cp} decomposition and the minimum achievable $R$ is referred to as the \emph{rank} of the tensor.
Suppose $\bm{\mathcal{T}}$ is a rank-$R$ tensor, the \ac{cp} decomposition decomposes $\bm{\mathcal{T} }$ as 
\begin{equation}\label{eq_cpd}
  \bm{\mathcal{T}} = \sum_{r=1}^R \lambda_r \mathbf{u}_r^{(1)}\circ\mathbf{u}_r^{(2)}\circ\dots\circ\mathbf{u}_r^{(N)}.
\end{equation}
Here, $[\mathbf{u}^{(n)}_1,\dots,\mathbf{u}^{(n)}_R]=\mathbf{U}^{(n)}\in\mathbb{C}^{I_n\times R}$ is the factor matrix along the $n$th mode. Each matrix $\mathbf{U}^{(n)}$ has orthonormal columns.
A visual representation of \eqref{eq_cpd} in the third-order case ($N=3$) is shown in Fig.~\ref{fig_cpd}.

\begin{figure}[t]
  \centering
  \includegraphics[width=3in]{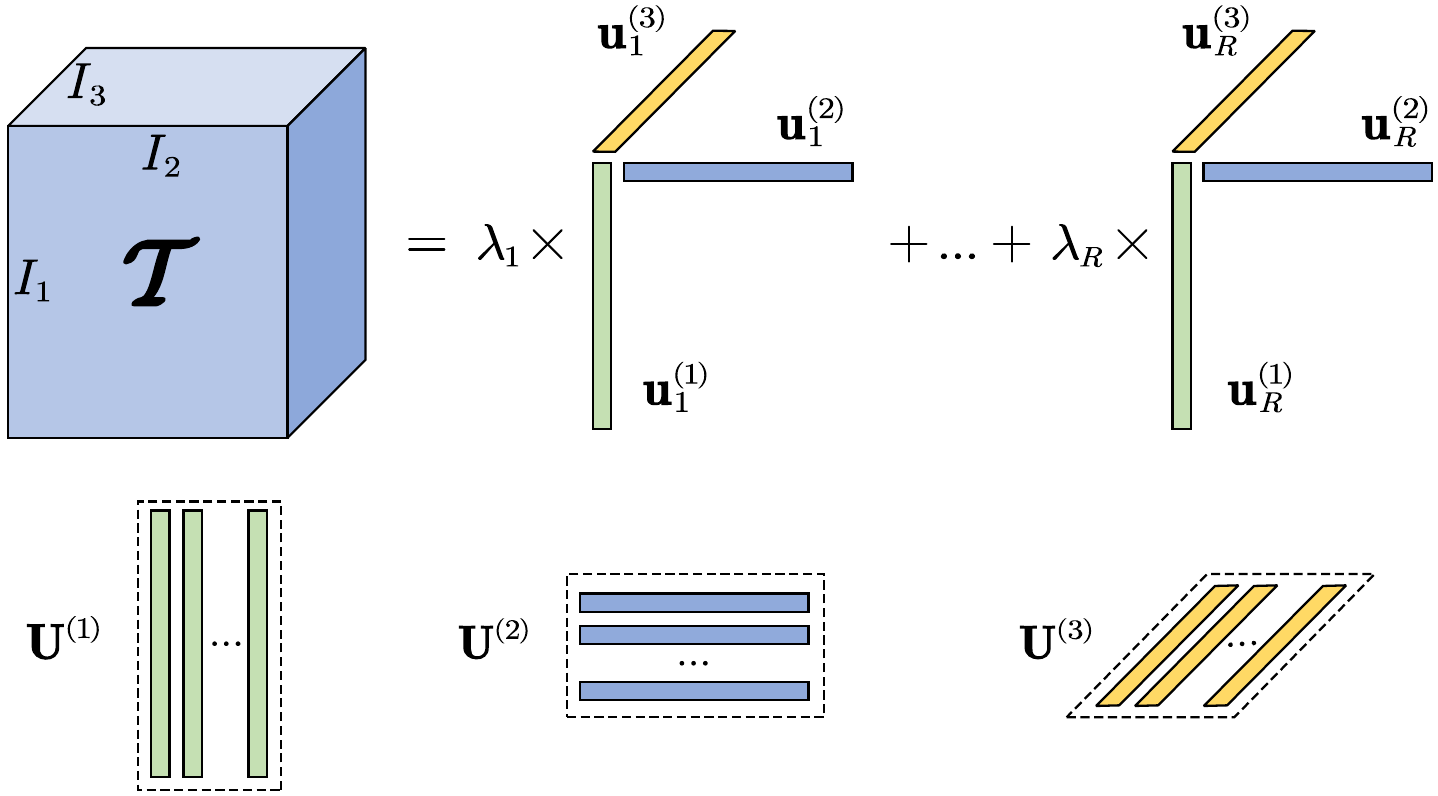}
  \caption{ 
    A \ac{cp} decomposition of a rank-$R$ third-order tensor. 
    }
  \label{fig_cpd}
\end{figure}

\subsection{Tensor-\ac{esprit}}
\ac{esprit} is a search-free signal subspace-based parameter estimation technique for multidimensional harmonic retrieval, which has been widely used in channel estimation~\cite{Zhang2021Gridless}, spectrum sensing~\cite{Jiang2022Joint}, sub-Nyquist sampling~\cite{Fu2020Short}, etc.
This subsection recaps two \ac{esprit} variants, namely, element-space and beamspace tensor \ac{esprit}.

\vspace{1em}
\noindent\textbf{Element-Space Tensor \ac{esprit}:}\
Consider a multidimensional harmonic retrieval problem with an observation tensor $\bm{\mathcal{T}}$ whose elements are given by
\begin{equation}\label{eq_obs}
  [\bm{\mathcal{T}}]_{i_1,i_2,\dots,i_N} = \sum_{r=1}^R\lambda_r\prod_{n=1}^{N}e^{j(i_n-1)\omega_{n,r}}.
\end{equation}
To estimate the unknown angular frequencies $\omega_{n,r}$, we define 
\begin{align}\label{eq_a}
  &\mathbf{a}^{(I_n)}(\omega_{n,r}) = [1,e^{j\omega_{n,r}},\dots,e^{j(I_n-1)\omega_{n,r}}]^\mathsf{T}\in\mathbb{C}^{I_n},\\
  &\mathbf{A}_n = \big[ \mathbf{a}^{(I_n)}(\omega_{n,1}),\dots,\mathbf{a}^{(I_n)}(\omega_{n,R}) \big]\in\mathbb{C}^{I_n\times R}.
\end{align}
Then~\eqref{eq_obs} can be rewritten as 
\begin{equation}\label{eq_m1}
  \bm{\mathcal{T}} = \sum_{r=1}^R \lambda_r \mathbf{a}^{(I_1)}(\omega_{1,r})\circ\mathbf{a}^{(I_2)}(\omega_{2,r})\circ\dots\circ\mathbf{a}^{(I_N)}(\omega_{N,r}).
\end{equation}
We define two selection matrices
$
  \mathbf{J}_{n,1} = \left[\mathbf{I}_{I_n-1},\mathbf{0}_{(I_n-1)\times 1}\right] \in \mathbb{R}^{(I_n-1)\times I_n},\ 
  \mathbf{J}_{n,2} = \left[\mathbf{0}_{(I_n-1)\times 1},\mathbf{I}_{I_n-1}\right] \in \mathbb{R}^{(I_n-1)\times I_n},
$
and let $\bm{\Phi}_n=\text{diag}\left\{[e^{j\omega_{n,1}},\dots,e^{j\omega_{n,R}}]^\mathsf{T}\right\}\in\mathbb{C}^{R\times R}$.
To estimate the angular frequency, ESPRIT relies on the shift-invariance property of $\mathbf{A}_{n}$, which is given by
\begin{equation}\label{eq_sip}
  \mathbf{J}_{n,1}\mathbf{A}_{n}\bm{\Phi}_n = \mathbf{J}_{n,2}\mathbf{A}_{n}.
\end{equation}
Then the associated frequencies $\omega_{n,r},r=1,\dots,R,n=1,\dots,N$, can be estimated as follows. By applying the \ac{cp} decomposition to $\bm{\mathcal{T}}$, we obtain the factor matrices $\mathbf{U}^{(n)}, n=1,\dots,N$.
Since $\mathbf{a}^{(I_n)}(\omega_{n,r})$ lies in the column space of $\mathbf{U}^{(n)}$, we replace $\mathbf{A}_n$ by $\mathbf{A}_n = \mathbf{U}^{(n)}\mathbf{D}_n$, where $\mathbf{D}_n\in\mathbb{C}^{R\times R}$ is a non-singular matrix.
Then~\eqref{eq_sip} becomes
\begin{equation}\label{eq_sip2}
  \mathbf{J}_{n,1}\mathbf{U}^{(n)}\bm{\Theta}_n = \mathbf{J}_{n,2}\mathbf{U}^{(n)},
\end{equation}
where $\bm{\Theta}_n=\mathbf{D}_n\bm{\Phi}_n\mathbf{D}_n^{-1}$.
Based on~\eqref{eq_sip2}, the \ac{ls} estimate of $\bm{\Theta}_n$ reads~\cite{Haardt2008Higher} 
\begin{equation}\label{eq_estTheta}
  \hat{\bm{\Theta}}_n = \big(\mathbf{J}_{n,1}\mathbf{U}^{(n)}\big)^\dagger\mathbf{J}_{n,2}\mathbf{U}^{(n)}.
\end{equation}
Since the $r$th eigenvalue of ${\bm{\Theta}}_n$ is given by $e^{j\omega_{n,r}}$, the angular frequencies $\hat{\omega}_{n,r},r=1,\dots,R$, can be obtained from the phase components of the eigenvalues of $\hat{\bm{\Theta}}_n$.

\vspace{1em}
\noindent\textbf{Beamspace Tensor \ac{esprit}:}\
For beamspace measurements, after the $n$th-mode product of $\bm{\mathcal{T}}$ with a linear transformation matrix~\cite{Kolda2009Tensor}, the model~\eqref{eq_m1} is modified to
\begin{equation}
  \bm{\mathcal{T}}_b = \sum_{r=1}^R \lambda_r \mathbf{b}^{(J_1)}(\omega_{1,r})\circ\mathbf{b}^{(J_2)}(\omega_{2,r})\circ\dots\circ\mathbf{b}^{(J_N)}(\omega_{N,r}),
\end{equation}
where $\mathbf{b}^{(J_n)}(\omega_{n,r}) = \mathbf{T}_n^\mathsf{H}\mathbf{a}^{(I_n)}(\omega_{n,r}) $ with $\mathbf{T}_n\in\mathbb{C}^{I_n\times J_n}$.
The beamspace array manifold becomes $\mathbf{B}_n=\big[ \mathbf{b}^{(J_n)}(\omega_{n,1}),\dots,\mathbf{b}^{(J_n)}(\omega_{n,R}) \big]\in\mathbb{C}^{J_n\times R}$ 
and $\bm{\mathcal{T}}_b\in\mathbb{C}^{J_1\times J_2\times\dots\times J_n}$.

In the beamspace model with the transformation matrix $\mathbf{T}_n$, generally the shift-invariance property~\eqref{eq_sip} does not hold, i.e., $\mathbf{J}_{n,1}\mathbf{B}_{n}\bm{\Phi}_n \neq \mathbf{J}_{n,2}\mathbf{B}_{n}$.
However, if $\mathbf{T}_n$ has a proper structure, the following proposition~\cite{Wen2020Tensor,Jiang2021Beamspace} shows that the lost shift-invariance property can be restored.

\begin{proposition}
  Assume that $\mathbf{T}_n$ has a shift-invariant structure that satisfies
  $
    \mathbf{J}_{n,1}\mathbf{T}_{n} = \mathbf{J}_{n,2}\mathbf{T}_{n}\mathbf{F}_n
  $
  where $\mathbf{F}_n\in\mathbb{C}^{J_n\times J_n}$ is a non-singular matrix and $\mathbf{T}_n^\mathsf{H} = [\mathbf{t}_1,\dots,\mathbf{t}_{I_n}]^\mathsf{T} $.
  If there exists a matrix $\mathbf{Q}_n\in\mathbb{C}^{J_n\times J_n}$ such that 
\begin{equation}
  \begin{cases}
  \mathbf{Q}_n\mathbf{t}_{I_n} = \mathbf{0}_{J_n\times 1},\\
  \mathbf{Q}_n\mathbf{F}_n^\mathsf{H} \mathbf{t}_1 = \mathbf{0}_{J_n\times 1},
  \end{cases}
\end{equation}
then we have 
\begin{equation}\label{eq_rsip}
  \mathbf{Q}_n\mathbf{B}_n\bm{\Phi}_n^\mathsf{H} = \mathbf{Q}_n\mathbf{F}_n^\mathsf{H}\mathbf{B}_n,
\end{equation}
where the shift-invariance property is restored.
\end{proposition}

\begin{proof}
  See Appendix A in~\cite{Wen2018Tensor}.
\end{proof}

The matrices $\mathbf{Q}_n$ and $\mathbf{F}_n$ can be estimated as~\cite{Wen2020Tensor}
\begin{align}\label{eq_Q}
  \hat{\mathbf{Q}}_n &= \mathbf{I}_{J_n} - \mathbf{t}_{I_n}\mathbf{t}_{I_n}^\mathsf{H} - \left(\mathbf{F}_n^\mathsf{H}\mathbf{t}_1 \right)\left(\mathbf{F}_n^\mathsf{H}\mathbf{t}_1 \right)^\mathsf{H},\\
  \hat{\mathbf{F}}_n &= \left(\mathbf{J}_{n,2}\mathbf{T}_{n}\right)^\dagger \mathbf{J}_{n,1}\mathbf{T}_{n}. \label{eq_F}
\end{align}
Then, similarly, after obtaining $\mathbf{U}^{(n)}$ by applying the \ac{cp} decomposition, we can replace $\mathbf{B}_n$ by $\mathbf{B}_n=\mathbf{U}^{(n)}\mathbf{D}_n$, thus~\eqref{eq_rsip} becomes
$
  \mathbf{Q}_n\mathbf{U}^{(n)}\bm{\Theta}_n = \mathbf{Q}_n\mathbf{F}_n^\mathsf{H}\mathbf{U}^{(n)}.
$
Therefore, 
\begin{equation}\label{eq_estTheta2}
  \hat{\bm{\Theta}}_n = \big(\mathbf{Q}_{n}\mathbf{U}^{(n)}\big)^\dagger\mathbf{Q}_{n}\mathbf{F}_n^\mathsf{H} \mathbf{U}^{(n)}.
\end{equation}
Finally, the unknown angular frequencies are obtained from the phase of the eigenvalues of $\hat{\bm{\Theta}}_n$.

%%%%%%%%%%%%%%%%%%%%%%%%%%%%%%%%%%%%%%%%%%%%%%%%%%%%%%%%%%%%%%%%
%%%%%%%%%%%%%%%%%%%%%%%%%%%%%%%%%%%%%%%%%%%%%%%%%%%%%%%%%%%%%%%%
%%%%%%%%%%%%%%%%%%%%%%%%%%%%%%%%%%%%%%%%%%%%%%%%%%%%%%%%%%%%%%%%

\section{Proposed Localization Method}\label{sec_algos}
This section develops a two-stage localization method to solve the \ac{jrcup} problem formulated in Subsection~\ref{sec_localizability} based on the preliminaries in Section~\ref{sec_mathpre}.

\subsection{Coarse Channel Estimation via Tensor-\ac{esprit}}\label{sec_corase}

Given $\mathbf{y}_{g,k}$ in~(10), we can obtain the beamspace channel estimates by multiplying the received signals by the conjugate of the pilot symbols and dividing them by the average power as
\begin{align}\label{eq_geth}
  \hat{\mathbf{h}}_{g,k}^{(\text{b})} = \frac{\mathbf{y}_{g,k}x_{g,k}^*}{P_\text{T}}
  &= {\mathbf{W}^\mathsf{H} \mathbf{h}_{g,k}} + \Delta{\mathbf{h}}_{g,k}^{(\text{b})},
\end{align}
where $\Delta{\mathbf{h}}_{g,k}^{(\text{b})}$ is the noise term. We organize the beamspace channels as 
\begin{align}\label{eq_geth2}
  \hat{\mathbf{H}} &= 
  \begin{bmatrix}
    \hat{\mathbf{h}}_{1,1}^{(\text{b})} & \cdots & \hat{\mathbf{h}}_{G,1}^{(\text{b})}\\
    \vdots & \ddots & \vdots \\
    \hat{\mathbf{h}}_{1,K}^{(\text{b})} & \cdots & \hat{\mathbf{h}}_{G,K}^{(\text{b})}
  \end{bmatrix} = [\hat{\mathbf{h}}_{1}^{(\text{b})},\dots,\hat{\mathbf{h}}_{G}^{(\text{b})}]\in\mathbb{C}^{KN_1N_2\times G}.
\end{align}
%and the noise-free version of beamspace channels are denoted as $\mathbf{H}\in\mathbb{C}^{KN_1N_2\times G}$.
The goal of this subsection is to estimate the localization-related channel parameters in $\bm{\eta}$ defined in~\eqref{eq_eta} given $\hat{\mathbf{H}}$.

To apply the tensor-\ac{esprit} method, we design the combiner matrix $\mathbf{W}$ and the total \ac{ris} profile matrix $\bm{\Gamma}$ to follow the structure
\begin{align}
  \mathbf{W} &= \mathbf{T}_1\otimes\mathbf{T}_2\in\mathbb{C}^{N_{\text{B},1}N_{\text{B},2}\times N_1N_2},
  \label{eq_Wconstraint}\\
  \bm{\Upsilon} &= [\bm{\gamma}_1,\bm{\gamma}_2,\dots,\bm{\gamma}_G]
  = \mathbf{T}_3\otimes\mathbf{T}_4\in\mathbb{C}^{N_{\text{R},1}N_{\text{R},2}\times G},
  \label{eq_Gammaconstraint}
\end{align}
where $\mathbf{T}_1\in\mathbb{C}^{N_{\text{B},1}\times N_1}$, $\mathbf{T}_2\in\mathbb{C}^{N_{\text{B},2}\times N_2}$, $\mathbf{T}_3\in\mathbb{C}^{N_{\text{R},1}\times \sqrt{G}}$, and $\mathbf{T}_4\in\mathbb{C}^{N_{\text{R},2}\times \sqrt{G}}$ with $G$ thus set to a square number.
Note that we utilize a fixed combiner $\mathbf{W}$ over different transmissions while the \ac{ris} profile $\bm{\gamma}_g$ changes with $g$.
Considering the structure of the array response vectors in~\eqref{eq_ars}, and according to~\eqref{eq_a}, we can write
\begin{align}
  \mathbf{a}_\text{B}(\bm{\theta}_\text{L}) &= \mathbf{a}^{(N_{\text{B},1})}(\omega_{\text{L},1}) \otimes \mathbf{a}^{(N_{\text{B},2})}(\omega_{\text{L},2}),\\
  \mathbf{a}_\text{B}(\bm{\theta}_\text{R}) &= \mathbf{a}^{(N_{\text{B},1})}(\omega_{\text{R},1}) \otimes \mathbf{a}^{(N_{\text{B},2})}(\omega_{\text{R},2}),
\end{align}
where 
\begin{align}\label{eq_wL}
  \omega_{\text{L},1} &= ({2\pi f_cd_\text{B}}/{c})\sin(\theta_\text{L}^\text{az})\cos(\theta_\text{L}^\text{el}),\\
  \omega_{\text{L},2} &= ({2\pi f_cd_\text{B}}/{c})\sin(\theta_\text{L}^\text{el}),\\
  \omega_{\text{R},1} &= ({2\pi f_cd_\text{B}}/{c})\sin(\theta_\text{R}^\text{az})\cos(\theta_\text{R}^\text{el}),\\
  \omega_{\text{R},2} &= ({2\pi f_cd_\text{B}}/{c})\sin(\theta_\text{R}^\text{el}). \label{eq_wR}
\end{align}
Furthermore, we define
\begin{equation}\label{eq_wtau}
  \omega_{\tau_\text{L}} = -2\pi\Delta_f\tau_\text{L},\quad \omega_{\tau_\text{R}} = -2\pi\Delta_f\tau_\text{R}.
\end{equation}

\subsubsection{Estimating $\theta_\text{L}^{\text{az}},\theta_\text{L}^{\text{el}},\theta_\text{R}^{\text{az}},\theta_\text{R}^{\text{el}},\tau_\text{L},\tau_\text{R}$}
We first use the sum of the channels over $g=1,\dots,G$, which is 
$
  \hat{\mathbf{h}}^{(\text{b})} = \sum_{g=1}^G \hat{\mathbf{h}}_{g}^{(\text{b})}\in\mathbb{C}^{KN_1N_2}.
$
Suppose $\hat{\mathbf{h}}^{(\text{b})} = \mathbf{h}^{(\text{b})} + \Delta{\mathbf{h}}^{(\text{b})}$ with $\mathbf{h}^{(\text{b})}$ represents the true beamspace channels.
The true beamspace channel $\mathbf{h}^{(\text{b})}$ can be naturally represented by a tensor $\bm{\mathcal{H}}^{{(b)}}\in \mathbb{C}^{K\times N_1\times N_2}$ as 
\begin{multline}\label{eq_bH}
  \bm{\mathcal{H}}^{{(b)}} = G\alpha_\text{L}\mathbf{a}^{(K)}(\omega_{\tau_\text{L}})\circ\mathbf{T}_1^\mathsf{H} \mathbf{a}^{(N_{\text{B},1})}(\omega_{\text{L},1})\circ\mathbf{T}_2^\mathsf{H} \mathbf{a}^{(N_{\text{B},2})}(\omega_{\text{L},2}) \\ + 
  (\sum_{g=1}^G\beta_{\text{R},g})\mathbf{a}^{(K)}(\omega_{\tau_\text{R}})\circ\mathbf{T}_1^\mathsf{H} \mathbf{a}^{(N_{\text{B},1})}(\omega_{\text{R},1})\circ\mathbf{T}_2^\mathsf{H} \mathbf{a}^{(N_{\text{B},2})}(\omega_{\text{R},2}),
\end{multline}
where 
$
  {{\beta}_{\text{R},g}} = \alpha_{\text{R}}\left[(\mathbf{a}_\text{R}(\bm{\phi}_\text{A})\odot\mathbf{a}_\text{R}(\bm{\phi}_\text{D}))^\mathsf{T}\bm{\gamma}_g\right].
$
We can see that $\bm{\mathcal{H}}^{(\text{b})}$ is a rank-two third-order tensor based on the definition in Subsection~\ref{sec_tensor}, as it is the sum of two rank-one tensors (i.e., the UE-BS channel and the UE-RIS-BS channel) and each of them is an outer product of three vectors. The first mode of $\bm{\mathcal{H}}^{(\text{b})}$ lies in the element-space while the rest of modes lie in the beamspace generated by the transformation matrices $\mathbf{T}_1$ and $\mathbf{T}_2$.
Further, we can define $\bm{\Phi}_1=\text{diag}\left\{[e^{j\omega_{\tau_\text{L}}},e^{j\omega_{\tau_\text{R}}}]^\mathsf{T}\right\}$,
$\bm{\Phi}_2=\text{diag}\left\{[e^{j\omega_{\text{L},1}},e^{j\omega_{\text{R},1}}]^\mathsf{T}\right\}$, 
and $\bm{\Phi}_3=\text{diag}\left\{[e^{j\omega_{\text{L},2}},e^{j\omega_{\text{R},2}}]^\mathsf{T}\right\}$.

To estimate the underlying angular frequencies, we first apply the \ac{cp} decomposition to $\hat{\bm{\mathcal{H}}}$ (generated from $\hat{\mathbf{h}}^{(\text{b})}$) and obtain the factor matrices $\mathbf{U}^{(n)},n=1,2,3$.
Then $\bm{\Theta}_1$ can be estimated by~\eqref{eq_estTheta} while $\bm{\Theta}_2$ and $\bm{\Theta}_3$ can be estimated through~\eqref{eq_Q}, \eqref{eq_F}, and \eqref{eq_estTheta2}.
Next, the estimated angular frequencies $\hat{\omega}_{\tau_\text{L}}$, $\hat{\omega}_{\tau_\text{R}}$, $\hat{\omega}_{\text{L},1}$, $\hat{\omega}_{\text{R},1}$, $\hat{\omega}_{\text{L},2}$, and $\hat{\omega}_{\text{R},2}$ can be obtained from the phases of the eigenvalues of $\hat{\bm{\Theta}}_1$, $\hat{\bm{\Theta}}_2$, and $\hat{\bm{\Theta}}_3$,
and the corresponding channel parameters can be recovered based on~\eqref{eq_wL}--\eqref{eq_wtau}.

\subsubsection{Estimating $\vartheta_2$ and $\vartheta_3$}

Based on the estimated parameters, we construct a matrix $\mathbf{R}$ as 
\begin{equation}\notag
  \mathbf{R} = 
  \begin{bmatrix}
    \mathbf{a}^{(K)}(\hat{\omega}_{\tau_\text{L}})\otimes\mathbf{W}^\mathsf{H}\mathbf{a}_\text{B}(\hat{\bm{\theta}}_\text{L}) & \mathbf{a}^{(K)}(\hat{\omega}_{\tau_\text{R}})\otimes\mathbf{W}^\mathsf{H}\mathbf{a}_\text{B}(\hat{\bm{\theta}}_\text{R})
  \end{bmatrix}.
\end{equation}
Then, we have the following system of equations:
\begin{equation}
  \mathbf{R}
  [
    \alpha_\text{L}, 
    \beta_{\text{R},g}
  ]^\mathsf{T}  = \mathbf{h}_g^{(\text{b})},\ g=1,\dots,G.
\end{equation}
Thus, we have a set of \ac{ls} estimates $\{\hat{\alpha}_\text{L},\hat{\beta}_{\text{R},g}\},g=1,\dots,G$ given by
\begin{equation}\label{eq_coarseLS}
  [\hat{\alpha}_\text{L},\hat{\beta}_{\text{R},g}]^\mathsf{T} = (\mathbf{R}^\mathsf{H}\mathbf{R})^{-1}\mathbf{R}^\mathsf{H}\mathbf{h}_g^{(\text{b})}.
\end{equation}
Now, let
\begin{align}
  \hat{\mathbf{s}}_\text{R}&=[\hat{\beta}_{\text{R},1},\hat{\beta}_{\text{R},2},\dots,\hat{\beta}_{\text{R},G}]^\mathsf{T}
  = {\mathbf{s}}_\text{R} + \Delta{\mathbf{s}}_\text{R}\in\mathbb{C}^{G},
\end{align}  
where ${\mathbf{s}}_\text{R}=[{\beta}_{\text{R},1},{\beta}_{\text{R},2},\dots,{\beta}_{\text{R},G}]^\mathsf{T}\in\mathbb{C}^{G}$ and $\Delta{\mathbf{s}}_\text{R}$ is the estimation error of $\hat{\mathbf{s}}_\text{R}$.
By defining 
\begin{align}\label{eq_omegavartheta}
  \omega_{\vartheta_2}=({2\pi f_cd_\text{R}}/{c})\vartheta_2,\quad 
 \omega_{\vartheta_3}=({2\pi f_cd_\text{R}}/{c})\vartheta_3,
\end{align}
the noise-free vector $\mathbf{s}_\text{R}$ can be represented as a second-order tensor (matrix) $\bm{\mathcal{S}}^{(b)}\in\mathbb{C}^{\sqrt{G}\times\sqrt{G}}$ as 
\begin{align}\label{eq_tensorS}
  \bm{\mathcal{S}}^{(b)} &= \alpha_\text{R}(\mathbf{T}_3\otimes\mathbf{T}_4) \Big(\mathbf{a}^{(N_{\text{R},1})}(\omega_{\vartheta_2})\otimes\mathbf{a}^{(N_{\text{R},2})}(\omega_{\vartheta_3})\Big) \\ 
  &= \alpha_\text{R} \mathbf{T}_3 \mathbf{a}^{(N_{\text{R},1})}(\omega_{\vartheta_2}) \circ \mathbf{T}_4\mathbf{a}^{(N_{\text{R},2})}(\omega_{\vartheta_3}).
\end{align}
Therefore, the parameters $\vartheta_2$ and $\vartheta_3$ can be estimated from $\hat{\mathbf{s}}_\text{R}$ using the same routine given by~\eqref{eq_Q}, \eqref{eq_F}, and \eqref{eq_estTheta2}.

As a summary, the complete steps of the coarse channel estimation process are presented in Algorithm~\ref{algo1}. 

\begin{algorithm}[t]
  \renewcommand{\algorithmicrequire}{\textbf{Input:}}
  \renewcommand{\algorithmicensure}{\textbf{Output:}}
  \caption{Coarse Channel Estimation via Tensor-ESPRIT}
  \label{algo1}
  \begin{algorithmic}[1]
      \REQUIRE $\mathbf{y}_{g,k},g=1,\dots,G,k=1,\dots,K$.
      \ENSURE channel parameters estimate $\hat{\bm{\eta}}$.
      \STATE Obtain $\mathbf{H} = [{{\mathbf{h}}_{1}^{(\text{b})}},\dots,{{\mathbf{h}}_{G}^{(\text{b})}}]$ based on~\eqref{eq_geth} and~\eqref{eq_geth2}.
      \STATE Calculate $\mathbf{h}^{(\text{b})} = \sum_{g=1}^G {{\mathbf{h}}_{g}^{(\text{b})}}$ and form tensor $\bm{\mathcal{H}}^{(b)}$ in~\eqref{eq_bH}.
      \STATE Apply \ac{cp} decomposition to $\bm{\mathcal{H}}^{(b)}$ and obtain the factor matrices $\mathbf{U}^{(n)},n=1,\dots,3$.
      \STATE Estimate $\hat{\bm{\Theta}}_1$ through~\eqref{eq_estTheta} and obtain $\{\hat{\omega}_{\tau_\text{L}},\hat{\omega}_{\tau_\text{R}}\}$ from the phases of the eigenvalues of $\hat{\bm{\Theta}}_1$.
      \STATE Estimate $\hat{\bm{\Theta}}_2$ and $\hat{\bm{\Theta}}_3$ through~\eqref{eq_Q}--\eqref{eq_estTheta2}, and obtain $\{\hat{\omega}_{\text{L},1},\hat{\omega}_{\text{R},1},\hat{\omega}_{\text{L},2},\hat{\omega}_{\text{R},2}\}$ from the phases of the eigenvalues of $\hat{\bm{\Theta}}_2$ and $\hat{\bm{\Theta}}_3$.
      \STATE Recover $\{\hat{\theta}_\text{L}^{\text{az}},\hat{\theta}_\text{L}^{\text{el}},\hat{\theta}_\text{R}^{\text{az}},\hat{\theta}_\text{R}^{\text{el}},\hat{\tau}_\text{L},\hat{\tau}_\text{R}\}$ based on~\eqref{eq_wL}--\eqref{eq_wtau}.
      \STATE Obtain $\mathbf{s}_\text{R}$ by~\eqref{eq_coarseLS} and form tensor $\bm{\mathcal{S}}^{(b)}$ in~\eqref{eq_tensorS}.
      \STATE Obtain $\{\hat{\omega}_{\vartheta_2},\hat{\omega}_{\vartheta_3}\}$ by applying \ac{cp} decomposition and repeating Step 5 to $\bm{\mathcal{S}}^{(b)}$.
      \STATE Recover $\{\hat{\vartheta}_2,\hat{\vartheta}_3\}$ based on~\eqref{eq_omegavartheta}.
      \STATE Return $\hat{\bm{\eta}} = [\hat{\theta}_\text{L}^{\text{az}},\hat{\theta}_\text{L}^{\text{el}},\hat{\theta}_\text{R}^{\text{az}},\hat{\theta}_\text{R}^{\text{el}},\hat{\tau}_\text{L},\hat{\tau}_\text{R},\hat{\vartheta}_2,\hat{\vartheta}_3]^\mathsf{T}$.
  \end{algorithmic}
\end{algorithm}

\subsection{Channel Parameters Refinement via Least Squares}\label{sec_refined}

Channel parameter estimates $\tilde{\bm{\eta}}$ are refined based on the \ac{ls} criterion initialized using the coarse estimates from Subsection~\ref{sec_corase}.
By defining
\begin{align}
  \bm{\mu}_\text{L}^{g,k} &= \mathbf{W}^\mathsf{H} e^{-j2\pi(k-1)\Delta_f\tau_\text{L}}\mathbf{a}_\text{B}(\bm{\theta}_\text{L})x_{g,k},\\
  \bm{\mu}_\text{R}^{g,k} &= \mathbf{W}^\mathsf{H} \left[(\mathbf{a}_\text{R}(\bm{\phi}_\text{A})\odot\mathbf{a}_\text{R}(\bm{\phi}_\text{D}))^\mathsf{T}\bm{\gamma}_g\right]\\
  &\qquad\qquad\times e^{-j2\pi(k-1)\Delta_f\tau_\text{R}}\mathbf{a}_\text{B}(\bm{\theta}_\text{R})x_{g,k},
\end{align} 
we have the noise-free signals given by 
\begin{equation}\label{eq_mu}
  {\bm{\mu}}_{g,k} = \alpha_\text{L}\bm{\mu}_\text{L}^{g,k} + \alpha_\text{R}\bm{\mu}_\text{R}^{g,k}.
\end{equation}
We further define vectors $\{\mathbf{y},\bm{\mu}_\text{L}, \bm{\mu}_\text{R} \}\in\mathbb{C}^{GKN_1N_2}$ as the concatenation of $\mathbf{y}_{g,k}$, $\bm{\mu}_\text{L}^{g,k}$ and $\bm{\mu}_\text{R}^{g,k}$ over $g=1,\dots,G,k=1,\dots,K.$
Since the noise in~\eqref{eq_y_new} is colored with unknown statistics, we perform a sub-optimal estimator based on the \ac{ls} criterion as
\begin{equation}\label{eq_LS}
  \hat{\bm{\eta}}_{\text{LS}} = \arg\min_{\bm{\eta}}\ \|\mathbf{y}-\alpha_\text{L}\bm{\mu}_\text{L}(\bm{\eta}) - \alpha_\text{R}\bm{\mu}_\text{R}(\bm{\eta})\|_2^2.
\end{equation} 
The value of the complex channel gains $\alpha_\text{L}$ and $\alpha_\text{R}$ can be obtained as a function of $\mathbf{y}$, $\bm{\mu}_\text{L}$ and $\bm{\mu}_\text{R}$ by solving 
\begin{equation}
  \left\{
\begin{array}{l}
  {\partial \|\mathbf{y}-\alpha_\text{L}\bm{\mu}_\text{L}(\bm{\eta}) - \alpha_\text{R}\bm{\mu}_\text{R}(\bm{\eta})\|_2^2}/{\partial \mathfrak{R}(\alpha_\text{L})} = 0,\\
  {\partial \|\mathbf{y}-\alpha_\text{L}\bm{\mu}_\text{L}(\bm{\eta}) - \alpha_\text{R}\bm{\mu}_\text{R}(\bm{\eta})\|_2^2}/{\partial \mathfrak{R}(\alpha_\text{R})} = 0,
\end{array}\right.
\end{equation}
which give
\begin{align}
  \hat{\alpha}_\text{L} &= \frac{\bm{\mu}_\text{L}^\mathsf{H} \mathbf{y}\|\bm{\mu}_\text{R}\|_2^2 - \bm{\mu}_\text{R}^\mathsf{H} \mathbf{y}\bm{\mu}_\text{L}^\mathsf{H} \bm{\mu}_\text{R}}{\|\bm{\mu}_\text{L}\|_2^2\|\bm{\mu}_\text{R}\|_2^2 - |\bm{\mu}_\text{L}^\mathsf{H} \bm{\mu}_\text{R}|^2},\\
  \hat{\alpha}_\text{R} &= \frac{\bm{\mu}_\text{R}^\mathsf{H} \mathbf{y}\|\bm{\mu}_\text{L}\|_2^2 - \bm{\mu}_\text{L}^\mathsf{H} \mathbf{y}\bm{\mu}_\text{R}^\mathsf{H} \bm{\mu}_\text{L}}{\|\bm{\mu}_\text{L}\|_2^2\|\bm{\mu}_\text{R}\|_2^2 - |\bm{\mu}_\text{L}^\mathsf{H} \bm{\mu}_\text{R}|^2}.
\end{align}
Then, \eqref{eq_LS} can be solved by using, e.g., gradient descent method given the initialization $\hat{\bm{\eta}}$.
We denote the refined localization-related channel parameters as $\tilde{\bm{\eta}}$.

\subsection{Conversion to the Localization Domain}

Based on the refined channel parameters estimates $\tilde{\bm{\eta}}$, the localization parameters $\bm{\xi}$ can be recovered by carrying out a 2D search over $o_3$ and $\Delta$; the rest of the localization parameters can be determined from each search point $[o_3,\Delta]^\mathsf{T}$ and a cost metric can be defined to compare the fitness of different search points, as will be explained imminently. 

Given channel parameters $[\theta_\text{L}^{\text{az}},\theta_\text{L}^{\text{el}},\theta_\text{R}^{\text{az}},\theta_\text{R}^{\text{el}},\tau_\text{L},\tau_\text{R},\vartheta_2,\vartheta_3]^\mathsf{T}$ and any $\check{o}_3$ and $\check{\Delta}$, we can first determine the propagation distance of the \ac{los} \ac{ue}-\ac{bs} path and \ac{ris} reflection path as 
\begin{align}\label{eq_dldr}
  \check{d}_\text{L} = c(\tau_\text{L} - \check{\Delta}),\qquad
  \check{d}_\text{R} = c(\tau_\text{R} - \check{\Delta}).
\end{align}
which further determine the \ac{ue} position as 
\begin{equation}\label{eq_pu}
  \check{\mathbf{p}}_\text{U} = \mathbf{p}_\text{B} + \check{d}_\text{L}\mathbf{R}_\text{B}\mathbf{t}(\theta_\text{L}^{\text{az}},\theta_\text{L}^{\text{el}}).
\end{equation}
The \ac{ris} position can be obtained as the intersection of the ellipsoid $\|\check{\mathbf{p}}_\text{R}-\mathbf{p}_\text{B}\|_2+\|\check{\mathbf{p}}_\text{R}-\check{\mathbf{p}}_\text{U}\|_2=\check{d}_\text{R}$
and the line $\check{\mathbf{p}}_\text{R}=\mathbf{p}_\text{B}+x\mathbf{R}_\text{B}\mathbf{t}(\theta_\text{R}^{\text{az}},\theta_\text{R}^{\text{el}})$, with $x>0$ being the distance between \ac{bs} and \ac{ris}. The intersection is determined by solving for $x$ to obtain
\begin{equation}\label{eq_pr}
  x = \frac{\check{d}_\text{R}^2-\|\mathbf{p}_\text{B}-\check{\mathbf{p}}_\text{U}\|_2^2}{2\left(\check{d}_\text{R}+(\mathbf{R}_\text{B}\mathbf{t}(\theta_\text{R}^{\text{az}},\theta_\text{R}^{\text{el}}))^\mathsf{T} (\mathbf{p}_\text{B}-\check{\mathbf{p}}_\text{U})\right)}.
\end{equation}
Then, we can predict the intermediate measurements $\{\check{\vartheta}_2,\check{\vartheta}_3\}$ according to~\eqref{eq_vartheta2} and~\eqref{eq_vartheta3} based on $\check{o}_3$, $\check{\mathbf{p}}_\text{U}$, $\check{\mathbf{p}}_\text{R}$ and ${\mathbf{p}}_\text{B}$.
Thus, we can compute the cost metric as
\begin{equation}\label{eq_cost}
  f(\check{o}_3,\check{\Delta}) = \big\|[\check{\vartheta}_2,\check{\vartheta}_3]^\mathsf{T}-[\hat{\vartheta}_2,\hat{\vartheta}_3]^\mathsf{T} \big\|_2^2,
\end{equation}
which allows us to perform a \ac{2d} search over all the $\{o_3,\Delta\}$ candidates in a pre-defined search space. 
The optimal values that minimize $f(\check{o}_3,\check{\Delta})$ in~\eqref{eq_cost} are then returned as the estimated $\hat{o_3}$ and $\hat{\Delta}$, and the corresponding $\hat{\mathbf{p}}_\text{U}$ and $\hat{\mathbf{p}}_\text{R}$ can be determined accordingly.

Let $\{\mathcal{P}_{o_3},\mathcal{P}_\Delta\}$ denotes the search spaces and $\{d_{o_3},d_\Delta\}$ denotes the search resolutions.
We can further perform multiple rounds of search grid refinement.
As an example, for the $i$th round refined search over the search space $\{\mathcal{P}_{o_3}^i,\mathcal{P}_\Delta^i\}$ and resolution $\{d_{o_3}^i,d_\Delta^i\}$ that returns estimates $\{\hat{o}_3^i,\hat{\Delta}^i\}$, we can shrink the resolution and the search space in round $i+1$ such that
\begin{align}
  d_{o_3}^{i+1} &= \kappa d_{o_3}^{i},\qquad d_{\Delta}^{i+1} = \kappa d_{\Delta}^{i},\label{eq_shrinkreso}\\
  \mathcal{P}_{o_3}^{i+1} &= \{\dots,\hat{o}_3^i-d_{o_3}^{i+1},\hat{o}_3^i,\hat{o}_3^i+d_{o_3}^{i+1},\dots\},\label{eq_shrinkPo}\\
  \mathcal{P}_{\Delta}^{i+1} &= \{\dots,\hat{\Delta}^i-d_{\Delta}^{i+1},\hat{\Delta}^i,\hat{\Delta}^i+d_{\Delta}^{i+1},\dots\},\label{eq_shrinkPd}
\end{align}
where $\kappa\in(0,1)$ and the cardinalities of $\mathcal{P}_{o_3}^i$ and $\mathcal{P}_\Delta^i$ are fixed as $\mathrm{card}(\mathcal{P}_{o_3}^i)=C_{o_3}$ and $\mathrm{card}(\mathcal{P}_{\Delta}^i)=C_{\Delta}$ for all $i$ values.
The pseudo-code of the proposed search method is summarized in Algorithm~\ref{algo3}.

\begin{algorithm}[t]
  \renewcommand{\algorithmicrequire}{\textbf{Input:}}
  \renewcommand{\algorithmicensure}{\textbf{Output:}}
  \caption{2D Search-Based Localization Algorithm}
  \label{algo3}
  \begin{algorithmic}[1]
      \REQUIRE refined channel parameters $\tilde{\bm{\eta}}$, number of refinement~$Q$.
      \ENSURE localization parameters estimates $\hat{\bm{\xi}}$.
      \STATE Initialize search space and resolution $\{\mathcal{P}_{o_3}^0,\mathcal{P}_\Delta^0, d_{o_3}^0,d_\Delta^0\}$.
      \FOR {$i=0,\dots,Q$}
        \FOR {every candidate $\check{o}_3\in\mathcal{P}_{o_3}^i$}
          \FOR {every candidate $\check{\Delta}\in\mathcal{P}_{\Delta}^i$}
            \STATE Compute $\check{d}_\text{L}$ and $\check{d}_\text{R}$ through~\eqref{eq_dldr}.
            \STATE Estimate \ac{ue} position $\check{\mathbf{p}}_\text{U}$ through~\eqref{eq_pu} and estimate \ac{ris} position $\check{\mathbf{p}}_\text{R}$ through~\eqref{eq_pr}.
            \STATE Predict the intermediate measurements $\{\check{\vartheta}_2,\check{\vartheta}_3\}$ based on $\check{o}_3$, $\check{\mathbf{p}}_\text{U}$, $\check{\mathbf{p}}_\text{R}$ and ${\mathbf{p}}_\text{B}$.
            \STATE Compute the cost metric $f(\check{o}_3,\check{\Delta})$ based on~\eqref{eq_cost}.
          \ENDFOR
        \ENDFOR
        \STATE Select the candidate pair $\{\hat{o}_3^i,\hat{\Delta}^i\}$ that minimize $f({o}_3,{\Delta})$ and determine $\hat{\mathbf{p}}_\text{U}^i$ and $\hat{\mathbf{p}}_\text{R}^i$ based on $\{\hat{o}_3^i,\hat{\Delta}^i\}$.
        \STATE Shrink $\{\mathcal{P}_{o_3}^{i+1},\mathcal{P}_\Delta^{i+1}, d_{o_3}^{i+1},d_\Delta^{i+1}\}$ by~\eqref{eq_shrinkreso}--\eqref{eq_shrinkPd}.
      \ENDFOR
      \STATE Return $\hat{\bm{\xi}}=[(\hat{\mathbf{p}}_\text{U}^Q)^\mathsf{T},(\hat{\mathbf{p}}_\text{R}^Q)^\mathsf{T} ,\hat{o}_3^Q,\hat{\Delta}^Q]^\mathsf{T}$.
  \end{algorithmic}
\end{algorithm}

\subsection{Complexity Analysis}\label{sec_complexity}

This subsection evaluates the computational complexity of the proposed algorithms.
Among the proposed method, Algorithm~\ref{algo1} involves the channel matrix recovery with a complexity $\mathcal{O}(KGN_1N_2)$, the \ac{cp} decomposition of $\bm{\mathcal{H}}^{(b)}$ and the corresponding matrix multiplications with a complexity $\mathcal{O}(KN_1N_2)+\mathcal{O}(K^2)+\mathcal{O}(N_1^2)+\mathcal{O}(N_2^2)$, and the \ac{cp} decomposition of $\bm{\mathcal{S}}^{(b)}$ and the corresponding matrix multiplications with a complexity $\mathcal{O}(G)$. 
The \ac{ls} refinement of the channel parameters performs an iterative procedure, which gives a complexity $\mathcal{O}(TKGN_{1}N_{2})$ where $T$ denotes the total number of iterations.
Finally, a complexity of $\mathcal{O}(QC_{o_3}C_{\Delta})$ is introduced by Algorithm~\ref{algo3}. In summary, the overall complexity of the proposed solution for the \ac{jrcup} problem is given by
\begin{multline}\label{eq_complexity}
    \mathcal{O}_\text{total} = \mathcal{O}(K^2)+\mathcal{O}(N_1^2)+\mathcal{O}(N_2^2) \\ + \mathcal{O}(TKGN_{1}N_{2}) + \mathcal{O}(QC_{o_3}C_{\Delta}).
\end{multline}

%%%%%%%%%%%%%%%%%%%%%%%%%%%%%%%%%%%%%%%%%%%%%%%%%%%%%%%%%%%%%%%%
%%%%%%%%%%%%%%%%%%%%%%%%%%%%%%%%%%%%%%%%%%%%%%%%%%%%%%%%%%%%%%%%
%%%%%%%%%%%%%%%%%%%%%%%%%%%%%%%%%%%%%%%%%%%%%%%%%%%%%%%%%%%%%%%%

\section{Localization Error Bounds Derivation}
\label{sec_crlb}

\subsection{\ac{crlb} for Channel Parameters Estimation}\label{sec_CRLB1}
Considering the fact that the received noise~\eqref{eq_noisenew} is zero-mean noncircular complex Gaussian random variables, the \ac{fim} of all the unknown channel parameters $\bm{\eta}_{\text{ch}}$ is given by the following Proposition~\ref{pro_2}.
\begin{proposition}\label{pro_2}
  Based on the channel model~\eqref{eq_noisenew}, \eqref{eq_y_new} and~\eqref{eq_mu}, the \ac{fim} of channel parameters $\bm{\eta}_{\text{ch}}$ can be computed as
\begin{align}
  \mathbf{J}({\boldsymbol\eta}_\text{ch}) & 
  = \sum^{G}_{g=1} \sum^K_{k=1}
    \mathbf{D}_{g,k}^\mathsf{T} \left[\mathbf{C}_0 + \mathbf{C}_r^{g,k}\right]^{-1} \mathbf{D}_{g,k},
  \label{eq:FIM_measurement}
\end{align}
where 
\begin{align}
  \mathbf{D}_{g,k} &= \bigg[
    \mathfrak{R}\left(\frac{\partial {\bm{\mu}}_{g,k} }{\partial {\boldsymbol\eta}_\text{ch}}\right)^\mathsf{T},\ 
    \mathfrak{I}\left(\frac{\partial {\bm{\mu}}_{g,k} }{\partial {\boldsymbol\eta}_\text{ch}}\right)^\mathsf{T}
    \bigg]^\mathsf{T},\\
  \mathbf{C}_0 &= \frac{\sigma_0^2}{2}\begin{bmatrix}
    \mathfrak{R}\left(\mathbf{A}_0\mathbf{A}_0^\mathsf{H} \right) & \mathfrak{I}\left(\mathbf{A}_0\mathbf{A}_0^\mathsf{H} \right)^\mathsf{T} \vspace{0.3em} \\
    \mathfrak{I}\left(\mathbf{A}_0\mathbf{A}_0^\mathsf{H} \right) & \mathfrak{R}\left(\mathbf{A}_0\mathbf{A}_0^\mathsf{H} \right)
  \end{bmatrix},\\
  \mathbf{C}_r^{g,k} &= \frac{\sigma_r^2}{2}\begin{bmatrix}
    \mathfrak{R}\left(\mathbf{A}_r^{g,k}(\mathbf{A}_r^{g,k})^\mathsf{H} \right) & \mathfrak{I}\left(\mathbf{A}_r^{g,k}(\mathbf{A}_r^{g,k})^\mathsf{H} \right)^\mathsf{T} \vspace{0.3em} \vspace{0.2em} \\
    \mathfrak{I}\left(\mathbf{A}_r^{g,k}(\mathbf{A}_r^{g,k})^\mathsf{H} \right) & \mathfrak{R}\left(\mathbf{A}_r^{g,k}(\mathbf{A}_r^{g,k})^\mathsf{H} \right)
  \end{bmatrix},
\end{align}
with $\mathbf{A}_0 = \mathbf{W}^\mathsf{H}$ and $\mathbf{A}_r^{g,k}=\mathbf{W}^\mathsf{H}\mathbf{H}_{\text{R},2}^k\bm{\Gamma}_g$.
\end{proposition}
\begin{proof}
  See Appendix~\ref{appendix_A}
\end{proof}
\begin{remark}
Note that in~\eqref{eq:FIM_measurement}, the covariance matrix $\mathbf{C}_r^{g,k}$ is a function of ${\boldsymbol\eta}_\text{ch}$, which also contributes to the \ac{fim} of ${\boldsymbol\eta}_\text{ch}$ in principle~\cite[B.3.3]{Stoica2005Spectral}. 
Nonetheless, since the noise statistics information is not used in the channel estimation processes in Subsection~\ref{sec_corase} and~\ref{sec_refined}, we ignore this relationship in this work.
\end{remark}

Based on~\eqref{eq:FIM_measurement}, we can compute the \ac{fim} of localization-related parameters $\bm{\eta}$ using Schur's complement: we partition $\mathbf{J}(\bm{\eta}_\text{ch})=[\mathbf{X},\mathbf{Y};\mathbf{Y}^\mathsf{T},\mathbf{Z}]$,
where $\mathbf{X}\in\mathbb{R}^{8\times 8}$ so that $\mathbf{J}(\bm{\eta})=\mathbf{X}-\mathbf{Y}\mathbf{Z}^{-1}\mathbf{Y}^\mathsf{T}$.
Then the estimation error bounds for $\bm{\theta}_\text{L}$, $\bm{\theta}_\text{R}$, $\tau_\text{L}$, $\tau_\text{R}$ and $\bm{\vartheta}$ can be derived as
\begin{align}
  \mathrm{EB}(\bm{\theta}_\text{L}) & = \sqrt{\text{tr}([\mathbf{J}({\boldsymbol\eta})^{-1}]_{1:2, 1:2})},\\
  \mathrm{EB}(\bm{\theta}_\text{R}) &= \sqrt{\text{tr}([\mathbf{J}({\boldsymbol\eta})^{-1}]_{3:4, 3:4})},\\
  \mathrm{EB}(\tau_\text{L}) & = \sqrt{[\mathbf{J}({\boldsymbol\eta})^{-1}]_{5, 5}},\\
  \mathrm{EB}(\tau_\text{R}) &= \sqrt{[\mathbf{J}({\boldsymbol\eta})^{-1}]_{6, 6}},\\
  \mathrm{EB}(\bm{\vartheta}) &= \sqrt{\text{tr}([\mathbf{J}({\boldsymbol\eta})^{-1}]_{7:8, 7:8})},
\end{align}
which lower bound the estimation \acp{rmse} for the corresponding parameters.

\subsection{\ac{crlb} for Localization Parameters Estimation}

Based on the calculated $\mathbf{J}({\boldsymbol\eta})$ and the geometric model in~\eqref{eq_geo1}--\eqref{eq_tau} and \eqref{eq_vartheta2}--\eqref{eq_vartheta3}, 
we can further derive the \ac{fim} of the localization parameters $\bm{\xi}$ using the chain rule of the \ac{fim} transformation as~\cite{kay1993fundamentals}
\begin{equation}\label{eq_Ixi}
  \mathbf{J}({\boldsymbol\xi}) = \mathbf{T}^\mathsf{T} \mathbf{J}({\boldsymbol\eta})\mathbf{T},
\end{equation}
where $\mathbf{T}=\partial \bm{\eta}/\partial \bm{\xi} \in \mathbb{R}^{8\times 8}$ is the Jacobian matrix.
Then the lower bounds for the estimation \ac{rmse} of $\mathbf{p}_\text{U}$, $\mathbf{p}_\text{R}$, $o_3$ are
\begin{align}
  \mathrm{EB}(\mathbf{p}_\text{U}) &= \sqrt{\text{tr}([\mathbf{J}({\boldsymbol\xi})^{-1}]_{1:3, 1:3})},\\ 
  \mathrm{EB}(\mathbf{p}_\text{R})  &= \sqrt{\text{tr}([\mathbf{J}({\boldsymbol\xi})^{-1}]_{4:6, 4:6})},\\ 
  \mathrm{EB}(o_3) &= \sqrt{[\mathbf{J}({\boldsymbol\xi})^{-1}]_{7, 7}}.
\end{align}

%%%%%%%%%%%%%%%%%%%%%%%%%%%%%%%%%%%%%%%%%%%%%%%%%%%%%%%%%%%%%%%%
%%%%%%%%%%%%%%%%%%%%%%%%%%%%%%%%%%%%%%%%%%%%%%%%%%%%%%%%%%%%%%%%
%%%%%%%%%%%%%%%%%%%%%%%%%%%%%%%%%%%%%%%%%%%%%%%%%%%%%%%%%%%%%%%%

\section{Numerical Results}
\label{sec_sims}

\subsection{Evaluation Setup}

\begin{table}[t]
  \renewcommand{\arraystretch}{1.2}
  \begin{center}
  \caption{Default Simulation Parameters}\vspace{-0.5em}
  \label{tab1}
  \begin{tabular}{  c !{\vrule width1pt} c }
    \Xhline{1pt}
    \textbf{Parameter} & \textbf{Value}\\
    \Xhline{1pt}
    Propagation Speed $c$ & $\unit[2.9979\times 10^8]{m/s}$ \\
    \hline 
    Carrier Frequency $f_c$ & $\unit[28]{GHz}$ \\
    \hline
    Bandwidth $B$ & $\unit[100]{MHz}$ \\
    \hline 
    \# Subcarriers $K$ & $32$ \\
    \hline
    \# Transmissions $G$ & $9$ \\
    \hline
    Clock Offset $\rho$ & $\unit[100]{ns}$\\
    \hline
    Transmission Power $P_\text{T}$ & $\unit[10]{dBm}$ \\
    \hline
    Active RIS Power $P_\text{R}$ & $\unit[7]{dBm}$ \\
    \hline
    Noise PSD of Receiver \& RIS & $\unit[-174]{dBm/Hz}$ \\
    \hline
    Noise Figure of Receiver \& \ac{ris} & $\unit[10]{dB}$ \\
    \Xhline{1pt}
    Array Size of \ac{bs} / \ac{rfc} / \ac{ris}  & $10\times 10$ / $5\times 5$ / $15\times 15$\\
    \hline
    Position \& Orientation of \ac{bs}  & \makecell[c]{$\mathbf{p}_{\text{B}}=[0,5,3]^\mathsf{T}$, $\mathbf{o}_\text{B}=[0,0,-\pi/2]^\mathsf{T}$ } \\
    \hline
    Position \& Orientation of \ac{ris} & \makecell[c]{$\mathbf{p}_\text{R}=[-5,0,3]^\mathsf{T}$, $\mathbf{o}_\text{R}=[0,0,0]^\mathsf{T}$} \\
    \hline
    Position of \ac{ue} & $\mathbf{p}_\text{B}=[3,2,1]^\mathsf{T} $\\
    \Xhline{1pt}
    \end{tabular}
\end{center}
\end{table}

We consider an indoor localization scenario within a $\unit[10]{m}\times\unit[10]{m}\times\unit[3]{m}$ space.
We use random signal symbols $x_{g,k}$ with the power constraint $|x_{g,k}|=\sqrt{P_\text{T}}$, and random precoder $\mathbf{W}$ and \ac{ris} profiles $\bm{\Upsilon}$ satisfying the structure constraints~\eqref{eq_Wconstraint} and~\eqref{eq_Gammaconstraint}.
Based on~\eqref{eq_PR1}, the active \ac{ris} amplification factor $p$ is calculated as
\begin{equation}\label{eq_PR_p}
  p = \sqrt{\frac{P_\text{R}}{
    N_{\text{R},1} N_{\text{R},2}(P_\text{T}|\alpha_{\text{R},1}|^2 + \sigma_r^2 ) }+1}.
\end{equation}
The channel gains of the \ac{los} \ac{ue}-\ac{bs}, \ac{ue}-\ac{ris} and \ac{ris}-\ac{bs} paths are generated by 
$
  \alpha_\text{L} = \frac{\lambda_c}{4\pi\|\mathbf{p}_\text{U}-\mathbf{p}_\text{B}\|_2}e^{j\psi_\text{L}},\ 
  \alpha_{\text{R},1} = \frac{\lambda_c}{4\pi\|\mathbf{p}_\text{U}-\mathbf{R}_\text{R}\|_2}e^{j\psi_{\text{R},1}},\ 
  \alpha_{\text{R},2} = \frac{\lambda_c}{4\pi\|\mathbf{p}_\text{R}-\mathbf{R}_\text{B}\|_2}e^{j\psi_{\text{R},2}},
$
where $\lambda_c=c/f_c$ and $\psi_\text{L}$, $\psi_{\text{R},1}$ and $\psi_{\text{R},2}$ are independently generated from a uniform distribution $\mathcal{U}(0,2\pi)$.
When the multipath effect is introduced, as an example, the channel gains of the \ac{nlos} paths between the \ac{ue} and \ac{bs} in~\eqref{eq_mulpa1} are set as $\alpha_\text{L}^i = \frac{\sqrt{4\pi c_{\text{L},i}}\lambda_c}{16\pi^2d_{\text{U},i}d_{\text{B},i}}e^{j\psi_{\text{L},i}},\ i=1,\dots,I_\text{L}$, where $c_{\text{L},i}$ represents the radar cross section (RCS) coefficient and $\psi_{\text{L},i}$ is the random phase. Here, $d_{\text{U},i}$ and $d_{\text{B},i}$ are the distances between the \ac{ue} and the $i$th \ac{sp} and the distance between the \ac{bs} and the $i$th \ac{sp}. The channel gains $\alpha_{\text{R},1}^i$ and $\alpha_{\text{R},2}^i$ in~\eqref{eq_mulpa2} and~\eqref{eq_mulpa3} are defined in a similar manner.
In addition, we define the received \ac{snr}~as 
\begin{equation}
  \mathrm{SNR} \triangleq \frac{\sum_{g=1}^G\sum_{k=1}^K\|{\bm{\mu}}_{g,k}\|_2^2}{\sum_{g=1}^G\sum_{k=1}^K\text{tr}(\mathbf{C}_0 + \mathbf{C}_r^{g,k})}.
\end{equation}

In this paper, the three Euler angles, i.e. $[\mathbf{o}]_1$, $[\mathbf{o}]_2$ and $[\mathbf{o}]_3$, represent the rotations around $X$-axis, $Y$-axis and $Z$-axis, respectively.
The default orientation $\mathbf{o}=[0,0,0]^\mathsf{T}$ is set to face the positive $X$-axis.
The element spacings of the \ac{bs} and \ac{ris} are set as $\unit[0.5]{\lambda_c}$ and $\unit[0.2]{\lambda_c}$, respectively.
Other default simulation parameters are listed in Table~\ref{tab1}.
Throughout the simulation examples, all the involved \acp{rmse} are computed over 500 Monte Carlo trials. {The channel delays and the clock bias are presented in units of meters by multiplying them by the constant propagation speed $c$ for better intuition.}

\subsection{Performance Evaluation of the Proposed Algorithms}

\begin{figure*}[t]
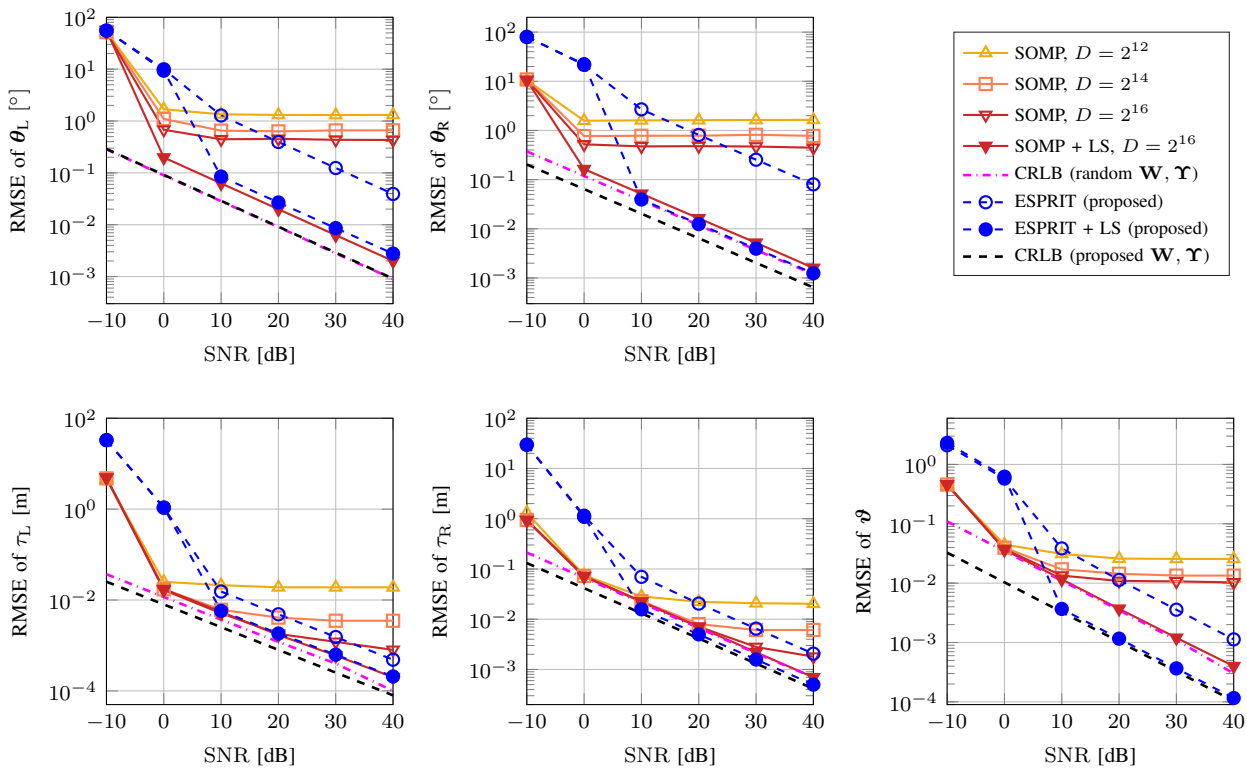

  \centering
  \include{figures/2.tex}
  \vspace{-2em}
  \caption{ 
    The evaluation of \ac{rmse} of $\bm{\theta}_\text{L}$, $\bm{\theta}_\text{R}$, $\tau_\text{L}$, $\tau_\text{R}$ and $\bm{\vartheta}$ versus received \ac{snr} for the existing SOMP algorithm, the proposed tensor-ESPRIT coarse estimation, and the proposed \ac{ls}-based refinement. 
    }
  \label{fig_2}
\end{figure*}

\begin{figure*}[t]
  \centering
\definecolor{mycolor1}{rgb}{1.00000,0.00000,1.00000}%
\definecolor{mycolor2}{rgb}{0.92941,0.69412,0.12549}%
\definecolor{ForestGreen}{rgb}{0.1333    0.5451    0.1333}
\definecolor{DarkGoldenrod}{rgb}{0.9333    0.6784    0.0549}
\definecolor{BlueViolet}{rgb}{ 0.6275    0.1255    0.9412}
\definecolor{Firebrick}{rgb}{0.8039    0.1490    0.1490}
\begin{tikzpicture}

\begin{axis}[%
  width=1.5in,
  height=1.5in,
  at={(0in,0in)},
scale only axis,
xmin=-10,
xmax=40,
xtick={-10,0,10,20,30,40},
xlabel style={font=\color{white!15!black},font=\footnotesize},
xlabel={$\mathrm{SNR}$ [dB]},
xticklabel style = {font=\color{white!15!black},font=\footnotesize},
ymode=log,
ymin=0.0004,
ymax=5000,
yminorticks=true,
ylabel style={font=\color{white!15!black},font=\footnotesize},
ylabel={RMSE of {$\mathbf{p}_\text{U}$} [m]},
yticklabel style = {font=\color{white!15!black},font=\footnotesize},
axis background/.style={fill=white},
xmajorgrids,
ymajorgrids,
legend style={at={(1,1)}, anchor=north east, legend cell align=left, align=left, draw=white!15!black,font=\scriptsize}
]
\addplot [color=blue, mark=+, line width=0.8pt, mark options={solid, blue}, mark size=2.5pt]
  table[row sep=crcr]{%
  -10	32.24850298857\\
  0	1.57765038527159\\
  10	0.138336816258025\\
  20	0.0652511553239066\\
  30	0.0595392279456372\\
  40	0.0589144589437207\\
};
\addlegendentry{Initial search}

\addplot [color=ForestGreen, mark=+, line width=0.8pt, mark options={solid, ForestGreen}, mark size=2.5pt]
  table[row sep=crcr]{%
  -10	32.2268998427402\\
  0	1.58840633117479\\
  10	0.120897288722865\\
  20	0.0401015573477628\\
  30	0.0166880902800326\\
  40	0.0121630164166215\\
};
\addlegendentry{1st refinement}

\addplot [color=mycolor2, mark=+, line width=0.8pt, mark options={solid, mycolor2}, mark size=2.5pt]
  table[row sep=crcr]{%
  -10	32.2235378279457\\
  0	1.59192422437218\\
  10	0.121411671620684\\
  20	0.0380105508669202\\
  30	0.0123533291207608\\
  40	0.00481657547386804\\
};
\addlegendentry{2nd refinement}

\addplot [color=mycolor1, mark=x, line width=0.8pt, mark options={solid, mycolor1}, mark size=2.5pt]
  table[row sep=crcr]{%
  -10	32.2228364818255\\
  0	1.59252921285188\\
  10	0.12149997688846\\
  20	0.0383748334297091\\
  30	0.0121845759424578\\
  40	0.00427116865465831\\
};
\addlegendentry{3rd refinement}

\addplot [color=black, line width=1pt, dashed]
  table[row sep=crcr]{%
  -10	0.933018577575658\\
  0	0.295046380438959\\
  10	0.0933018577576079\\
  20	0.0295046380439162\\
  30	0.00933018577575469\\
  40	0.0029504638043913\\
};
\addlegendentry{CRLB}

\end{axis}

\begin{axis}[%
  width=1.5in,
  height=1.5in,
at={(2.2in,0in)},
scale only axis,
xmin=-10,
xmax=40,
xtick={-10,0,10,20,30,40},
xlabel style={font=\color{white!15!black},font=\footnotesize},
xlabel={$\mathrm{SNR}$ [dB]},
xticklabel style = {font=\color{white!15!black},font=\footnotesize},
ymode=log,
ymin=0.0004,
ymax=5000,
yminorticks=true,
ylabel style={font=\color{white!15!black},font=\footnotesize},
ylabel={RMSE of {$\mathbf{p}_\text{R}$} [m]},
yticklabel style = {font=\color{white!15!black},font=\footnotesize},
axis background/.style={fill=white},
xmajorgrids,
ymajorgrids,
yminorgrids
]
\addplot [color=blue, mark=+, line width=0.8pt, mark options={solid, blue}, mark size=2.5pt, forget plot]
  table[row sep=crcr]{%
  -10	1499.93257372716\\
  0	25.1839172138169\\
  10	0.0368606341115181\\
  20	0.0166418182616513\\
  30	0.0146530797836833\\
  40	0.0143874942495323\\
};
\addplot [color=ForestGreen, mark=+, line width=0.8pt, mark options={solid, ForestGreen}, mark size=2.5pt, forget plot]
  table[row sep=crcr]{%
  -10	131.836420479885\\
  0	13.0832280130279\\
  10	0.032636897188646\\
  20	0.0106456558157767\\
  30	0.00423822968341568\\
  40	0.00300613691821486\\
};
\addplot [color=mycolor2, mark=+, line width=0.8pt, mark options={solid, mycolor2}, mark size=2.5pt, forget plot]
  table[row sep=crcr]{%
  -10	241.806065423954\\
  0	11.8612697936297\\
  10	0.0327605918613553\\
  20	0.0102601156313666\\
  30	0.00332126579742434\\
  40	0.00140860250011618\\
};
\addplot [color=mycolor1, mark=x, line width=0.8pt, mark options={solid, mycolor1}, mark size=2.5pt, forget plot]
  table[row sep=crcr]{%
  -10	729.031161986878\\
  0	11.6405265529622\\
  10	0.0327932208970493\\
  20	0.0103573246964043\\
  30	0.00329582505411602\\
  40	0.00114084204197421\\
};
\addplot [color=black, line width=1pt, dashed, forget plot]
  table[row sep=crcr]{%
  -10	0.245656280748866\\
  0	0.0776833368692197\\
  10	0.0245656280748956\\
  20	0.00776833368692687\\
  30	0.00245656280748813\\
  40	0.000776833368692581\\
};
\end{axis}

\begin{axis}[%
  width=1.5in,
  height=1.5in,
at={(4.4in,0in)},
scale only axis,
xmin=-10,
xmax=40,
xtick={-10,0,10,20,30,40},
xlabel style={font=\color{white!15!black},font=\footnotesize},
xlabel={$\mathrm{SNR}$ [dB]},
xticklabel style = {font=\color{white!15!black},font=\footnotesize},
ymode=log,
ymin=0.003,
ymax=100,
yminorticks=true,
ylabel style={font=\color{white!15!black},font=\footnotesize},
ylabel={RMSE of $o_3$ [$^\circ$]},
yticklabel style = {font=\color{white!15!black},font=\footnotesize},
axis background/.style={fill=white},
xmajorgrids,
ymajorgrids
]
\addplot [color=blue, mark=+, line width=0.8pt, mark options={solid, blue}, mark size=2.5pt, forget plot]
  table[row sep=crcr]{%
  -10	24.8151234566687\\
  0	12.2217547260621\\
  10	0.375126022404133\\
  20	0.194091461801553\\
  30	0.172019029783017\\
  40	0.168322236157438\\
};

\addplot [color=ForestGreen, mark=+, line width=0.8pt, mark options={solid, ForestGreen}, mark size=2.5pt, forget plot]
  table[row sep=crcr]{%
  -10	29.0183911761361\\
  0	12.7061752293422\\
  10	0.325226879641208\\
  20	0.109305422551062\\
  30	0.0463685278770797\\
  40	0.0357946987870218\\
};

\addplot [color=mycolor2, mark=+, line width=0.8pt, mark options={solid, mycolor2}, mark size=2.5pt, forget plot]
  table[row sep=crcr]{%
  -10	29.8011684156976\\
  0	12.808024193705\\
  10	0.326651474239657\\
  20	0.101988821254449\\
  30	0.0330265911086745\\
  40	0.0157890740625926\\
};

\addplot [color=mycolor1, mark=x, line width=0.8pt, mark options={solid, mycolor1}, mark size=2.5pt, forget plot]
  table[row sep=crcr]{%
  -10	29.9577491530352\\
  0	12.8293296973365\\
  10	0.326579144413771\\
  20	0.103090002636957\\
  30	0.0327020671490713\\
  40	0.0115016444274919\\
};

\addplot [color=black, line width=1pt, dashed]
  table[row sep=crcr]{%
  -10	2.54841277473194\\
  0	0.805878878642253\\
  10	0.254841277473301\\
  20	0.0805878878642764\\
  30	0.0254841277473145\\
  40	0.00805878878642693\\
};

\end{axis}

\end{tikzpicture}%
  \vspace{-2em}
  \caption{ 
    The evaluation of \ac{rmse} of $\mathbf{p}_\text{U}$, $\mathbf{p}_\text{R}$ and $o_3$ versus received \ac{snr} for different numbers ($0,1,2,3$) of grid refinements. 
    }
  \label{fig_3}
\end{figure*}
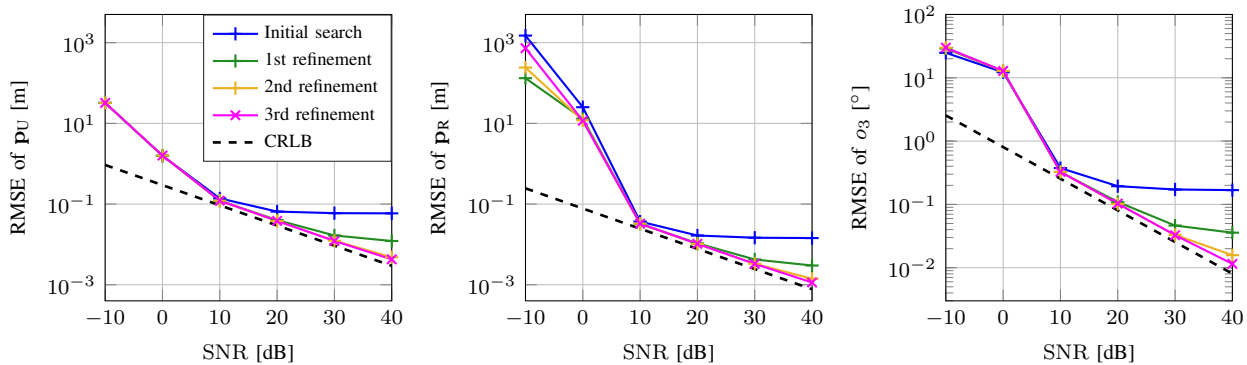

\subsubsection{Channel Estimation Performance}

We first evaluate the performance of the proposed channel estimators.
To provide a benchmark in addition to the \acp{crlb}, the proposed channel estimator is compared with the existing simultaneous orthogonal matching pursuit (SOMP) algorithm which has been verified to offer better channel estimation performance than the original orthogonal matching pursuit (OMP) algorithm~\cite{Tarboush2023Compressive}. We use SOMP to estimate the two strongest paths as the \ac{ue}-\ac{bs} and the \ac{ue}-\ac{ris}-\ac{bs} channels. For implementation details of SOMP, the readers are referred to~\cite{Tarboush2023Compressive,tarboush2023cross}.
To meet the \emph{restricted isometry property}~\cite{Candes2006Near} that SOMP requires, we use random combiner $\mathbf{W}$ and random \ac{ris} profile $\bm{\Upsilon}$ when performing SOMP, while the proposed $\{\mathbf{W},\bm{\Upsilon}\}$ in~\eqref{eq_Wconstraint} and~\eqref{eq_Gammaconstraint} are used when evaluating the proposed methods.
We derive and present the \acp{crlb} for both cases.
The SOMP dictionary sizes for each parameter in $\{\bm{\theta}_\text{L},\bm{\theta}_\text{R},\tau_\text{L},\tau_\text{R},\bm{\vartheta}\}$ are set equally as $D$, where $D=\{2^{12},2^{14},2^{16}\}$.
Besides, the \ac{ls} refinement in~\eqref{eq_LS} is solved using the trust-region method, which is implemented through Manopt toolbox~\cite{manopt} and the number of iterations is set as $T=40$. 

Fig.~\ref{fig_2} shows the evaluation of the \acp{rmse} of $\bm{\theta}_\text{L}$, $\bm{\theta}_\text{R}$, $\tau_\text{L}$, $\tau_\text{R}$ and $\bm{\vartheta}$ versus the received \ac{snr} for the SOMP algorithm, the proposed tensor-ESPRIT coarse estimation, and the proposed \ac{ls}-based refinement.
{It is observed that the proposed algorithm performs better in high-SNR regions compared to the existing SOMP algorithm but is inferior in low-SNR regions.}
While the \acp{rmse} of both coarse estimation methods exhibit large gaps from the \acp{crlb}, the proposed \ac{ls} refinement can significantly reduce the distance to the \acp{crlb} in high-\ac{snr} regions for both methods.\footnote{In the region that $\mathrm{SNR}<\unit[0]{dB}$, the \ac{ls} refinement cannot improve the performance for both SOMP and tensor-\ac{esprit} methods. This can be referred to as the threshold of \textit{no information region}, which is a well-documented phenomenon in maximum likelihood estimators~\cite{Athley2005Threshold}. The same phenomenon can be observed in Fig.~\ref{fig_3}.}
Nevertheless, there are still non-negligible gaps between the results of LS refinement (especially for $\bm{\theta}_\text{L}$, $\bm{\theta}_\text{R}$ and $\tau_\text{L}$) and the theoretical bounds, which result from the mismatch between the used \ac{ls} criterion~\eqref{eq_LS} and the actual statistics of the noise~\eqref{eq_noisenew}.
By comparing the \acp{crlb} of the random and the proposed $\{\bm{\Upsilon},\mathbf{W}\}$, we can observe that the proposed design offers lower bounds on channel parameter estimation, especially for $\bm{\theta}_\text{R}$, $\tau_\text{R}$, and $\bm{\vartheta}$.
Consequently, it is noticed that the proposed method (i.e., tensor-\ac{esprit}+\ac{ls} refinement) provides more accurate estimation of most channel parameters ($\bm{\theta}_\text{R}$, $\tau_\text{R}$, and $\bm{\vartheta}$) than the SOMP+\ac{ls} solution in high-SNR regions.

For reference, the computational complexity of the SOMP algorithm used in this paper is provided as $\mathcal{O}(DN_1N_2N_{\text{B},1}N_{\text{B},2}) + \mathcal{O}(DKN_1N_2)+\mathcal{O}(DKG)+\mathcal{O}(DGN_{\text{R},1}N_{\text{R},2})$.
According to~\eqref{eq_complexity}, the computational complexity of the proposed tensor-\ac{esprit} solution is $\mathcal{O}(KGN_1N_2) + \mathcal{O}(K^2)+\mathcal{O}(N_1^2)+\mathcal{O}(N_2^2)$, which is not a function of $D$ (i.e., search-free). 
The performance of SOMP relies on the dictionary size $D$. A large dictionary that brings heavy computation is needed for the SOMP to offer satisfactory performance, making tensor-ESPRIT preferred in scenarios that require a fast response and low computational load (e.g., as an initialization).

\subsubsection{Localization Performance}

Then, we assess the performance of Algorithm~\ref{algo3} for the second stage of localization parameters estimation.
As the second stage in JrCUP is a specialized problem, there exists no corresponding benchmark method, and only \ac{crlb} are compared.
Fig.~\ref{fig_3} presents the \acp{rmse} of estimating $\mathbf{p}_\text{U}$, $\mathbf{p}_\text{R}$ and $o_3$ versus the received \ac{snr} for different numbers $(0,1,2,3)$ of grid-search refinement iterations and $\kappa=0.1$. 
Here, the input of Algorithm~\ref{algo3} is the result of the proposed tensor-\ac{esprit}+\ac{ls} refinement in the first stage.
It can be observed that in the low \ac{snr} regions (lower than \unit[10]{dB}), the \acp{rmse} stay far from the theoretical bound. In these regions, the input channel parameter estimates contain large errors that lead to localization failure. Thus, increasing the number of grid-search refinements does not improve performance.
In the high \ac{snr} regions (\unit[10]{dB} or higher), however, we can see that the \acp{rmse} decrease as more search refinements are carried out, which indicates localization success. 
The \acp{rmse} follow the \ac{crlb} closely after two or more search iterations are performed.
These results confirm that our proposed algorithms can achieve a nearly efficient localization performance at practical \ac{snr}s (higher than \unit[10]{dB}). 
The refinement dependence of performance presents an unavoidable trade-off between localization accuracy and computational complexity in practice.

\begin{figure}[t]
  \begin{minipage}[b]{0.99\linewidth}
  \centering
\definecolor{mycolor1}{rgb}{1.00000,1.00000,0.00000}%
\definecolor{ForestGreen}{rgb}{0.1333    0.5451    0.1333}
\definecolor{BlueViolet}{rgb}{ 0.6275    0.1255    0.9412}
\definecolor{DarkGoldenrod}{rgb}{0.9333    0.6784    0.0549}
\definecolor{Coral}{rgb}{1.0000    0.4980    0.3137}
\definecolor{Firebrick}{rgb}{0.8039    0.1490    0.1490}
\definecolor{Navy}{rgb}{0    0    0.7}
\begin{tikzpicture}

\begin{axis}[%
  width=2.4in,
  height=1.5in,
  at={(0in,0in)},
scale only axis,
xmin=0,
xmax=18,
xlabel style={font=\color{white!15!black},font=\footnotesize},
xticklabel style = {font=\color{white!15!black},font=\footnotesize},
xlabel={Number of SPs, $I$},
xtick = {0,6,12,18},
ymode=log,
ymin=0.002,
ymax=7,
yminorticks=true,
ylabel style={font=\color{white!15!black},font=\footnotesize},
yticklabel style = {font=\color{white!15!black},font=\footnotesize},
ylabel={RMSE of {$\bm{\theta}_\text{R}$} [$^\circ$]},
axis background/.style={fill=white},
xmajorgrids,
ymajorgrids,
legend style={at={(1.05,1.07)}, anchor=south east, legend cell align=left, align=left, legend columns=3, draw=white!15!black, font=\tiny}
]
\addplot [color=DarkGoldenrod, line width=0.8pt, mark=triangle, mark options={solid, DarkGoldenrod}, mark size=2.5pt]
  table[row sep=crcr]{%
0	1.62069904355632\\
6	1.61745482710902\\
12	1.62977597236959\\
18	1.63825964025237\\
};
\addlegendentry{SOMP, $D = 2^{12}$}

\addplot [color=Coral, line width=0.8pt, mark=square, mark options={solid, Coral}, mark size=2.3pt]
  table[row sep=crcr]{%
0	0.814564645842275\\
6	0.831188522617283\\
12	0.802971938599737\\
18	0.843599814190016\\
};
\addlegendentry{SOMP, $D = 2^{14}$}

\addplot [color=Firebrick, line width=0.8pt, mark=triangle, mark options={solid, rotate=180, Firebrick}, mark size=2.5pt]
  table[row sep=crcr]{%
0	0.50301005012262\\
6	0.496067282518737\\
12	0.470446980717835\\
18	0.494704202044795\\
};
\addlegendentry{SOMP, $D = 2^{16}$}

\addplot [color=Firebrick, line width=0.8pt, mark=triangle*, mark options={solid, rotate=180, fill=Firebrick, draw=Firebrick}, mark size=2.5pt]
  table[row sep=crcr]{%
0	0.00645351457062163\\
6	0.0138269296742946\\
12	0.039712923258373\\
18	0.0405178874970108\\
};
\addlegendentry{SOMP + LS}

\addplot [color=blue, dashed, mark=o, line width=0.8pt, mark options={solid, blue}, mark size=2.3pt]
  table[row sep=crcr]{%
0	0.26151904830957\\
6	0.409987667540652\\
12	2.63499626232717\\
18	3.36437282994404\\
};
\addlegendentry{ESPRIT}

\addplot [color=blue, dashed, mark=*, line width=0.8pt, mark options={solid, blue}, mark size=2.3pt]
  table[row sep=crcr]{%
0	0.00413868057051906\\
6	0.0230193096062072\\
12	0.0383508036506439\\
18	0.0625131554366955\\
};
\addlegendentry{ESPRIT + LS}

\end{axis}

\end{tikzpicture}%
      \vspace{-2.5em}
      \small
      \centerline{\qquad(a) \ac{rmse} evaluation of channel parameter $\bm{\theta}_\text{R}$}
      \normalsize
  \end{minipage}
  \begin{minipage}[b]{0.99\linewidth}
  \vspace{1em}
  \centering
\definecolor{mycolor1}{rgb}{1.00000,0.00000,1.00000}%
\begin{tikzpicture}

\begin{axis}[%
  width=2.4in,
  height=1.5in,
  at={(0in,0in)},
scale only axis,
xmin=0,
xmax=18,
xlabel style={font=\color{white!15!black},font=\footnotesize},
xticklabel style = {font=\color{white!15!black},font=\footnotesize},
xlabel={Number of SPs, $I$},
xtick = {0,6,12,18},
ymode=log,
ymin=0.01,
ymax=0.21,
yminorticks=true,
ylabel style={font=\color{white!15!black},font=\footnotesize},
yticklabel style = {font=\color{white!15!black},font=\footnotesize},
ylabel={RMSE [m]},
axis background/.style={fill=white},
xmajorgrids,
ymajorgrids,
yminorgrids,
legend style={at={(1,1.07)}, anchor=south east, legend cell align=left, align=left, legend columns=2, font=\tiny, draw=white!15!black}
]
\addplot [color=mycolor1, mark=square*, line width=0.8pt, mark options={solid, mycolor1}, mark size=2.2pt]
  table[row sep=crcr]{%
0	0.0599936993956086\\
6	0.0662178647680831\\
12	0.096788230868072\\
18	0.206382469892843\\
};
\addlegendentry{$\mathbf{p}_\text{U}$ (SOMP ($D=2^{16}$) + LS + Algo. 2)}

\addplot [color=blue, mark=*, line width=0.8pt, mark options={solid, fill=blue, draw=blue}, mark size=2.5pt]
  table[row sep=crcr]{%
0	0.0604320593708442\\
6	0.063720987948902\\
12	0.108100132648599\\
18	0.153012907314534\\
};
\addlegendentry{$\mathbf{p}_\text{U}$ (proposed)}

\addplot [color=mycolor1, mark=diamond*, line width=0.8pt, mark options={solid, mycolor1}, mark size=2.7pt]
  table[row sep=crcr]{%
0	0.0148211411520685\\
6	0.0163955412256516\\
12	0.0221681745008054\\
18	0.0329960531291467\\
};
\addlegendentry{$\mathbf{p}_\text{R}$ (SOMP ($D=2^{16}$) + LS + Algo. 2)}

\addplot [color=blue,  mark=triangle*, line width=0.8pt, mark options={solid, fill=blue, draw=blue}, mark size=2.5pt]
  table[row sep=crcr]{%
0	0.014859931038444\\
6	0.0151680180575772\\
12	0.0212999998200252\\
18	0.0418431344760304\\
};
\addlegendentry{$\mathbf{p}_\text{R}$ (proposed)}

\end{axis}
\end{tikzpicture}%
      \vspace{-2.5em}
      \small
      \centerline{\qquad(b) \ac{rmse} evaluation of localization parameters $\mathbf{p}_\text{U}$ and $\mathbf{p}_\text{R}$}
      \normalsize
  \end{minipage}
    \caption{ 
    The evaluation of \acp{rmse} of the estimated channel parameters and localization parameters under the multipath effect. The tested numbers of \acp{sp} are set as $I=\{0,6,12,18\}$.  
    (a) \acp{rmse} of $\bm{\theta}_\text{R}$ by SOMP ($D=\{2^{12},2^{14},2^{16}\}$), SOMP+\ac{ls} refinement, the proposed tensor-\ac{esprit}, and the proposed tensor-\ac{esprit}+\ac{ls} refinement; 
    (b) \acp{rmse} of $\mathbf{p}_\text{U}$ and $\mathbf{p}_\text{R}$ by SOMP ($D=2^{16}$)+\ac{ls} refinement+ Algorithm~\ref{algo3} and  the proposed method (i.e., tensor-\ac{esprit}+\ac{ls} refinement+ Algorithm~\ref{algo3}).}
    \label{fig_multipath}
\end{figure}

\subsubsection{Impact of Multipath}

The impact of the multipath effect is evaluated in Fig.~\ref{fig_multipath}.
In this trial, we set the number of \acp{sp} in different channels as $I_\text{L}=I_{\text{R},1}=I_{\text{R},2}=I$.
For each of the \ac{ue}-\ac{bs}, the \ac{ue}-\ac{ris}, and the \ac{ris}-\ac{bs} channel, we randomly generate $I$ \acp{sp} within the space defined by $\unit[-5]{m}<x<\unit[5]{m},\unit[-5]{m}<y<\unit[5]{m},\unit[0]{m}<z<\unit[5]{m}$, to produce a total of $3I$ \acp{sp}.
The RCS coefficients of all \ac{nlos} paths are fixed as $c_{L,i} = c_{\text{R},1,i} = c_{\text{R},2,i} = \unit[0.5]{m^2},\ i=1,\dots,I$~\cite{Chen2023Multi}.
The received SNR is set as $30$ dB.

Fig.~\ref{fig_multipath}-(a) illustrates the multipath effect on the channel estimation performance of the proposed method together with benchmark methods. The parameter $\bm{\theta}_\text{R}$ is considered as a representative.  
From Fig.~\ref{fig_multipath}-(a), we can see that the SOMP algorithm is more robust to the multipath effect, as its estimation error remains stable with the increase of \acp{sp}. On the other hand, the estimation error of the tensor-\ac{esprit} increases with the increase of \acp{sp}.
This can be attributed to the SOMP's strategy that involves matching the atoms with the highest correlation in the dictionary, making the matching results less sensitive to weak multipath noise.
However, since the proposed tensor-ESPRIT approach is based on tensor decomposition, the structured noise introduced by \ac{nlos} multipath can increase the rank of the channel tensor, which in turn directly affects the decomposition result.
Nonetheless, both the performance of SOMP and tensor-\ac{esprit} can be effectively improved (to a similar level) by applying the \ac{ls} refinement.
The corresponding \acp{rmse} of $\mathbf{p}_\text{U}$ and $\mathbf{p}_\text{R}$ are shown in Fig.~\ref{fig_multipath}-(b).
It is clearly shown that after the proposed \ac{ls} refinement and running the localization algorithm, the SOMP and tensor-\ac{esprit} reach a similar localization accuracy. 
The more severe the multipath effect, the higher the estimation errors for both methods.
It is worth noting that in environments with sparse \ac{nlos} multipath (e.g., $I<6$) that most mmWave/THz wireless systems can satisfy~\cite{Sarieddeen2021Overview,He2022Beyond}, the final positioning accuracy remains very close to that of the multipath-free case (i.e., $I=0$), which demonstrates the robustness of the proposed \ac{ls} refinement and localization algorithm in both cases of initialization using SOMP and tensor-\ac{esprit}.

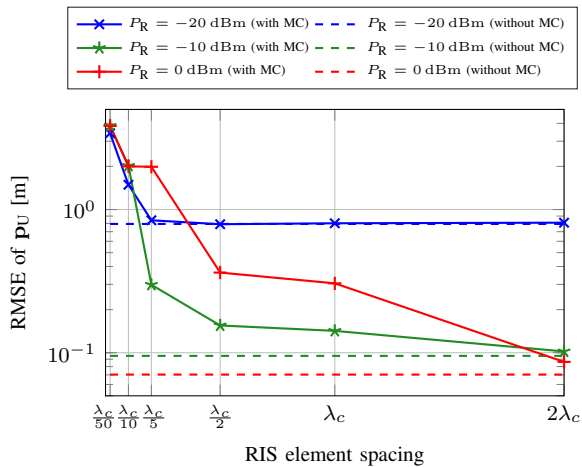
\begin{figure}[t]
  \begin{minipage}[b]{1\linewidth}
    \centering
      \definecolor{ForestGreen}{rgb}{0.1333    0.5451    0.1333}
\definecolor{DarkGoldenrod}{rgb}{0.9333    0.6784    0.0549}
\definecolor{BlueViolet}{rgb}{ 0.6275    0.1255    0.9412}
\definecolor{Firebrick}{rgb}{0.8039    0.1490    0.1490}
\begin{tikzpicture}

\begin{axis}[%
  width=2.4in,
  height=1.5in,
at={(0in,0in)},
scale only axis,
xmin=0,
xmax=2,
  xlabel style={font=\color{white!15!black},font=\footnotesize},
  xticklabel style = {font=\color{white!15!black},font=\footnotesize},
  xtick = {0.02,0.1,0.2,0.5,1,2},
  xticklabels = {{\tiny $\frac{\lambda_c}{50}\quad$},{\tiny $\frac{\lambda_c}{10}$},{\tiny $\ \frac{\lambda_c}{5}$},{\tiny $\frac{\lambda_c}{2}$},${\lambda_c}$,${2\lambda_c}$},
xlabel={RIS element spacing},
ymode=log,
ymin=0.05,
ymax=5,
yminorticks=true,
  ylabel style={font=\color{white!15!black},font=\footnotesize},
  yticklabel style = {font=\color{white!15!black},font=\footnotesize},
ylabel={RMSE of $\mathbf{p}_\text{U}$ [m]},
axis background/.style={fill=white},
xmajorgrids,
ymajorgrids,
  legend style={at={(1.03,1.07)}, anchor=south east, legend columns=2, legend cell align=left, align=left, draw=white!15!black, font=\tiny}
]
\addplot [color=blue, mark=x, line width=0.8pt, mark options={solid, blue}, mark size=2.5pt]
  table[row sep=crcr]{%
0.02	3.41103871414426\\
0.1	1.48904781768689\\
0.2	0.8409357203368\\
0.5	0.789469226034377\\
1	0.801460469201778\\
2	0.809321051280247\\
};
\addlegendentry{$P_\text{R}=\unit[-20]{dBm}$ (with MC)}

\addplot [color=blue, line width=0.8pt, dashed]
  table[row sep=crcr]{%
0.02	0.793168924299175\\
0.1	0.793168924299175\\
0.2	0.793168924299175\\
0.5	0.793168924299175\\
1	0.793168924299175\\
2	0.793168924299175\\
};
\addlegendentry{$P_\text{R}=\unit[-20]{dBm}$ (without MC)}

\addplot [color=ForestGreen, line width=0.8pt, mark=star, mark options={solid, ForestGreen}, mark size=2.5pt]
  table[row sep=crcr]{%
0.02	3.84241745750017\\
0.1	2.0203727135893\\
0.2	0.298282175607212\\
0.5	0.154979050634374\\
1	0.142061077384493\\
2	0.102067317291177\\
};
\addlegendentry{$P_\text{R}=\unit[-10]{dBm}$ (with MC)}

\addplot [color=ForestGreen, line width=0.8pt, dashed]
  table[row sep=crcr]{%
0.02	0.0950299133757751\\
0.1	0.0950299133757751\\
0.2	0.0950299133757751\\
0.5	0.0950299133757751\\
1	0.0950299133757751\\
2	0.0950299133757751\\
};
\addlegendentry{$P_\text{R}=\unit[-10]{dBm}$ (without MC)}

\addplot [color=red, mark=+, line width=0.8pt, mark options={solid, red}, mark size=2.5pt]
  table[row sep=crcr]{%
0.02	3.8339094556171\\
0.1	2.00279378587305\\
0.2	1.9882741544784\\
0.5	0.362778262400599\\
1	0.305157178566445\\
2	0.0862686598150372\\
};
\addlegendentry{$P_\text{R}=\unit[0]{dBm}$ (with MC)}

\addplot [color=red, line width=0.8pt, dashed]
  table[row sep=crcr]{%
0.02	0.0703472731388794\\
0.1	0.0703472731388794\\
0.2	0.0703472731388794\\
0.5	0.0703472731388794\\
1	0.0703472731388794\\
2	0.0703472731388794\\
};
\addlegendentry{$P_\text{R}=\unit[0]{dBm}$ (without MC)}
\end{axis}

\end{tikzpicture}%
      \vspace{-2.5em}
  \end{minipage}
  \caption{
    {Evaluation of \acp{rmse} of $\mathbf{p}_\text{U}$ versus RIS element spacing considering MC. Different active RIS powers $P_\text{R}=\{-20,-10,0\}\ \text{dBm}$ are tested.}
  }
  \label{fig_MC}
\end{figure}

{
\subsubsection{Impact of RIS Mutual Coupling} \label{sec_MC}

As previously mentioned in Subsection~\ref{eq_IIA}, the utilization of active RIS amplifies the impact of the MC among RIS elements, rendering it unignorable.
When MC is taken into account, the reflection matrix of RIS (denoted as $\tilde{\bm{\Gamma}}_g$) is given by~\cite{Shen2022Modeling,Wijekoon2023Beamforming}
\begin{equation}\label{eq_withMC}
	\tilde{\bm{\Gamma}}_g = (\bm{\Gamma}_g^{-1}-\mathbf{S})^{-1},\ g=1,\dots,G,
\end{equation} 
where $\mathbf{S}$ denotes the scattering matrix of RIS elements.
According to microwave network theory~\cite{Pozar2011Microwave,Shen2022Modeling}, the scattering matrix $\mathbf{S}$ is given by $\mathbf{S} = (\mathbf{Z}+Z_0\mathbf{I})^{-1}(\mathbf{Z}-Z_0\mathbf{I})$, where $\mathbf{Z}$ denotes the impedance matrix of RIS elements and $Z_0$ is the reference impedance (typically $Z_0=\unit[50]{\Omega}$).
In general, the matrices $\mathbf{S}$ and $\mathbf{Z}$ can be acquired through standard electromagnetic solvers such as CST Microwave Studio~\cite{Rao2023Active}. For the sake of simulation convenience, we adopt the analytical model in~\cite{Di2023Modeling}.
By assuming all the RIS antennas are cylindrical thin wires of perfectly conducting material, the mutual impedances between every pair of scattering elements of RIS can be explicitly calculated using~\cite[Eq. (2)]{Di2023Modeling} or \cite[Eq. (3)]{Zheng2023Impact}. 
The results in~\cite{Qian2021Mutual,Zheng2023Impact} reveal that a denser integration of RIS elements generally generates a greater impact on, e.g., received signal power and channel estimation performance.

Fig.~\ref{fig_MC} presents the \acp{rmse} of $\mathbf{p}_\text{U}$ versus RIS element spacings for different active RIS powers $P_\text{R}=\{-20,-10,0\}\ \text{dBm}$.
The cases with MC use the reflection matrix $\tilde{\bm{\Gamma}}_g$ in~\eqref{eq_withMC}, while the cases without MC use $\bm{\Gamma}_g$. 
To obtain the best performance, we perform both tensor-ESPRIT and SOMP at the coarse channel estimation stage and choose the result with lower residual error in~\eqref{eq_LS} to initialize the LS refinement; then the localization parameters are obtained through Algorithm~\ref{algo3}. We fix the received SNR as $\unit[30]{dB}$, and the other parameters are set according to Table~\ref{tab1}.
It can be observed that the shorter the RIS element spacing, the higher the estimation \ac{rmse}, which coincides with the results in~\cite{Qian2021Mutual,Zheng2023Impact}.
Furthermore, the gap between the cases with and without MC increases as we enlarge $P_\text{R}$ at fixed RIS element spacing, revealing that a higher active RIS power accentuates the impact of MC.
A noteworthy phenomenon is that in the absence of MC, the higher the active RIS power, the lower the estimation error. But this rule no longer holds when MC is considered. 
Higher RIS power helps to increase the signal strength but also amplifies the impact of MC, which implies that increasing the RIS power is not always beneficial. For instance, the case with $P_\text{R}=\unit[-10]{dBm}$ can provide better localization accuracy than $P_\text{R}=\unit[0]{dBm}$ when the RIS element spacing is less than $\lambda_c$.
This result reveals that an optimal active RIS power exists when MC is taken into account.
}

\subsection{Active \ac{ris} versus passive \ac{ris}}

\begin{figure}[t]
  \begin{minipage}[b]{0.99\linewidth}
  \centering
\definecolor{ForestGreen}{rgb}{0.1333    0.5451    0.1333}
\definecolor{DarkGoldenrod}{rgb}{0.9333    0.6784    0.0549}
\definecolor{BlueViolet}{rgb}{ 0.6275    0.1255    0.9412}
\definecolor{Firebrick}{rgb}{0.8039    0.1490    0.1490}

\begin{tikzpicture}

\begin{axis}[%
  width=2.4in,
  height=1.5in,
  at={(0in,0in)},
scale only axis,
xmin=-80,
xmax=40,
xlabel style={font=\color{white!15!black},font=\footnotesize},
xticklabel style = {font=\color{white!15!black},font=\footnotesize},
xlabel={$P_\text{var}$ [dBm]},
ymode=log,
ymin=0.001,
ymax=15,
yminorticks=true,
ylabel style={font=\color{white!15!black},font=\footnotesize},
yticklabel style = {font=\color{white!15!black},font=\footnotesize},
ylabel={$\mathrm{EB}(\mathbf{p}_\text{U})$  [m]},
axis background/.style={fill=white},
xmajorgrids,
ymajorgrids,
legend style={at={(0,0)}, anchor=south west, legend cell align=left, align=left, draw=white!15!black, font=\tiny}
]
\addplot [color=ForestGreen, line width=0.8pt]
  table[row sep=crcr]{%
-80	10.3978451505809\\
-70	10.3780524714206\\
-60	10.1861479095833\\
-50	8.71407082370734\\
-40	4.5419694457929\\
-30	1.57960739448338\\
-20	0.509497328708976\\
-10	0.175420951262903\\
0	0.0884603379284453\\
10	0.0743345077776185\\
20	0.072759954825996\\
30	0.0725969992721638\\
40	0.0725795390232516\\
50	0.0725774308806538\\
};
\addlegendentry{$K=32, G=9$}

\addplot [color=ForestGreen, line width=0.8pt, dashed, forget plot]
  table[row sep=crcr]{%
-80	10.399799705678\\
-70	10.3997996120798\\
-60	10.3997986760979\\
-50	10.3997893162885\\
-40	10.3996957191206\\
-30	10.3987598400938\\
-20	10.389410305772\\
-10	10.2968314020572\\
0	9.45436337825254\\
10	5.19989985803889\\
20	0.945436337825254\\
30	0.102968314020572\\
40	0.010389410305772\\
50	0.00103987598400938\\
};

\addplot [color=blue, line width=0.8pt]
  table[row sep=crcr]{%
-80	1.38265008773504\\
-70	1.38001820531443\\
-60	1.35450034180476\\
-50	1.15875631021899\\
-40	0.60399957522037\\
-30	0.210175115922603\\
-20	0.0681741146405112\\
-10	0.0245546204907199\\
0	0.0140535457842903\\
10	0.01252789439617\\
20	0.0123649527124855\\
30	0.0123485327318277\\
40	0.0123468871529263\\
50	0.0123467218780059\\
};
\addlegendentry{$K=32, G=169$}

\addplot [color=blue, line width=0.8pt, dashed, forget plot]
  table[row sep=crcr]{%
-80	1.382888344975\\
-70	1.38288833252901\\
-60	1.38288820806907\\
-50	1.38288696347093\\
-40	1.38287451761271\\
-30	1.38275007135075\\
-20	1.38150683951837\\
-10	1.36919638253256\\
0	1.25717122396172\\
10	0.691444173178945\\
20	0.125717122396172\\
30	0.0136919638253256\\
40	0.00138150683951837\\
50	0.000138275007135075\\
};
\addplot [color=red, line width=0.8pt]
  table[row sep=crcr]{%
-80	0.488839917831294\\
-70	0.487909414117702\\
-60	0.478887497827154\\
-50	0.409681629767233\\
-40	0.213545775328375\\
-30	0.0743079946371235\\
-20	0.0241031287072115\\
-10	0.00868133289195414\\
0	0.00496865777910204\\
10	0.00442926775027439\\
20	0.0043716624300029\\
30	0.00436585812712933\\
40	0.00436527665362605\\
50	0.00436521830777691\\
};
\addlegendentry{$K=128, G=169$}

\addplot [color=red, line width=0.8pt, dashed, forget plot]
  table[row sep=crcr]{%
-80	0.488924166321243\\
-70	0.488924161920925\\
-60	0.488924117917755\\
-50	0.488923677886489\\
-40	0.488919277617391\\
-30	0.488875279282239\\
-20	0.488435731079088\\
-10	0.484083333475413\\
0	0.44447651528197\\
10	0.244462083405083\\
20	0.044447651528197\\
30	0.00484083333475413\\
40	0.000488435731079088\\
50	4.88875279282239e-05\\
};
\end{axis}

\begin{axis}[%
  width=2.4in,
  height=1.5in,
  at={(0in,0in)},
  scale only axis,
  xmin=0,
  xmax=1,
  ymin=0,
  ymax=1,
  axis line style={draw=none},
  ticks=none,
  axis x line*=bottom,
  axis y line*=left
  ]
  \draw [black] (axis cs:0.6,0.78) ellipse [x radius=0.03, y radius=0.21];
  \node[right, align=left]
    at (axis cs:0.64,0.84) {\footnotesize{Passive RIS}};
  \draw [black] (axis cs:0.7,0.3) ellipse [x radius=0.03, y radius=0.2];
  \node[right, align=left]
    at (axis cs:0.69,0.06) {\footnotesize{Active RIS}};
  \end{axis}

\end{tikzpicture}%
      \vspace{-2.5em}
      \small
      \centerline{\qquad(a) $\mathrm{EB}(\mathbf{p}_\text{U})$ evaluations for different $K,G$}
      \normalsize
  \end{minipage}
  \begin{minipage}[b]{0.99\linewidth}
  \vspace{1em}
  \centering
\definecolor{ForestGreen}{rgb}{0.1333    0.5451    0.1333}
\definecolor{DarkGoldenrod}{rgb}{0.9333    0.6784    0.0549}
\definecolor{BlueViolet}{rgb}{ 0.6275    0.1255    0.9412}
\definecolor{Firebrick}{rgb}{0.8039    0.1490    0.1490}

\begin{tikzpicture}

\begin{axis}[%
  width=2.4in,
  height=1.5in,
  at={(3.3in,0in)},
  scale only axis,
  xmin=-80,
  xmax=40,
  xlabel style={font=\color{white!15!black},font=\footnotesize},
  xticklabel style = {font=\color{white!15!black},font=\footnotesize},
  xlabel={$P_\text{var}$ [dBm]},
  ymode=log,
  ymin=0.02,
  ymax=300,
  yminorticks=true,
  ylabel style={font=\color{white!15!black},font=\footnotesize},
  yticklabel style = {font=\color{white!15!black},font=\footnotesize},
  ylabel={$\mathrm{EB}(\mathbf{p}_\text{U})$  [m]},
  axis background/.style={fill=white},
  xmajorgrids,
  ymajorgrids,
  legend style={at={(0,0)}, anchor=south west, legend cell align=left, align=left, draw=white!15!black, font=\tiny}
  ]
  \addplot [color=ForestGreen, line width=0.8pt]
    table[row sep=crcr]{%
  -80	197.069655709081\\
  -70	193.779060592011\\
  -60	167.956425204864\\
  -50	89.9372100847818\\
  -40	31.5410438850336\\
  -30	10.1032734557485\\
  -20	3.23831102706595\\
  -10	1.14189703416633\\
  0	0.620810281047265\\
  10	0.54179912676798\\
  20	0.533252947244004\\
  30	0.532390286509832\\
  40	0.532303782858068\\
  50	0.532295075092289\\
  };
  \addlegendentry{RIS size $5\times 5$}
  
  \addplot [color=ForestGreen, line width=0.8pt, dashed, forget plot]
    table[row sep=crcr]{%
-80	197.446183067513\\
-70	197.446181290497\\
-60	197.446163520343\\
-50	197.445985818973\\
-40	197.444208822871\\
-30	197.426440620897\\
-20	197.248934330628\\
-10	195.491270559365\\
0	179.496530240872\\
10	98.7230916324795\\
20	17.9496530240872\\
30	1.95491270559365\\
40	0.197248934330628\\
50	0.0197426440620897\\
  };
  
  \addplot [color=blue, line width=0.8pt]
    table[row sep=crcr]{%
  -80	22.2624045630298\\
  -70	22.1674421283552\\
  -60	21.280149575998\\
  -50	15.9300501407319\\
  -40	6.85839455244078\\
  -30	2.27046844829778\\
  -20	0.73003929036665\\
  -10	0.256781074080269\\
  0	0.138494331749693\\
  10	0.120427608252023\\
  20	0.118465995654172\\
  30	0.118266897869552\\
  40	0.118246606797015\\
  50	0.118244462953107\\
  };
  \addlegendentry{RIS size $10\times 10$}

  \addplot [color=blue, line width=0.8pt, dashed, forget plot]
    table[row sep=crcr]{%
-80	22.272718421885\\
-70	22.2727182214305\\
-60	22.2727162168861\\
-50	22.2726961714616\\
-40	22.2724957192005\\
-30	22.2704913950182\\
-20	22.2504679761815\\
-10	22.0521964793641\\
0	20.2479258583252\\
10	11.1363592220789\\
20	2.02479258583252\\
30	0.220521964793641\\
40	0.0222504679761815\\
50	0.00222704913950182\\
  };
  \addplot [color=red, line width=0.8pt]
    table[row sep=crcr]{%
  -80	10.3978451505809\\
  -70	10.3780524714206\\
  -60	10.1861479095833\\
  -50	8.71407082370734\\
  -40	4.5419694457929\\
  -30	1.57960739448338\\
  -20	0.509497328708976\\
  -10	0.175420951262903\\
  0	0.0884603379284453\\
  10	0.0743345077776185\\
  20	0.072759954825996\\
  30	0.0725969992721638\\
  40	0.0725795390232516\\
  50	0.0725774308806538\\
  };
  \addlegendentry{RIS size $15\times 15$}

  \addplot [color=red, line width=0.8pt, dashed, forget plot]
    table[row sep=crcr]{%
-80	10.399799705678\\
-70	10.3997996120798\\
-60	10.3997986760979\\
-50	10.3997893162885\\
-40	10.3996957191206\\
-30	10.3987598400938\\
-20	10.389410305772\\
-10	10.2968314020572\\
0	9.45436337825254\\
10	5.19989985803889\\
20	0.945436337825254\\
30	0.102968314020572\\
40	0.010389410305772\\
50	0.00103987598400938\\
  };

\end{axis}

\begin{axis}[%
  width=2.4in,
  height=1.5in,
  at={(3.3in,0in)},
  scale only axis,
  xmin=0,
  xmax=1,
  ymin=0,
  ymax=1,
  axis line style={draw=none},
  ticks=none,
  axis x line*=bottom,
  axis y line*=left
  ]
  \draw [black] (axis cs:0.6,0.78) ellipse [x radius=0.03, y radius=0.21];
  \node[right, align=left]
    at (axis cs:0.64,0.82) {\footnotesize{Passive RIS}};
  \draw [black] (axis cs:0.7,0.24) ellipse [x radius=0.03, y radius=0.15];
  \node[right, align=left]
    at (axis cs:0.69,0.06) {\footnotesize{Active RIS}};
  \end{axis}

\end{tikzpicture}%
      \vspace{-2.5em}
      \small
      \centerline{\qquad(b) $\mathrm{EB}(\mathbf{p}_\text{U})$ evaluations for different \ac{ris} size}
      \normalsize
  \end{minipage}
    \caption{ 
    The evaluation of $\mathrm{EB}(\mathbf{p}_\text{U})$ versus active \ac{ris} power supply level for the active and passive \ac{ris} setups. 
    (a) Evaluation for the cases $\{K=32,G=9\}$, $\{K=32,G=169\}$, and $\{K=128,G=169\}$ while other parameters are fixed as in Table~\ref{tab1}; 
    (b) Evaluation for the cases $N_{\text{R},1}\times N_{\text{R},2}=5\times 5, 10\times 10, 15\times 15$ while other parameters are fixed as in Table~\ref{tab1}.
    }
  \label{fig_4}
\end{figure}

Based on the channel model~\eqref{eq_y}, the active \ac{ris} provides power gain while also introducing additional noise. Therefore, the combined effect needs to be evaluated and compared with the passive case.
As a representation of the localization performance, we evaluate the value of $\mathrm{EB}(\mathbf{p}_\text{U})$ over different power supplies in a multipath-free and MC-free scenario.
{In this trial, we define $P_\text{var}$ denoting the additional system power. For active RIS cases, we fix the transmission power as $P_{\text{T}}$ in Table~\ref{tab1}, and set $P_{\text{R}} = P_\text{var}$. The passive \ac{ris} cases are simulated by setting RIS power $P_{\text{R}}=0$ (thus $p=1$ according to~\eqref{eq_PR_p}) and $\sigma_r=0$.
To guarantee a fair comparison, we set the transmission power in the passive RIS cases as $P_{\text{T}}+P_\text{var}$, thus the total power supply of the system stays the same as in active cases. We evaluate $\mathrm{EB}(\mathbf{p}_\text{U})$ over different values of~$P_\text{var}$ from \unit[-80]{dBm} to \unit[40]{dBm}.}

Fig.~\ref{fig_4} (a) demonstrates the results for different numbers of subcarriers ($K=\{32,128\}$) and transmissions ($G=\{9,169\}$).
It is clearly shown that $\mathrm{EB}(\mathbf{p}_\text{U})$ of the active \ac{ris} decreases with the increase of RIS power supply until saturating at around $P_\text{var}=\unit[0]{dBm}$. 
This shows that, even with the introduction of more noise, the power gain from active RIS can still provide a positive improvement in localization performance.
The performance saturation for the active RIS can be explained by analyzing the noise pattern. 
When the \ac{ris} power is large, the noise introduced by active \ac{ris} $\mathbf{n}_r$ dominates $\mathbf{n}_0$.
As the noise and signal powers at the \ac{ris} channel are boosted equally, the estimation performance saturates.
{In contrast, allocating the additional power to the transmitter (i.e., the passive RIS case) can continuously improve localization performance when $P_\text{var}>\unit[0]{dBm}$.
However, it is also noted that in the passive RIS case, increasing $P_\text{var}$ offers little improvement to localization performance when $P_\text{var}<\unit[0]{dBm}$. 
The comparison concludes that in a practical region (i.e.,~$\unit[-60]{dBm}<P_\text{var}<\unit[30]{dBm}$), allocating an extra power budget to the RIS can provide a higher localization accuracy than allocating the same power to the transmitter.
Additionally, Fig.~\ref{fig_4}~(a) shows that increasing the number of subcarriers $K$ and transmissions $G$ helps to improve the localization performance for both the active and passive \ac{ris} cases.
Furthermore, larger $\{K,G\}$ enables the passive RIS to outperform the active RIS starting from a lower $P_\text{var}$ value.
Fig.~\ref{fig_4} (b) demonstrates the results for different \ac{ris} sizes, which reveal that increasing the \ac{ris} size can also improve the localization performance for both \ac{ris} types.
}

\subsection{The Blind Areas Analysis}

\subsubsection{Blind Areas Visualization and Interpretation}
\begin{figure*}[t]
  \centering
      \centerline{\includegraphics[width=0.95\linewidth]{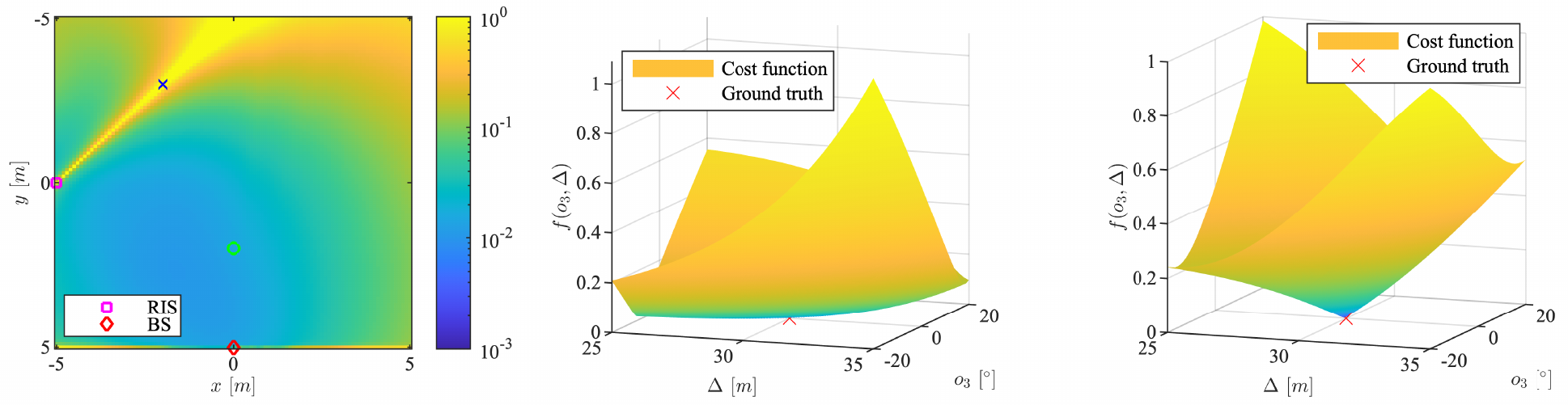}}
      \small
      \centerline{\hspace{3cm} (a) \hspace{5cm} (b) \hspace{5.5cm} (c) \hspace{3cm}}
      \caption{The visualization of the blind areas. (a)~$\mathrm{EB}(\mathbf{p}_\text{U})$ versus $\mathbf{p}_\text{U}$, where the geometries of the \ac{bs} and the \ac{ris} are fixed as in Table~\ref{tab1}. In addition, two sample locations of the blind area (blue cross) and non-blind area (green circle) are shown; (b)~$f(o_3,\Delta)$ of the blind location in (a); (c)~$f(o_3,\Delta)$ of the non-blind location in (a).}
      \normalsize
      \label{fig_BA1}
\end{figure*}

So far we have presented results for scenarios with fixed \ac{ue} and \ac{ris} positions and orientations. In this subsection, we examine the spatial variability of performance. To this end, the localization \acp{crlb} (take $\mathrm{EB}(\mathbf{p}_\text{U})$ as a representative) are computed over different \ac{ue} positions while the \ac{bs} and \ac{ris} positions and orientations are kept fixed. 
We assume the \ac{ue} to be placed across a $\unit[10]{m}\times\unit[10]{m}$ space at a fixed height of \unit[1]{m}, and the \ac{bs} and \ac{ris} are deployed as default parameters in Table~\ref{tab1}. The corresponding results are shown in Fig.~\ref{fig_BA1}~(a). 
From Fig.~\ref{fig_BA1}~(a), we can observe the presence of areas with extremely high \ac{crlb} (yellow areas). 
Since these areas with high \ac{crlb} yield a poor localization performance or are even unable to perform localization because of the existence of the ambiguity~\cite{Lu2022Joint}, we name those areas as the \emph{blind areas}. 
In Fig.~\ref{fig_BA1}~(a), we select two sample locations of the blind area (blue cross) and non-blind area (green circle), respectively, for further investigation.

Fig.~\ref{fig_BA1}~(b) and (c) visualize the cost function $f({o}_3,{\Delta})$ in \eqref{eq_cost} with the noise-free observations $\bm{\eta}$ for the selected blind and non-blind locations of the \ac{ue} as marked in Fig.~\ref{fig_BA1}~(a). 
The ground-truth values of ${o}_3$ and ${\Delta}$ are marked with a red cross. In both cases, we can see that the ground truth coincides with a local minimum point, which indicates that Algorithm~\ref{algo3} can converge to the ground truth given a proper initialization search of the interval.
However, the cost function of the blind location is flat around the ground truth, while the non-blind location shows a sharper descending structure.
This implies that the uncertainty in the blind location is much higher than that in the non-blind location,  
and so the blind locations would produce a larger estimation error than the non-blind locations in the noisy case, which results in the extremely high \acp{crlb} in Fig.~\ref{fig_BA1}~(a).
To avoid localization in blind areas, we propose two strategies, namely, leveraging more prior information and adding extra \acp{bs}.
  
\subsubsection{Evaluation of the Impact of Extra Prior Information}

  \begin{figure}[t]
  \centering
  \centerline{\includegraphics[width=3.4in]{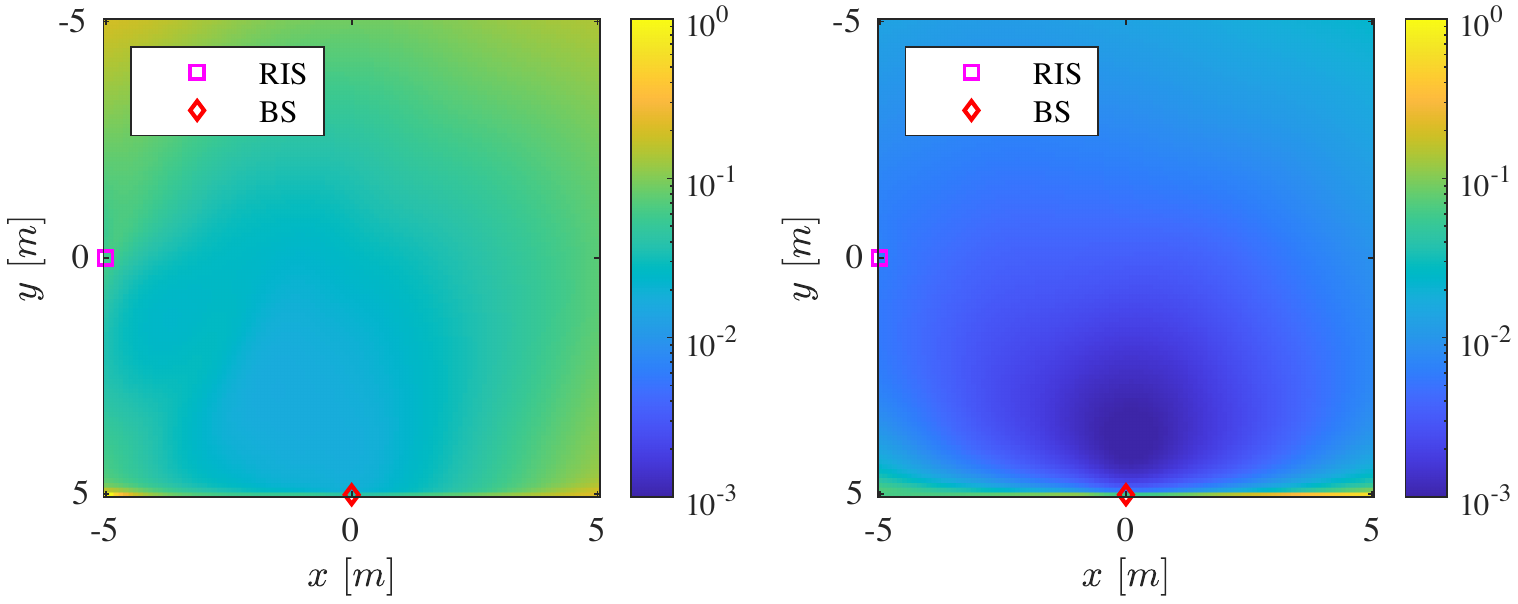}}
  \small
    \centerline{ (a) \hspace{3.5cm} (b)}
  \normalsize
  \caption{ 
    $\mathrm{EB}(\mathbf{p}_\text{U})$ versus $\mathbf{p}_\text{U}$. (a)~The \ac{ris}'s orientation $o_3$ is known; (b)~The clock bias $\Delta$ is known.
    }
    \label{fig_BA2}
  \end{figure}

We first assess the impact of using extra prior information on unknown localization parameters. 
We test two types of prior information, i.e., the \ac{ris}'s orientation $o_3$ and the clock bias $\Delta$. 
Assume that we know the values of these two parameters in advance. Then, we remove the corresponding columns in the Jacobian matrix $\mathbf{J}$ in~\eqref{eq_Ixi} and calculate a new \ac{crlb} accordingly. 
The results are shown in Fig.~\ref{fig_BA2}.
It can be seen that with prior knowledge of $o_3$ or $\Delta$, the blind area is greatly reduced. 
Furthermore, prior knowledge of $\Delta$ seems to provide a better performance in eliminating the blind area effect compared to using $o_3$.
This can be explained by the fact that $\Delta$ contains more information than $o_3$ regarding the geometry of the \ac{ue} and the \ac{ris}.
For example, once we know the clock bias $\Delta$, we can determine the positions of the \ac{ue} and the \ac{ris} immediately through~\eqref{eq_dldr}--\eqref{eq_pr}, while the \ac{ris} orientation $o_3$ cannot provide further information other than itself.

\subsubsection{Evaluation of the Impact of Additional \acp{bs}}

  \begin{figure}[t]
  \centering
  \centerline{\includegraphics[width=3.4in]{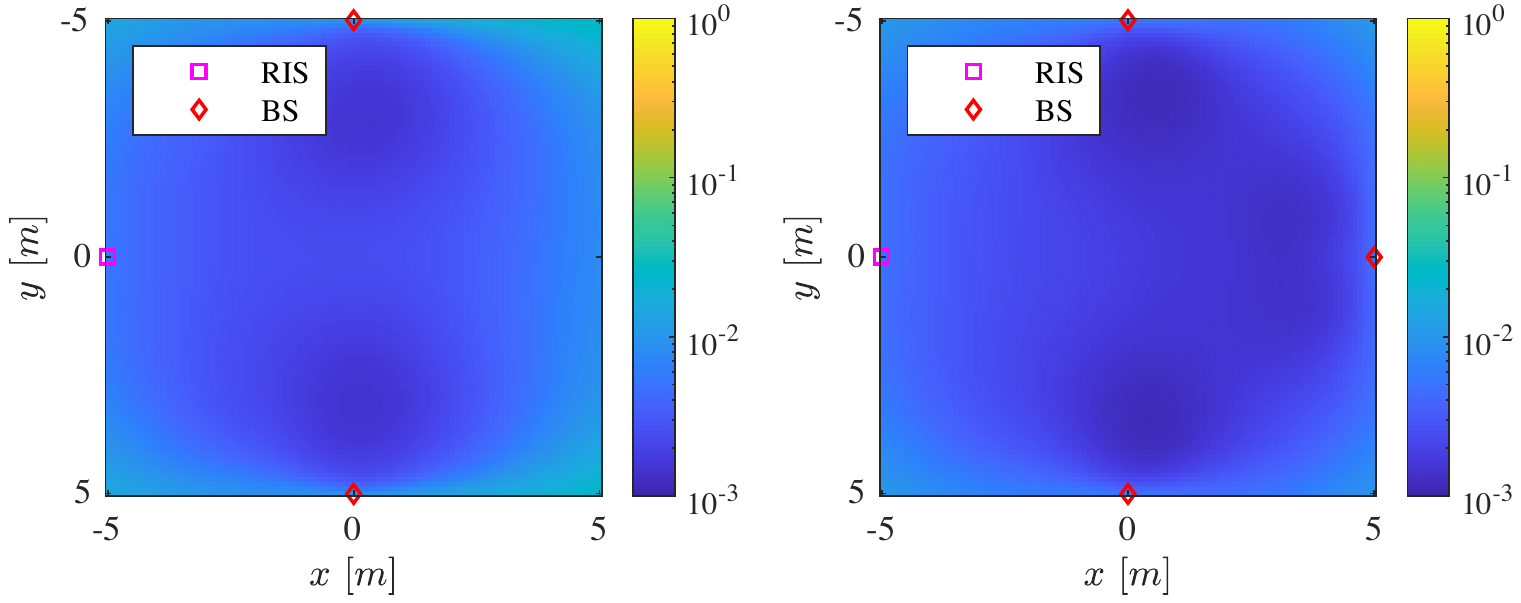}}
  \small
    \centerline{ (a) \hspace{3.5cm} (b)}
  \normalsize
  \caption{ 
       $\mathrm{EB}(\mathbf{p}_\text{U})$ versus $\mathbf{p}_\text{U}$. (a)~Add one \ac{bs} at $[0,-5,3]^\mathsf{T}$ with orientation $[0,0,\pi/2]^\mathsf{T}$; (b)~Add two \acp{bs} at $\{[0,-5,3]^\mathsf{T},[5,0,3]^\mathsf{T}\}$ with orientation $\{[0,0,\pi/2]^\mathsf{T},[0,0,\pi]^\mathsf{T}\}$.
    }
  \label{fig_BA3}
  \end{figure}

Finally, we evaluate the impact of adding more \acp{bs}.
Fig.~\ref{fig_BA3} demonstrates the value of $\mathrm{EB}(\mathbf{p}_\text{U})$ for \ac{ue} at different positions under different amount of \acp{bs}.
In Fig.~\ref{fig_BA3} (a), an additional \ac{bs} is introduced at $[0,-5,3]^\mathsf{T}$ with orientation $[0,0,\pi/2]^\mathsf{T}$, while one more \ac{bs} is added at $[5,0,3]^\mathsf{T}$ with orientation $[0,0,\pi]^\mathsf{T}$ in Fig.~\ref{fig_BA3} (b).
For each case, we collect parallel observations $\mathbf{y}_{g,k}$ from multiple \acp{bs}, which is dependent on the same geometric parameters of the \ac{ue} and the \ac{ris}.
Then the \acp{crlb} are rederived and evaluated. 
We can observe that adding more \acp{bs} can significantly reduce the blind areas. The more \acp{bs} we deploy, the lower the overall localization bounds.

\section{Conclusion}
\label{sec_conclusion}

In this paper, we formulated and solved a joint \ac{ris} calibration and user positioning problem for an active RIS-assisted uplink \ac{simo} system, where the 3D user position, 3D RIS position, 1D RIS orientation, and clock bias were estimated. A two-stage localization method has been proposed, which consists of a coarse channel parameter estimation using tensor-\ac{esprit}, a channel parameters refinement via \ac{ls} estimation, and a \ac{2d} search-based localization algorithm. 
The fundamental \acp{crlb} for the channel parameter and localization parameter estimation were derived.
Through simulation studies, we demonstrated the effectiveness of the proposed algorithms by comparing their estimation \acp{rmse} with those of the existing SOMP algorithm and the derived \acp{crlb}. 
In addition, a comparison between the performance of the active and passive \acp{ris} was carried out, which showed that active \acp{ris} outperform the passive \ac{ris}s in terms of localization accuracy within practical power supply regions. 
Furthermore, we show that blind areas exist in the \ac{jrcup} problem, which can be interpreted by the uncertainty in the cost function. 
Two strategies are proposed to combat blind areas, namely, using additional prior information (e.g., from extra sensors) or deploying more \acp{bs}. These strategies have been verified by numerical simulations.
Future research on the \ac{jrcup} problem includes extending the formulation to multi-user/multi-\ac{ris} scenarios, addressing \ac{3d} \ac{ris} orientation estimation, as well as developing the \ac{ris} mutual coupling-aware estimators.

\appendices
\section{}\label{appendix_A}
Consider estimating the deterministic unknowns $\mathbf{x}\in\mathbb{R}^l$ from $\tilde{\mathbf{y}}\in\mathbb{C}^m$ based on the noisy observation model
\begin{equation}\label{eq_obsmodel}
  \tilde{\mathbf{y}} = \tilde{\mathbf{f}}(\mathbf{x}) + \tilde{\mathbf{A}}\tilde{\mathbf{n}},
\end{equation}
where $\tilde{\mathbf{A}}\in\mathbf{C}^{m\times n}$, $\tilde{\mathbf{n}}\in\mathbb{C}^n$ and $\tilde{\mathbf{n}}\sim\mathcal{C}\mathcal{N}(\mathbf{0},\sigma^2\mathbf{I}_n) $.
Suppose $\tilde{\mathbf{y}} = \mathbf{y}_\mathrm{R} + j\mathbf{y}_\mathrm{I}$, $\tilde{\mathbf{f}}(\mathbf{x}) = \mathbf{f}_\mathrm{R}(\mathbf{x})+j\mathbf{f}_\mathrm{I}(\mathbf{x})$, $\tilde{\mathbf{A}}=\mathbf{A}_\mathrm{R}+j\mathbf{A}_\mathrm{I}$ and $\tilde{\mathbf{n}}=\mathbf{n}_\mathrm{R}+j\mathbf{n}_\mathrm{I}$.
By separating the real and imaginary parts of the observations $\tilde{\mathbf{y}}$, we can rewrite the model~\eqref{eq_obsmodel} as 
\begin{equation}\label{eq_reobsmodel}
  \begin{bmatrix}
    \mathbf{y}_\mathrm{R}\\
    \mathbf{y}_\mathrm{I}
  \end{bmatrix} = 
  \underbrace{
  \begin{bmatrix}
    \mathbf{f}_\mathrm{R}(\mathbf{x})\\
    \mathbf{f}_\mathrm{I}(\mathbf{x})
  \end{bmatrix}}_{\mathbf{f}(\mathbf{x})} + \underbrace{
  \begin{bmatrix}
    \mathbf{A}_\mathrm{R}\mathbf{n}_\mathrm{R}-\mathbf{A}_\mathrm{I}\mathbf{n}_\mathrm{I}\\
    \mathbf{A}_\mathrm{R}\mathbf{n}_\mathrm{I}+\mathbf{A}_\mathrm{I}\mathbf{n}_\mathrm{R}
  \end{bmatrix}}_{\mathbf{n}}.
\end{equation}
Since 
\vspace{-1em}
\begin{align}
  \mathbf{A}_\mathrm{R}\mathbf{n}_\mathrm{R}-\mathbf{A}_\mathrm{I}\mathbf{n}_\mathrm{I}\sim\mathcal{N}\Big(\mathbf{0},\frac{\sigma^2}{2}(\mathbf{A}_\mathrm{R}\mathbf{A}_\mathrm{R}^\mathsf{T}+\mathbf{A}_\mathrm{I}\mathbf{A}_\mathrm{I}^\mathsf{T} )\Big),\\
  \mathbf{A}_\mathrm{R}\mathbf{n}_\mathrm{I}+\mathbf{A}_\mathrm{I}\mathbf{n}_\mathrm{R}\sim\mathcal{N}\Big(\mathbf{0},\frac{\sigma^2}{2}(\underbrace{\mathbf{A}_\mathrm{R}\mathbf{A}_\mathrm{R}^\mathsf{T}+\mathbf{A}_\mathrm{I}\mathbf{A}_\mathrm{I}^\mathsf{T}}_{ = \mathfrak{R}(\tilde{\mathbf{A}}\tilde{\mathbf{A}}^\mathsf{H})} )\Big),
\end{align} \vspace{-1em}
and
\begin{align}
  &\mathbb{E}\left[(\mathbf{A}_\mathrm{R}\mathbf{n}_\mathrm{R}-\mathbf{A}_\mathrm{I}\mathbf{n}_\mathrm{I})(\mathbf{A}_\mathrm{R}\mathbf{n}_\mathrm{I}+\mathbf{A}_\mathrm{I}\mathbf{n}_\mathrm{R})^\mathsf{T} \right]\notag\\
  =& \frac{\sigma^2}{2}(\mathbf{A}_\mathrm{R}\mathbf{A}_\mathrm{I}^\mathsf{T}-\mathbf{A}_\mathrm{I}\mathbf{A}_\mathrm{R}^\mathsf{T} ) = \frac{\sigma^2}{2}\mathfrak{I}(\tilde{\mathbf{A}}\tilde{\mathbf{A}}^\mathsf{H})^\mathsf{T},
\end{align}
we have 
\begin{equation}\label{eq_noisestatistics}
  \mathbf{n}\sim\mathcal{N}\Bigg(\mathbf{0}, \underbrace{\frac{\sigma^2}{2}\begin{bmatrix}
    \mathfrak{R}\left(\mathbf{A}\mathbf{A}^\mathsf{H} \right) & \mathfrak{I}\left(\mathbf{A}\mathbf{A}^\mathsf{H} \right)^\mathsf{T} \vspace{0.3em} \\
    \mathfrak{I}\left(\mathbf{A}\mathbf{A}^\mathsf{H} \right) & \mathfrak{R}\left(\mathbf{A}\mathbf{A}^\mathsf{H} \right)
  \end{bmatrix}}_{\mathbf{C}_n}\Bigg).
\end{equation}
Then the \ac{fim} of $\mathbf{x}$ can be derived as~\cite[B.3.3]{Stoica2005Spectral}
\begin{equation}
  \mathbf{J}(\mathbf{x}) = \begin{bmatrix}
    \mathfrak{R}(\frac{\partial \tilde{\mathbf{f}}(\mathbf{x})}{\partial \mathbf{x}})\\
    \mathfrak{I}(\frac{\partial \tilde{\mathbf{f}}(\mathbf{x})}{\partial \mathbf{x}})
  \end{bmatrix}^\mathsf{T} 
  \mathbf{C}_n^{-1}
  \begin{bmatrix}
    \mathfrak{R}(\frac{\partial \tilde{\mathbf{f}}(\mathbf{x})}{\partial \mathbf{x}})\\
    \mathfrak{I}(\frac{\partial \tilde{\mathbf{f}}(\mathbf{x})}{\partial \mathbf{x}})
  \end{bmatrix}.
\end{equation}
Letting $\mathbf{A}_0 = \mathbf{W}^\mathsf{H}$, $\mathbf{A}_r^{g,k}=\mathbf{W}^\mathsf{H}\mathbf{H}_{\text{R},2}^k\bm{\Gamma}_g$ and adding up the \acp{fim} over all the transmissions $g=1,\dots,G$ and subcarriers $k=1,\dots,K$ yield Proposition~\ref{pro_2}.
\vspace{-0.5em}

\bibliography{references}

\begin{thebibliography}{10}
\providecommand{\url}[1]{#1}
\csname url@samestyle\endcsname
\providecommand{\newblock}{\relax}
\providecommand{\bibinfo}[2]{#2}
\providecommand{\BIBentrySTDinterwordspacing}{\spaceskip=0pt\relax}
\providecommand{\BIBentryALTinterwordstretchfactor}{4}
\providecommand{\BIBentryALTinterwordspacing}{\spaceskip=\fontdimen2\font plus
\BIBentryALTinterwordstretchfactor\fontdimen3\font minus
  \fontdimen4\font\relax}
\providecommand{\BIBforeignlanguage}[2]{{%
\expandafter\ifx\csname l@#1\endcsname\relax
\typeout{** WARNING: IEEEtran.bst: No hyphenation pattern has been}%
\typeout{** loaded for the language `#1'. Using the pattern for}%
\typeout{** the default language instead.}%
\else
\language=\csname l@#1\endcsname
\fi
#2}}
\providecommand{\BIBdecl}{\relax}
\BIBdecl

\bibitem{Liu2021Reconfigurable}
Y.~Liu, X.~Liu, X.~Mu, T.~Hou, J.~Xu, M.~Di~Renzo, and N.~Al-Dhahir,
  ``Reconfigurable intelligent surfaces: Principles and opportunities,''
  \emph{IEEE Communications Surveys \& Tutorials}, vol.~23, no.~3, pp.
  1546--1577, 2021.

\bibitem{Pan2022Overview}
C.~Pan, G.~Zhou, K.~Zhi, S.~Hong, T.~Wu, Y.~Pan, H.~Ren, M.~D. Renzo,
  A.~Lee~Swindlehurst, R.~Zhang, and A.~Y. Zhang, ``An overview of signal
  processing techniques for {RIS/IRS-Aided} wireless systems,'' \emph{IEEE
  Journal of Selected Topics in Signal Processing}, vol.~16, no.~5, pp.
  883--917, 2022.

\bibitem{He2022Beyond}
J.~He, F.~Jiang, K.~Keykhosravi, J.~Kokkoniemi, H.~Wymeersch, and M.~Juntti,
  ``Beyond {5G} {RIS} {mmWave} systems: Where communication and localization
  meet,'' \emph{IEEE Access}, vol.~10, pp. 68\,075--68\,084, 2022.

\bibitem{Sarieddeen2021Overview}
H.~Sarieddeen, M.-S. Alouini, and T.~Y. Al-Naffouri, ``An overview of signal
  processing techniques for terahertz communications,'' \emph{Proceedings of
  the IEEE}, vol. 109, no.~10, pp. 1628--1665, 2021.

\bibitem{Chen2022Tutorial}
H.~Chen, H.~Sarieddeen, T.~Ballal, H.~Wymeersch, M.-S. Alouini, and T.~Y.
  Al-Naffouri, ``A tutorial on terahertz-band localization for {6G}
  communication systems,'' \emph{IEEE Communications Surveys \& Tutorials},
  vol.~24, no.~3, pp. 1780--1815, 2022.

\bibitem{Zhang2022Active}
Z.~Zhang, L.~Dai, X.~Chen, C.~Liu, F.~Yang, R.~Schober, and H.~Vincent~Poor,
  ``Active {RIS} vs. passive {RIS}: Which will prevail in {6G}?'' \emph{IEEE
  Transactions on Communications}, vol.~71, no.~3, pp. 1707--1725, 2023.

\bibitem{Schroeder2022Two}
R.~Schroeder, J.~He, G.~Brante, and M.~Juntti, ``Two-stage channel estimation
  for hybrid {RIS} assisted {MIMO} systems,'' \emph{IEEE Transactions on
  Communications}, vol.~70, no.~7, pp. 4793--4806, 2022.

\bibitem{Mu2022Simultaneously}
X.~Mu, Y.~Liu, L.~Guo, J.~Lin, and R.~Schober, ``Simultaneously transmitting
  and reflecting ({STAR}) {RIS} aided wireless communications,'' \emph{IEEE
  Transactions on Wireless Communications}, vol.~21, no.~5, pp. 3083--3098,
  2022.

\bibitem{Peral2018Survey}
J.~A. del Peral-Rosado, R.~Raulefs, J.~A. López-Salcedo, and G.~Seco-Granados,
  ``Survey of cellular mobile radio localization methods: From {1G} to {5G},''
  \emph{IEEE Communications Surveys \& Tutorials}, vol.~20, no.~2, pp.
  1124--1148, 2018.

\bibitem{Zheng20235G}
P.~Zheng, X.~Liu, T.~Ballal, and T.~Y. Al-Naffouri, ``{5G}-aided {RTK}
  positioning in {GNSS}-deprived environments,'' in \emph{IEEE European Signal
  Processing Conference (EUSIPCO)}, 2023.

\bibitem{fang20203}
X.~Fang, X.~Li, and L.~Xie, ``{3-D} distributed localization with mixed local
  relative measurements,'' \emph{IEEE Transactions on Signal Processing},
  vol.~68, pp. 5869--5881, 2020.

\bibitem{Bjornson2022Reconfigurable}
E.~Bj\"{o}rnson, H.~Wymeersch, B.~Matthiesen, P.~Popovski, L.~Sanguinetti, and
  E.~de~Carvalho, ``Reconfigurable intelligent surfaces: A signal processing
  perspective with wireless applications,'' \emph{IEEE Signal Processing
  Magazine}, vol.~39, no.~2, pp. 135--158, 2022.

\bibitem{chen2023riss}
H.~Chen, H.~Kim, M.~Ammous, G.~Seco-Granados, G.~C. Alexandropoulos, S.~Valaee,
  and H.~Wymeersch, ``{RISs} and sidelink communications in smart cities: {The}
  key to seamless localization and sensing,'' \emph{accepted by IEEE
  Communications Magazine}, 2023.

\bibitem{Kim2022RIS}
H.~Kim, H.~Chen, M.~F. Keskin, Y.~Ge, K.~Keykhosravi, G.~C. Alexandropoulos,
  S.~Kim, and H.~Wymeersch, ``{RIS}-enabled and access-point-free simultaneous
  radio localization and mapping,'' \emph{preprint arXiv:2212.07141}, 2022.

\bibitem{Chen2023Multi}
H.~Chen, P.~Zheng, M.~F. Keskin, T.~Al-Naffouri, and H.~Wymeersch,
  ``Multi-{RIS}-enabled {3D} sidelink positioning,'' \emph{preprint
  arXiv:2302.12459}, 2023.

\bibitem{Keykhosravi2022RIS}
K.~Keykhosravi, M.~F. Keskin, G.~Seco-Granados, P.~Popovski, and H.~Wymeersch,
  ``{RIS-Enabled} {SISO} localization under user mobility and spatial-wideband
  effects,'' \emph{IEEE Journal of Selected Topics in Signal Processing},
  vol.~16, no.~5, pp. 1125--1140, 2022.

\bibitem{He2022Simultaneous}
J.~He, A.~Fakhreddine, and G.~C. Alexandropoulos, ``Simultaneous indoor and
  outdoor {3D} localization with {STAR-RIS-assisted} millimeter wave systems,''
  in \emph{IEEE Vehicular Technology Conference (VTC)}, 2022.

\bibitem{Zhang2021MetaLocalization}
H.~Zhang, H.~Zhang, B.~Di, K.~Bian, Z.~Han, and L.~Song, ``Metalocalization:
  Reconfigurable intelligent surface aided multi-user wireless indoor
  localization,'' \emph{IEEE Transactions on Wireless Communications}, vol.~20,
  no.~12, pp. 7743--7757, 2021.

\bibitem{Elzanaty2021Reconfigurable}
A.~Elzanaty, A.~Guerra, F.~Guidi, and M.-S. Alouini, ``Reconfigurable
  intelligent surfaces for localization: Position and orientation error
  bounds,'' \emph{IEEE Transactions on Signal Processing}, vol.~69, pp.
  5386--5402, 2021.

\bibitem{Fascista2022RIS}
A.~Fascista, M.~F. Keskin, A.~Coluccia, H.~Wymeersch, and G.~Seco-Granados,
  ``{RIS}-aided joint localization and synchronization with a single-antenna
  receiver: Beamforming design and low-complexity estimation,'' \emph{IEEE
  Journal of Selected Topics in Signal Processing}, vol.~16, no.~5, pp.
  1141--1156, 2022.

\bibitem{Gao2022Wireless}
P.~Gao, L.~Lian, and J.~Yu, ``Wireless area positioning in {RIS-Assisted}
  {mmWave} systems: Joint passive and active beamforming design,'' \emph{IEEE
  Signal Processing Letters}, vol.~29, pp. 1372--1376, 2022.

\bibitem{Keskin2022Optimal}
M.~F. Keskin, F.~Jiang, F.~Munier, G.~Seco-Granados, and H.~Wymeersch,
  ``Optimal spatial signal design for {mmWave} positioning under imperfect
  synchronization,'' \emph{IEEE Transactions on Vehicular Technology}, vol.~71,
  no.~5, pp. 5558--5563, 2022.

\bibitem{Zheng2022Misspecified}
P.~Zheng, H.~Chen, T.~Ballal, H.~Wymeersch, and T.~Y. Al-Naffouri,
  ``Misspecified {Cram\'er-Rao} bound of {RIS-aided} localization under
  geometry mismatch,'' in \emph{IEEE International Conference on Acoustics,
  Speech, \& Signal Processing (ICASSP)}, 2023.

\bibitem{Emenonye2022RIS}
D.-R. Emenonye, H.~S. Dhillon, and R.~M. Buehrer, ``{RIS-Aided} localization
  under position and orientation offsets in the near and far field,''
  \emph{preprint arXiv:2210.03599}, 2022.

\bibitem{Samir2021Optimizing}
M.~Samir, M.~Elhattab, C.~Assi, S.~Sharafeddine, and A.~Ghrayeb, ``Optimizing
  age of information through aerial reconfigurable intelligent surfaces: A deep
  reinforcement learning approach,'' \emph{IEEE Transactions on Vehicular
  Technology}, vol.~70, no.~4, pp. 3978--3983, 2021.

\bibitem{Ge2022Reconfigurable}
L.~Ge, H.~Zhang, J.-B. Wang, and G.~Y. Li, ``Reconfigurable wireless relaying
  with {Multi-UAV-carried} intelligent reflecting surfaces,'' \emph{IEEE
  Transactions on Vehicular Technology}, vol.~72, no.~4, pp. 4932--4947, 2023.

\bibitem{Lu2022Joint}
Y.~Lu, H.~Chen, J.~Talvitie, H.~Wymeersch, and M.~Valkama, ``Joint {RIS}
  calibration and multi-user positioning,'' in \emph{IEEE Vehicular Technology
  Conference (VTC)}, 2022.

\bibitem{Ghazalian2022Joint}
R.~Ghazalian, H.~Chen, G.~C. Alexandropoulos, G.~Seco-Granados, H.~Wymeersch,
  and R.~J{\"a}ntti, ``Joint user localization and location calibration of a
  hybrid reconfigurable intelligent surface,'' \emph{preprint
  arXiv:2210.10150}, 2022.

\bibitem{Mylonopoulos2022Active}
G.~Mylonopoulos, C.~D’Andrea, and S.~Buzzi, ``Active reconfigurable
  intelligent surfaces for user localization in {mmWave} {MIMO} systems,'' in
  \emph{IEEE International Workshop on Signal Processing Advances in Wireless
  Communications (SPAWC)}, 2022.

\bibitem{Mylonopoulos2023Maximum}
G.~Mylonopoulos, L.~Venturino, S.~Buzzi, and C.~D’Andrea,
  ``Maximum-likelihood user localization in active-{RIS} empowered {mmWave}
  wireless networks,'' in \emph{17th European Conference on Antennas and
  Propagation}, 2023.

\bibitem{Abeywickrama2020Intelligent}
S.~Abeywickrama, R.~Zhang, Q.~Wu, and C.~Yuen, ``Intelligent reflecting
  surface: Practical phase shift model and beamforming optimization,''
  \emph{IEEE Transactions on Communications}, vol.~68, no.~9, pp. 5849--5863,
  2020.

\bibitem{Long2021Active}
R.~Long, Y.-C. Liang, Y.~Pei, and E.~G. Larsson, ``Active reconfigurable
  intelligent surface-aided wireless communications,'' \emph{IEEE Transactions
  on Wireless Communications}, vol.~20, no.~8, pp. 4962--4975, 2021.

\bibitem{Rao2023Active}
J.~Rao, Y.~Zhang, S.~Tang, Z.~Li, C.-Y. Chiu, and R.~Murch, ``An active
  reconfigurable intelligent surface utilizing phase-reconfigurable reflection
  amplifiers,'' \emph{IEEE Transactions on Microwave Theory and Techniques},
  vol.~71, no.~7, pp. 3189--3202, 2023.

\bibitem{Peng2022Multi}
Z.~Peng, X.~Liu, C.~Pan, L.~Li, and J.~Wang, ``Multi-pair {D2D} communications
  aided by an active {RIS} over spatially correlated channels with phase
  noise,'' \emph{IEEE Wireless Communications Letters}, vol.~11, no.~10, pp.
  2090--2094, 2022.

\bibitem{Saab2022Optimizing}
S.~Saab, A.~Mezghani, and R.~W. Heath, ``Optimizing the mutual information of
  frequency-selective multi-port antenna arrays in the presence of mutual
  coupling,'' \emph{IEEE Transactions on Communications}, vol.~70, no.~3, pp.
  2072--2084, 2022.

\bibitem{Di2022Communication}
M.~Di~Renzo, F.~H. Danufane, and S.~Tretyakov, ``Communication models for
  reconfigurable intelligent surfaces: From surface electromagnetics to
  wireless networks optimization,'' \emph{Proceedings of the IEEE}, vol. 110,
  no.~9, pp. 1164--1209, 2022.

\bibitem{Mitikiri2019Acceleration}
Y.~Mitikiri and K.~Mohseni, ``Acceleration compensation for gravity sense using
  an accelerometer in an aerodynamically stable {UAV},'' in \emph{IEEE
  Conference on Decision and Control}, 2019, pp. 1177--1182.

\bibitem{Cichocki2015Tensor}
A.~Cichocki, D.~Mandic, L.~De~Lathauwer, G.~Zhou, Q.~Zhao, C.~Caiafa, and H.~A.
  PHAN, ``Tensor decompositions for signal processing applications: From
  two-way to multiway component analysis,'' \emph{IEEE Signal Processing
  Magazine}, vol.~32, no.~2, pp. 145--163, 2015.

\bibitem{Sidiropoulos2017Tensor}
N.~D. Sidiropoulos, L.~De~Lathauwer, X.~Fu, K.~Huang, E.~E. Papalexakis, and
  C.~Faloutsos, ``Tensor decomposition for signal processing and machine
  learning,'' \emph{IEEE Transactions on Signal Processing}, vol.~65, no.~13,
  pp. 3551--3582, 2017.

\bibitem{Haardt2008Higher}
M.~Haardt, F.~Roemer, and G.~Del~Galdo, ``Higher-order {SVD-Based} subspace
  estimation to improve the parameter estimation accuracy in multidimensional
  harmonic retrieval problems,'' \emph{IEEE Transactions on Signal Processing},
  vol.~56, no.~7, pp. 3198--3213, 2008.

\bibitem{Roy1989Esprit}
R.~Roy and T.~Kailath, ``{ESPRIT-estimation} of signal parameters via
  rotational invariance techniques,'' \emph{IEEE Transactions on acoustics,
  speech, and signal processing}, vol.~37, no.~7, pp. 984--995, 1989.

\bibitem{Zhang2021Gridless}
J.~Zhang, D.~Rakhimov, and M.~Haardt, ``Gridless channel estimation for hybrid
  {mmWave} {MIMO} systems via {Tensor-ESPRIT} algorithms in {DFT} beamspace,''
  \emph{IEEE Journal of Selected Topics in Signal Processing}, vol.~15, no.~3,
  pp. 816--831, 2021.

\bibitem{Jiang2022Joint}
S.~Jiang, N.~Fu, Z.~Wei, X.~Li, L.~Qiao, and X.~Peng, ``Joint spectrum,
  carrier, and {DOA} estimation with beamforming {MWC} sampling system,''
  \emph{IEEE Transactions on Instrumentation and Measurement}, vol.~71, pp.
  1--15, 2022.

\bibitem{Fu2020Short}
N.~Fu, Z.~Wei, L.~Qiao, and Z.~Yan, ``Short-observation measurement of multiple
  sinusoids with multichannel sub-{Nyquist} sampling,'' \emph{IEEE Transactions
  on Instrumentation and Measurement}, vol.~69, no.~9, pp. 6853--6869, 2020.

\bibitem{Kolda2009Tensor}
T.~G. Kolda and B.~W. Bader, ``Tensor decompositions and applications,''
  \emph{SIAM review}, vol.~51, no.~3, pp. 455--500, 2009.

\bibitem{Wen2020Tensor}
F.~Wen, H.~C. So, and H.~Wymeersch, ``Tensor decomposition-based beamspace
  {ESPRIT} algorithm for multidimensional harmonic retrieval,'' in \emph{IEEE
  International Conference on Acoustics, Speech, \& Signal Processing
  (ICASSP)}, 2020, pp. 4572--4576.

\bibitem{Jiang2021Beamspace}
F.~Jiang, F.~Wen, Y.~Ge, M.~Zhu, H.~Wymeersch, and F.~Tufvesson, ``Beamspace
  multidimensional {ESPRIT} approaches for simultaneous localization and
  communications,'' \emph{preprint arXiv:2111.07450}, 2021.

\bibitem{Wen2018Tensor}
F.~Wen, N.~Garcia, J.~Kulmer, K.~Witrisal, and H.~Wymeersch, ``Tensor
  decomposition based beamspace {ESPRIT} for millimeter wave {MIMO} channel
  estimation,'' in \emph{IEEE Global Communications Conference (GLOBECOM)},
  2018.

\bibitem{Stoica2005Spectral}
P.~Stoica, R.~L. Moses \emph{et~al.}, \emph{Spectral analysis of
  signals}.\hskip 1em plus 0.5em minus 0.4em\relax Pearson Prentice Hall Upper
  Saddle River, NJ, 2005, vol. 452.

\bibitem{kay1993fundamentals}
S.~M. Kay, \emph{Fundamentals of statistical signal processing: estimation
  theory}.\hskip 1em plus 0.5em minus 0.4em\relax Prentice-Hall, Inc., 1993.

\bibitem{Tarboush2023Compressive}
S.~Tarboush, A.~Ali, and T.~Y. Al-Naffouri, ``Compressive estimation of near
  field channels for ultra massive-{MIMO} wideband {THz} systems,'' in
  \emph{IEEE International Conference on Acoustics, Speech, \& Signal
  Processing (ICASSP)}, 2023.

\bibitem{tarboush2023cross}
------, ``Cross-field channel estimation for ultra massive-{MIMO} {THz}
  systems,'' \emph{preprint arXiv:2305.13757}, 2023.

\bibitem{Candes2006Near}
E.~J. Candes and T.~Tao, ``Near-optimal signal recovery from random
  projections: Universal encoding strategies?'' \emph{IEEE Transactions on
  Information Theory}, vol.~52, no.~12, pp. 5406--5425, 2006.

\bibitem{manopt}
\BIBentryALTinterwordspacing
N.~Boumal, B.~Mishra, P.-A. Absil, and R.~Sepulchre, ``{M}anopt, a {M}atlab
  toolbox for optimization on manifolds,'' \emph{Journal of Machine Learning
  Research}, vol.~15, no.~42, pp. 1455--1459, 2014. [Online]. Available:
  \url{https://www.manopt.org}
\BIBentrySTDinterwordspacing

\bibitem{Athley2005Threshold}
F.~Athley, ``Threshold region performance of maximum likelihood direction of
  arrival estimators,'' \emph{IEEE Transactions on Signal Processing}, vol.~53,
  no.~4, pp. 1359--1373, 2005.

\bibitem{Shen2022Modeling}
S.~Shen, B.~Clerckx, and R.~Murch, ``Modeling and architecture design of
  reconfigurable intelligent surfaces using scattering parameter network
  analysis,'' \emph{IEEE Transactions on Wireless Communications}, vol.~21,
  no.~2, pp. 1229--1243, 2022.

\bibitem{Wijekoon2023Beamforming}
D.~Wijekoon, A.~Mezghani, and E.~Hossain, ``Beamforming optimization in
  {RIS}-aided {MIMO} systems under multiple-reflection effects,'' in \emph{IEEE
  International Conference on Acoustics, Speech, \& Signal Processing
  (ICASSP)}, 2023.

\bibitem{Pozar2011Microwave}
D.~M. Pozar, \emph{Microwave engineering}.\hskip 1em plus 0.5em minus
  0.4em\relax John wiley \& sons, 2011.

\bibitem{Di2023Modeling}
M.~Di~Renzo, V.~Galdi, and G.~Castaldi, ``Modeling the mutual coupling of
  reconfigurable metasurfaces,'' in \emph{17th European Conference on Antennas
  and Propagation}, 2023.

\bibitem{Zheng2023Impact}
P.~Zheng, X.~Ma, and T.~Y. Al-Naffouri, ``On the impact of mutual coupling on
  {RIS}-assisted channel estimation,'' \emph{preprint arXiv:2309.04990}, 2023.

\bibitem{Qian2021Mutual}
X.~Qian and M.~D. Renzo, ``Mutual coupling and unit cell aware optimization for
  reconfigurable intelligent surfaces,'' \emph{IEEE Wireless Communications
  Letters}, vol.~10, no.~6, pp. 1183--1187, 2021.

\end{thebibliography}
\bibliographystyle{IEEEtran}

\end{document}